\providecommand{\journal}{0}

\documentclass[a4paper, twoside, reqno]{amsart}

\usepackage{amsthm}
\usepackage{amssymb}
\usepackage[foot]{amsaddr}
\usepackage[
    backend=biber,
    natbib=true,
    citestyle=authoryear-comp,
    bibstyle=authoryear,
    maxcitenames=2,
    maxbibnames=9,
    uniquelist=false,
    uniquename=false,
    dashed=false,
    firstinits=true,
    doi=false,
    isbn=false,
    date=year]{biblatex}
\usepackage[
    hidelinks,
    colorlinks=true,
    allcolors=blue,
    pdfproducer={Latex with hyperref},
    pdfcreator={pdflatex}]{hyperref}
\usepackage[font=footnotesize]{caption}
\usepackage{graphicx}
\usepackage{float}
\usepackage{subcaption}
\usepackage{marvosym}
\usepackage{textcase}
\usepackage{mathtools}
\usepackage{glossaries-prefix}
\usepackage{bbm} %
\usepackage{enumitem}
\usepackage{upgreek}
\usepackage{booktabs}
\usepackage{multicol}
\usepackage{multirow}

\newcommand{\papertitle}{Mode-based estimation of the center of symmetry}

\newcommand{\myorcid}[1]{%
    \href{https://orcid.org/#1}{%
        \includegraphics[height=2ex]{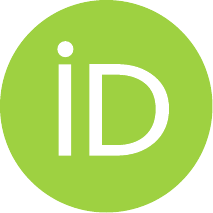}%
        \url{https://orcid.org/#1}
    }%
}

\newcommand{\joseechacon}{Jos\'e E. Chac\'on}

\ifnum\journal=0
    \newcommand{\addressjoseechaon}{Departamento de Matem\'aticas, Universidad de Extremadura, Badajoz, Spain.}
\fi

\newcommand{\emailjoseechacon}{jechacon@unex.es}
\newcommand{\orcidjoseechacon}{0000-0002-3675-1960}
\newcommand{\footjoseechacon}{$^{\dagger}$}
\newcommand{\footaddressjoseechaon}{\footjoseechacon\addressjoseechaon}
\newcommand{\footorcidjoseechacon}{\footjoseechacon\myorcid{\orcidjoseechacon}}

\newcommand{\javierfdezserrano}{Javier Fern\'andez Serrano}

\ifnum\journal=0
    \newcommand{\addressjavierfdezserrano}{Departamento de Matem\'aticas, Universidad Aut\'onoma de Madrid, Madrid, Spain.}
\fi

\newcommand{\emailjavierfdezserrano}{javier.fernandezs01@estudiante.uam.es}
\newcommand{\orcidjavierfdezserrano}{0000-0001-5270-9941}
\newcommand{\footjavierfdezserrano}{$^{\ddagger}$}
\newcommand{\footaddressjavierfdezserrano}{\footjavierfdezserrano\addressjavierfdezserrano}
\newcommand{\footorcidjavierfdezserrano}{\footjavierfdezserrano\myorcid{\orcidjavierfdezserrano}}

\newcommand{\keywordkme}{kernel mode estimator}
\newcommand{\keywordcenterofsymmetry}{center of symmetry}
\newcommand{\keywordunimodality}{unimodality}
\newcommand{\keywordregularlyvarying}{regularly varying density}
\newcommand{\keywordredescendingmestimator}{redescending M-estimator}
\newcommand{\keywordefficientnonparametrics}{efficient nonparametric estimation}

\newcommand{\mscnonparamestimation}{62G05}
\newcommand{\mscdensityestimation}{62G07}
\newcommand{\mscnonparamrobustness}{62G35}

\newcommand{\smallabsvalof}[1]{\vert #1 \vert}
\newcommand{\absvalof}[1]{\left\vert #1 \right\vert}

\newcommand{\derivativeof}[1]{#1'}
\newcommand{\secondderivativeof}[1]{#1''}
\newcommand{\derivativewrtx}{\frac{\differential}{\dx}}

\newcommand{\indicator}[2]{\mathbbm{1}_{#1}(#2)}
\newcommand{\convolutionbetween}[2]{{#1} * {#2}}

\newcommand{\reals}{\mathbb{R}}
\newcommand{\naturals}{\mathbb{N}}

\newcommand{\x}{x}
\newcommand{\y}{y}
\newcommand{\genindexi}{i}
\newcommand{\genindexj}{j}
\newcommand{\constant}{C}
\newcommand{\genericexponent}{a}

\newcommand{\differential}{d}
\newcommand{\dx}{\differential\x}
\newcommand{\dy}{\differential\y}

\newcommand{\intoverreals}{\int_{-\infty}^{\infty}}
\newcommand{\intoverpositivereals}{\int_0^{\infty}}
\newcommand{\intoverunitinterval}{\int_0^1}

\newcommand{\randomvarx}{X}
\newcommand{\randomvary}{Y}
\newcommand{\samplerandomvar}{\randomvarx}
\newcommand{\samplerandomvarindex}[1]{\samplerandomvar_{#1}}
\newcommand{\randomvarobsindex}[1]{x_{#1}}
\newcommand{\ithsamplevar}{\samplerandomvarindex{\genindexi}}
\newcommand{\samplesize}{n}
\newcommand{\sqrtsamplesize}{\sqrt{\samplesize}}
\newcommand{\sumoverallsample}{\sum_{\genindexi = 1}^{\samplesize}}
\newcommand{\shiftedithsamplevar}{\x - \ithsamplevar}
\newcommand{\randomsample}{\samplerandomvarindex{1}, \dots, \samplerandomvarindex{\samplesize}}

\newcommand{\kernel}{K}
\newcommand{\kernelatx}{\kernel(\x)}
\newcommand{\bandwidth}{h}
\newcommand{\squarebandwidth}{\bandwidth^2}
\newcommand{\cubedbandwidth}{\bandwidth^3}
\newcommand{\kernelh}{\kernel_{\bandwidth}}
\newcommand{\bandwidthupperbound}{\bandwidth_{\mathrm{max}}}
\newcommand{\optmialbandwidth}{\bandwidth^{\star}}
\newcommand{\bandwidthsecondary}{g}
\newcommand{\samplesizebandwidth}{\samplesize\bandwidth}
\newcommand{\estimateof}[1]{\hat{#1}}
\newcommand{\kernelestimateof}[2]{\estimateof{#1}_{\samplesize, #2}}
\newcommand{\kde}{\kernelestimateof{\pdf}{\bandwidth}}
\newcommand{\kdesecondary}{\kernelestimateof{\pdf}{\bandwidthsecondary}}
\newcommand{\kdeat}[1]{\kde(#1)}
\newcommand{\kdeatx}{\kdeat{\x}}

\newcommand{\epanechnikovkernel}{\kernel_{\mathrm{Ep}}}
\newcommand{\indicatoroversupportatx}{\indicator{(-1, 1)}{\x}}

\newcommand{\derivativeofkernel}{\derivativeof{\kernel}}
\newcommand{\derivativeofkernelh}{\derivativeof{\kernelh}}
\newcommand{\secondderivativeofkernel}{\secondderivativeof{\kernel}}
\newcommand{\secondderivativeofkernelh}{\secondderivativeof{\kernelh}}

\newcommand{\followsdistr}{\sim}
\newcommand{\samedistribution}{\equiv_{d}}
\newcommand{\weaklyconverges}[2]{{#1} \leadsto {#2}}

\newcommand{\expectedvalueof}[1]{\mathbb{E}[#1]}
\newcommand{\varianceof}[1]{\mathrm{Var}(#1)}

\newcommand{\normaldistribution}[2]{\mathcal{N}(#1, #2)}

\newcommand{\goesto}{\rightarrow}
\newcommand{\goestoinfty}{\goesto \infty}
\newcommand{\goestopositivezero}{\goesto 0^+}

\newcommand{\samplesizegoestoinfty}{\samplesize \goestoinfty}
\newcommand{\bandwidthgoestoinfty}{\bandwidth \goestoinfty}
\newcommand{\bandwidthgoestopositivezero}{\bandwidth \goestopositivezero}
\newcommand{\nugoestoinfty}{\studenttdegreesoffreedom \goestoinfty}

\newcommand{\limxinf}{\lim_{\x \goestoinfty}}
\newcommand{\limhinf}{\lim_{\bandwidthgoestoinfty}}
\newcommand{\limhzero}{\lim_{\bandwidth \goesto 0^+}}

\newcommand{\asymptoticallyequivalent}{\sim}

\DeclareMathOperator*{\argmax}{arg\,max}
\DeclareMathOperator*{\argmin}{arg\,min}

\newcommand{\xinreals}{\x \in \reals}
\newcommand{\argmaxinrealsof}[1]{%
    \argmax_{\xinreals}{#1}
}

\newcommand{\functionmap}[2]{#1 \rightarrow #2}
\newcommand{\lambdafunction}[2]{#1 \mapsto #2}
\newcommand{\functiondef}[3]{#1: \functionmap{#2}{#3}}

\newcommand{\function}{f}
\newcommand{\genericfunction}{\varphi}
\newcommand{\pdf}{f}
\newcommand{\derivativeofpdf}{\derivativeof{\pdf}}
\newcommand{\secondderivativeofpdf}{\secondderivativeof{\pdf}}
\newcommand{\cdf}{F}

\newcommand{\mestimatorloss}{\rho}
\newcommand{\mestimatorlossderivative}{\psi}
\newcommand{\mestimatorparamgeneric}{\theta}
\newcommand{\mestimatorparamf}{\mestimatorparamgeneric_{\cdf}}
\newcommand{\mestimatorparam}{\mestimatorparamgeneric}
\newcommand{\mestimator}{\estimateof{\mestimatorparamgeneric}_{\samplesize}}

\newcommand{\mean}{\mu}
\newcommand{\median}{M}

\newcommand{\sampleestimatorof}[1]{\estimateof{#1}_{\samplesize}}
\newcommand{\samplemeanestimator}{\sampleestimatorof{\mean}}
\newcommand{\samplemedianestimator}{\sampleestimatorof{\median}}

\newcommand{\ithsamplevarminusmestimator}{\ithsamplevar - \mestimator}

\newcommand{\symmetrycenter}{\theta}
\newcommand{\inflectionpoint}{a}
\newcommand{\centeredvar}[1]{#1 - \symmetrycenter}
\newcommand{\centeredsamplerandomvar}{\centeredvar{\samplerandomvar}}
\newcommand{\centeredx}{\centeredvar{\x}}
\newcommand{\centeredkme}{\samplemode - \symmetrycenter}

\newcommand{\untranslatedpdf}{{\pdf}_0}
\newcommand{\mode}{\mathsf{m}}
\newcommand{\samplemode}{\kernelestimateof{\mode}{\bandwidth}}
\newcommand{\ithsamplevarminussamplemode}{\ithsamplevar - \samplemode}

\newcommand{\smoothingof}[1]{\bar{#1}_{\bandwidth}}
\newcommand{\smoothedpdf}{\smoothingof{\pdf}}
\newcommand{\smoothedmode}{\smoothingof{\mode}}

\newcommand{\stddev}{\sigma}
\newcommand{\squarestddev}{\stddev^2}
\newcommand{\pdfkerneldependence}[1]{{#1}_{\pdf, \kernel}}
\newcommand{\stddevmode}{\pdfkerneldependence{\stddev}}
\newcommand{\squarestddevmode}{\stddevmode^2}
\newcommand{\squarestddevmodeexpanded}{%
    \abscurvatureatsymmetrycenter^{-2}
    \pdfatsymmetrycenter
    \kernelroughness
}
\newcommand{\variancefunction}{\mathcal{V}}

\newcommand{\samplesizebandwidththree}{\samplesize \cubedbandwidth}
\newcommand{\sqrtsamplesizebandwidththree}{\sqrt{\samplesizebandwidththree}}

\newcommand{\kernelderivativesquareltwonorm}{\intoverreals \derivativeofkernel(\x)^2 \ \dx}

\newcommand{\kernelroughness}{\mathcal{R}(\derivativeofkernel)}

\newcommand{\squaredsmoothedstddevmodeof}[1]{\stddev_{\pdf, #1, \bandwidth}^2}
\newcommand{\squaredsmoothedstddevmode}{\squaredsmoothedstddevmodeof{\kernel}}

\newcommand{\convolutionvariancenumerator}{[\convolutionbetween{(\derivativeofkernelh)^2}{\pdf}](\symmetrycenter)}

\newcommand{\convolutionvariancedenominator}{(\convolutionbetween{\secondderivativeofkernelh}{\pdf})(\symmetrycenter)}

\newcommand{\convolutionvariance}{%
    \frac
    {\convolutionvariancenumerator}
    {\convolutionvariancedenominator^2}
}

\newcommand{\pdfatsymmetrycenter}{\pdf(\symmetrycenter)}

\newcommand{\convolvedpdfsecondderivative}{(\convolutionbetween{\kernelh}{\secondderivativeof{\pdf}})(\symmetrycenter)}

\newcommand{\curvatureatsymmetrycenter}{\secondderivativeof{\pdf}(\symmetrycenter)}
\newcommand{\abscurvatureatsymmetrycenter}{\absvalof{\curvatureatsymmetrycenter}}

\newcommand{\standardgaussianpdf}{\phi}

\newcommand{\pdfatx}{\pdf(\x)}
\newcommand{\smoothedpdfatx}{\smoothedpdf(\x)}
\newcommand{\pdfatmuplushx}{\pdf(\symmetrycenter + \bandwidth \x)}
\newcommand{\pdfatmuminushx}{\pdf(\symmetrycenter - \bandwidth \x)}
\newcommand{\pdfatmuplusx}{\pdf(\symmetrycenter + \x)}
\newcommand{\pdfatmplusx}{\pdf(\mode + \x)}
\newcommand{\pdfatmuminusx}{\pdf(\symmetrycenter - \x)}
\newcommand{\pdfatmuplush}{\pdf(\symmetrycenter + \bandwidth)}
\newcommand{\inversebandwidth}{\bandwidth^{-1}}

\newcommand{\targetvariance}{\stddev_{\pdf, \chaconserranokernelparam, \bandwidth}^2}
\newcommand{\standardrandomvar}{Z_{\bandwidth}}
\newcommand{\symmetricfunction}{\Psi}
\newcommand{\ithsymmetricfunction}[2]{
    \symmetricfunction_{#1}(#2; \chaconserranokernelparam)
}
\newcommand{\ithsymmetricexpectation}[1]{%
    \expectedvalueof{\ithsymmetricfunction{#1}{\standardrandomvar}}
}
\newcommand{\derivativeofwrt}[2]{\partial #1 / \partial #2}

\newcommand{\integralsquarekernelprimenotation}{I_{\kernel, \regularvariationindex}}
\newcommand{\integralsquarekernelprimenotationbis}{\bar{I}_{\kernel, \regularvariationindex}}

\newcommand{\negativelimit}{2 \slowlyvaryinglimit \integralsquarekernelprimenotation}

\newcommand{\chaconserranokernelparam}{\beta}
\newcommand{\chaconserranokernel}{%
    \kernel_{\chaconserranokernelparam}
}
\newcommand{\bumplikefunction}{%
    B_{\chaconserranokernelparam}
}

\newcommand{\derivativeofchaconserranokernel}{\derivativeof{\chaconserranokernel}}
\newcommand{\secondderivativeofchaconserranokernel}{\secondderivativeof{\chaconserranokernel}}

\newcommand{\weightfunction}{W}
\newcommand{\weightfunctionh}{\weightfunction_{\bandwidth}}
\newcommand{\scorefunction}{%
    \mathcal{Z}_{\chaconserranokernelparam}
}

\newcommand{\samplemodekthiteration}{k}
\newcommand{\kthsamplemode}[1]{\estimateof{\mode}_{#1}}
\newcommand{\previousitersamplemode}{\kthsamplemode{\samplemodekthiteration}}
\newcommand{\currentitersamplemode}{\kthsamplemode{\samplemodekthiteration + 1}}
\newcommand{\weightvector}[1]{\mathbf{w}^{(#1)}}
\newcommand{\currentweightvector}{\weightvector{\samplemodekthiteration}}
\newcommand{\softmax}{\upsigma}
\newcommand{\scorevector}{\mathbf{z}}
\newcommand{\ithsoftmaxscorevector}[1]{\softmax(\scorevector)_{#1}}

\newcommand{\samplemodeinitialguess}{\kthsamplemode{0}}

\newcommand{\chaconserranokernelpsix}{%
    \psi(\x; \chaconserranokernelparam)
}

\newcommand{\optimalparameter}[1]{{#1}^*}
\newcommand{\optimalchaconserranokernelparam}{\optimalparameter{\chaconserranokernelparam}}
\newcommand{\optimalbandwidth}{\optimalparameter{\bandwidth}}

\newcommand{\auxiliaryfunctionindex}{\eta}
\newcommand{\estimateofuntranslatedpdf}{\tilde{\untranslatedpdf}}
\newcommand{\symmetrickde}{\kdesecondary^*}

\newcommand{\parameters}{(\chaconserranokernelparam, \bandwidth)}
\newcommand{\optimalparameters}{%
    (\optimalchaconserranokernelparam, \optimalbandwidth)
}

\newcommand{\regularvariationindex}{\alpha}
\newcommand{\regularvariationconstant}{\lambda}
\newcommand{\regularvariationdomainlow}{x_0}
\newcommand{\studenttdegreesoffreedom}{\nu}
\newcommand{\slowlyvaryingcomponent}{L}
\newcommand{\slowlyvaryinglimit}{\ell}
\newcommand{\xtoregularvariationindex}{\x^{\regularvariationindex}}
\newcommand{\regularvariationindexplusthree}{\regularvariationindex + 3}

\newcommand{\meanshiftaverage}{\varpi_{\samplesize, \bandwidth}}
\newcommand{\kthmeanshiftaverage}[1]{\meanshiftaverage^{#1}}
\newcommand{\irwtolerance}{\varepsilon}

\newcommand{\randomerror}{\varepsilon}

\newcommand{\locationmodel}{\samplerandomvar = \symmetrycenter + \randomerror}

\newcommand{\trimlevel}{\upalpha}

\newcommand{\madn}{\mathrm{MADN}(\randomsample)}
\newcommand{\madnvar}{S}

\newcommand{\numrepetitions}{m}
\newcommand{\methodestimateofsymmetrycenter}{{\estimateof{\symmetrycenter}}^{\method}}
\newcommand{\methodoneestimateofsymmetrycenter}{{\estimateof{\symmetrycenter}}^{\methodone}}
\newcommand{\methodtwoestimateofsymmetrycenter}{{\estimateof{\symmetrycenter}}^{\methodtwo}}
\newcommand{\method}{\texttt{M}}
\newcommand{\methodone}{\method_1}
\newcommand{\methodtwo}{\method_2}
\newcommand{\mserrorof}[1]{\mathrm{\acrshort*{mse}}(#1)}
\newcommand{\mseofmethod}{\mserrorof{\method}}

\newcommand{\pdfsupportedge}{a}
\newcommand{\kernelsupportedge}{b}
\newcommand{\taylorintermediatevalue}{\xi_{\x}}
\newcommand{\cutpoint}{\x_0}

\newcommand{\exponentialminusone}{%
    \exp [-\genericfunction(\x)] - 1
}

\newcommand{\absexponentialminusone}{%
    \absvalof{\exponentialminusone}
}

\newcommand{\finiteintegral}{%
    \intoverunitinterval
    \x^{\genericexponent}
    \absexponentialminusone
    \ \dx
}

\newcommand{\integrablefunction}{%
    \x^{\genericexponent} \genericfunction(\x)
}

\newcommand{\asymptoticsmoothedvariance}{\tilde{\stddev}_{\pdf, \kernel, \bandwidth}^2}
\newcommand{\smoothedvariancedenominator}{R_{\pdf, \kernel, \bandwidth}}

\newcommand{\untranslatedsmoothedpdf}{\smoothedpdf(\symmetrycenter + \x)}

\newcommand{\derivativesmoothedpdf}{%
    \derivativewrtx
    \untranslatedsmoothedpdf
}

\input{01-3-configuration.tex}

\newacronym{mse}{MSE}{mean squared error}

\newacronym[
    plural=pdfs,
    firstplural=probability density functions]{pdf}{pdf}
{probability density function}

\newacronym[
    plural=KDEs,
    firstplural=kernel density estimators]{kde}{KDE}
{kernel density estimator}

\newacronym[
    plural=KMEs,
    firstplural=kernel mode estimators]{kme}{KME}
{\textit{kernel mode estimator}}

\newacronym[
    plural=MLEs,
    firstplural=maximum likelihood estimators]{mle}{MLE}
{maximum likelihood estimator}

\newacronym{iid}{i.i.d.}{independent and identically distributed}

\newacronym{irw}{IRW}{\textit{iterative reweighting}}

\begin{document}

\title{\papertitle}
\date{}

\hypersetup{
  pdftitle={\papertitle},
}

\author{\joseechacon\footjoseechacon}
\address{\footaddressjoseechaon}
\email{\footjoseechacon\emailjoseechacon \ \Letter}
\thanks{\footorcidjoseechacon}

\author{\javierfdezserrano\footjavierfdezserrano}
\address{\footaddressjavierfdezserrano}
\email{\footjavierfdezserrano\emailjavierfdezserrano}
\thanks{\footorcidjavierfdezserrano}

\hypersetup{
    pdfauthor={\joseechacon, \javierfdezserrano}
}

\newcommand{\abstractmeta}{%
    In the mean-median-mode triad of univariate centrality measures, the mode has been overlooked for estimating the center of symmetry in continuous and unimodal settings.
    This paper expands on the connection between kernel mode estimators and M-estimators for location, bridging the gap between the nonparametrics and robust statistics communities.
    The variance of modal estimators is studied in terms of a bandwidth parameter, establishing conditions for an optimal solution that outperforms the household sample mean.
    A purely nonparametric approach is adopted, modeling heavy-tailedness through regular variation.
    The results lead to an estimator proposal that includes a novel one-parameter family of kernels with compact support, offering extra robustness and efficiency.
    The effectiveness and versatility of the new method are demonstrated in a real-world case study and a thorough simulation study, comparing favorably to traditional and more competitive alternatives.
    Several myths about the mode are clarified along the way, reopening the quest for flexible and efficient nonparametric estimators.
}

\begin{abstract}
    \abstractmeta
\end{abstract}

\hypersetup {
    pdfsubject={\abstractmeta},
}


\ifnum\journal=0
    \subjclass[2020]{
        \mscnonparamestimation \ (Primary),
        \mscdensityestimation,
        \mscnonparamrobustness
    }
\fi

\hypersetup{
    pdfkeywords={%
            \keywordkme,
            \keywordcenterofsymmetry,
            \keywordunimodality,
            \keywordregularlyvarying,
            \keywordredescendingmestimator,
            \keywordefficientnonparametrics
        },
}

\maketitle

\section{Introduction}

The mean-median-mode triad is well-rooted in statistical \textit{folklore}.
In general, these three summaries depict different notions of centrality and have different properties, so the relative preference for either of them is dependent on the context.
Even so, the mode has certainly received less attention in the literature, perhaps because it poses a more challenging estimation problem and also due to definition technicalities in the continuous setting.
In any case, the concept of mode has recently emerged as a solution to several seemingly unrelated statistical problems.
See~\citet{Chacon2020} for a comprehensive overview.

It is well known that these three centrality measures coincide when assuming symmetry and unimodality, in which case they all represent the so-called \textit{center of symmetry} of the distribution.
Indeed, under symmetry, the mean and the median have long been studied as location statistics, and several attempts have been made to find robust intermediate estimators~\citep{Huber1964}.
Nevertheless, although the mean and the median are known to underperform in fat-tailed scenarios~\citep{Lai1983}, the mode has been overlooked for estimating the center of symmetry.

On the other hand, there is abundant literature on nonparametric estimation of the mode~\citep{Chacon2020}, particularly in kernel density estimation~\citep{Parzen1962}.
Very accurate results describing the asymptotic properties of the kernel-based mode estimator can be found in the exhaustive studies of~\citet{Romano1986,Romano1988,Romano1988a}.
However, symmetry has relevant implications in the analysis of this estimator, which have not been exploited before.
Interestingly, symmetry makes the bias of the kernel-based mode estimator vanish~\citetext{\citealp[p. 41]{Chernoff1964};~\citealp[Section 4]{Eddy1982}}, allowing us to focus on reducing variance.

This paper aims to bridge the gap between robust location statistics and kernel density estimation theory for mode estimation in the symmetric, unimodal case.
Theoretical asymptotic results and finite-sample simulations will demonstrate the effectiveness of the modal approach.
Our findings also expand on the connection noted by~\citet[Section 4]{Eddy1982} with the theory of M-estimators~\citetext{\citealp{Huber1964,Huber2009};~\citealp[p. 484]{Lehmann1998}}, which elegantly subsumes the three classic centrality points of view.

The rest of the paper is organized as follows.
The links between the two communities (robust statistics and nonparametrics) are briefly reviewed in Section~\ref{sec:background}, where the fundamental notation is also introduced.
Section~\ref{sec:method} presents original theoretical results and a novel proposal for mode-based estimation of the center of symmetry.
Proofs are given in Appendix~\ref{sec:proofs}.
The practical effectiveness of our method is demonstrated through a case study with real-world data in Section~\ref{sec:case-study} and a thorough simulation study in Section~\ref{sec:simulation-study}.
Finally, Section~\ref{sec:discussion} summarizes and discusses the relevance of our research.

\section{Background}\label{sec:background}

Estimating the mode of an unknown univariate \gls*{pdf} $\pdf$ has arguably been a foundational goal for nonparametrics, being the \gls*{kde} one of the most popular tools~\citep{Parzen1962,Grenander1965}.
Given a sample of \gls*{iid} absolutely continuous random variables $\randomsample$ with common \gls*{pdf} $\pdf$, a \textit{bandwidth} $\bandwidth \in (0, \infty)$, and a \textit{kernel} function $\functiondef{\kernel}{\reals}{\reals}$ (typically a symmetric, unimodal \gls*{pdf}), the \gls*{kde} is defined as
\begin{equation}
    \label{eq:kde-definition}
    \kdeatx
    =
    \frac{1}{\samplesize}
    \sumoverallsample
    \kernelh(\shiftedithsamplevar)
    =
    \frac{1}{\samplesizebandwidth}
    \sumoverallsample
    \kernel \left( \frac{\shiftedithsamplevar}{\bandwidth} \right)
    \,.
\end{equation}
The \gls*{kde} converges locally and globally to $\pdf$ as $\samplesizegoestoinfty$ if $\bandwidth \equiv \bandwidth_{\samplesize} \goestopositivezero$ and $\samplesizebandwidth \goestoinfty$, provided certain mild assumptions are satisfied~\citep{Chacon2018}.

In our context, the mode will be the unique maximizer $\mode \in \reals$ of $\pdf$, i.e.,
\begin{equation}
    \label{eq:mode-definition}
    \mode
    =
    \argmaxinrealsof{\pdfatx}
    \,.
\end{equation}
Assuming natural conditions on $\kernel$, there is a random variable $\samplemode$, which~\citet{Parzen1962} calls the \textit{sample mode} and which we shall refer to as the \gls*{kme}, such that
\begin{equation}
    \label{eq:sample-mode-condition}
    \kdeat{\samplemode}
    =
    \max_{\xinreals} \kdeatx
    \,.
\end{equation}
In general, there are several versions of $\samplemode$.
\citet{Eddy1980,Romano1988} pick the smallest argument that maximizes the \gls*{kde} when the optimum is not unique.
\citeauthor{Romano1988} claims all his results hold for any random variable complying with~\eqref{eq:sample-mode-condition}.
In turn,~\citet{Grund1995} directly state their results by \textit{breaking} ties at random.

On the other hand, M-estimators~\citep{Huber1964,Huber2009} of the location parameter $\mestimatorparam \equiv \mestimatorparamf \in \reals$ of a cumulative distribution function $\cdf$ are defined as sample random variables $\mestimator$ minimizing
\begin{equation}
    \label{eq:m-estimator-argmin}
    \sumoverallsample
    \mestimatorloss(\ithsamplevarminusmestimator)
    \,,
\end{equation}
for some fixed function $\functiondef{\mestimatorloss}{\reals}{\reals}$.
Or equivalently, if $\mestimatorloss$ is differentiable with $\mestimatorlossderivative = \derivativeof{\mestimatorloss}$, through
\begin{equation}
    \label{eq:m-estimator-equation-root}
    \sumoverallsample
    \mestimatorlossderivative(\ithsamplevarminusmestimator)
    =
    0
    \,.
\end{equation}
For instance, the sample mean $\samplemeanestimator$ and the sample median $\samplemedianestimator$ are the M-estimators corresponding to $\mestimatorloss(\x) = \x^2$ and $\mestimatorloss(\x) = \absvalof{\x}$, respectively.
Moreover, the population target $\mestimatorparam$---the mean $\mean$ and the median $\median$ in the previous examples---can also be expressed in terms of $\mestimatorloss$ and $\mestimatorlossderivative$ by cautiously replacing sums with expectations~\citep[pp. 46--47]{Huber2009}.

Regarding the mode,~\citet{Eddy1982} mentioned the link between the \gls*{kme} and M-estimators in his concluding remarks, but he did not further develop the idea.
Indeed, given (\ref{eq:kde-definition}) and (\ref{eq:sample-mode-condition}), for a symmetric kernel, the \gls*{kme} $\samplemode$ can be seen as the M-estimator corresponding to $\mestimatorloss(\x) = -\kernelh(\x) = -\kernel(\x / \bandwidth) / \bandwidth$. For fixed $\bandwidth > 0$, the population target of the \gls*{kme} would be the value $\smoothedmode \in \reals$ maximizing  $\smoothedpdfatx = \expectedvalueof{\kdeatx}$.
In general, $\smoothedmode \neq \mode$, and only as $\bandwidthgoestopositivezero$ the bias goes to zero.
However, under the symmetry of $\kernel$ (about zero) and $\pdf$ (about $\mode$), there is no bias effect~\citep{Chernoff1964}, so that $\smoothedmode = \mode$ (this is rigorously stated and proved in Theorem \ref{th:symmetric-unimodal-smoothed-pdf} below).
Hence, the \gls*{kme} would effectively be an M-estimator of the mode.
See in \figurename~\ref{fig:triweight-biweight} the M-estimator perspective of \glspl*{kme}.

\begin{figure}
    \centering
    \includegraphics[width=0.5\textwidth]{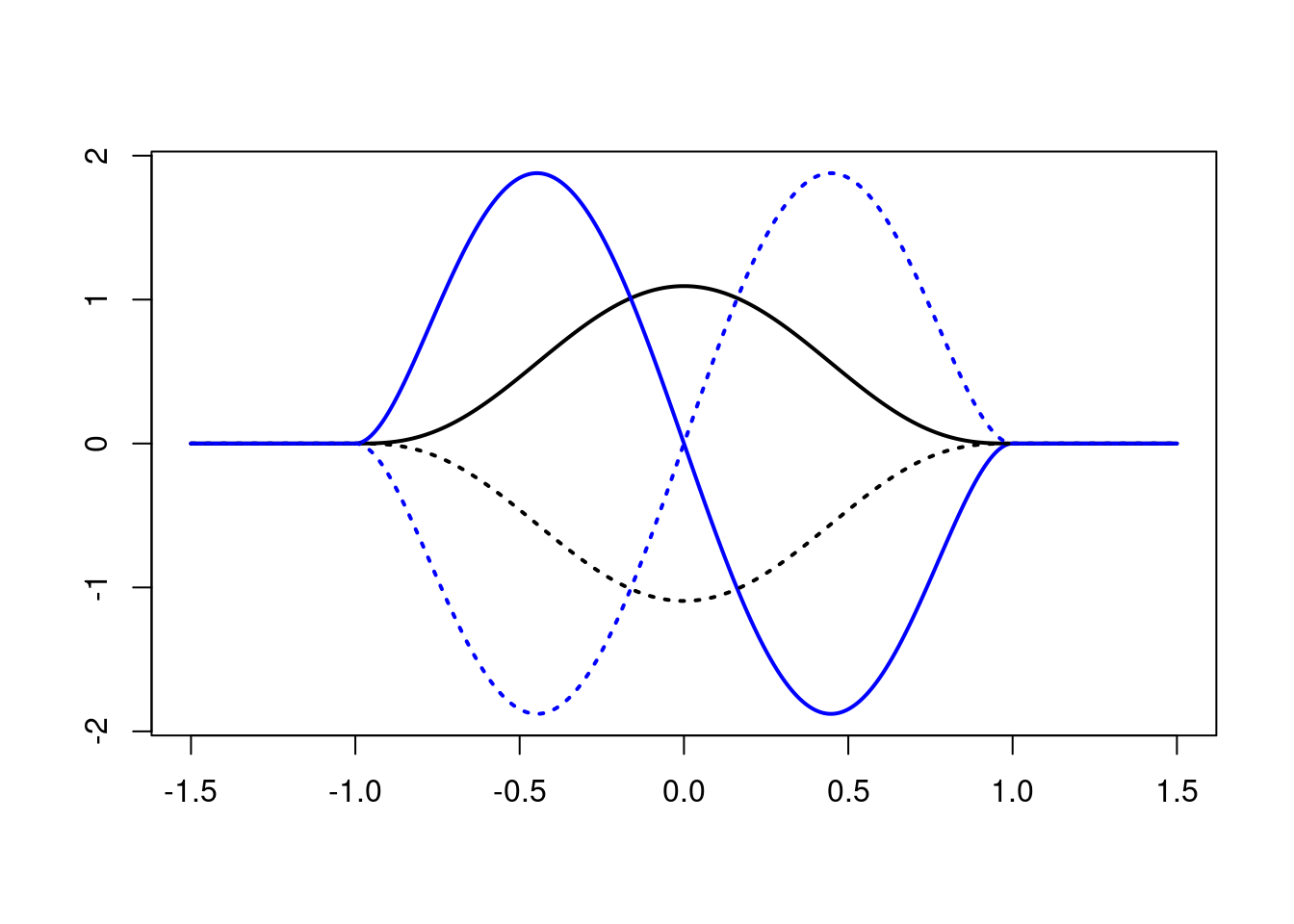}
    \caption{%
        The M-estimator perspective of \glspl*{kme}.
        A kernel function $\kernel$ and its derivative are shown in solid black and blue, respectively.
        The dashed black and blue lines are, respectively, $\mestimatorloss(\x) = -\kernelatx$, and $\mestimatorlossderivative(\x) = \derivativeof{\mestimatorloss}(\x) = -\derivativeofkernel(\x)$.
    }
    \label{fig:triweight-biweight}
\end{figure}

For most common kernels, \glspl*{kme} are members of the class of \textit{redescending} M-estimators, which are those characterized by a function $\mestimatorlossderivative$ that tends to zero as $\absvalof{\x} \goestoinfty$~\citep{Huber2009,Maronna2019}.
Surprisingly, neither of these celebrated robust statistics books mentions \glspl*{kme}.
The importance of redescending M-estimators lies in their resilience to \textit{outliers}, especially when $\mestimatorlossderivative$ has compact support (in the case of \glspl*{kme}, when $\kernel$ has compact support).
See, for instance, the concepts of \textit{influence function} and \textit{breakdown point} in the former references.
Nevertheless, despite their compelling features, redescending M-estimators have also received some criticism.
\citet[p. 101]{Huber2009} point out an efficiency decrease and a higher sensitivity to wrong \textit{scaling}, which in the case of \glspl*{kme} translates into the importance of a good choice of the bandwidth $\bandwidth$.

Consequently, one of the main goals of this paper is to study how the bandwidth affects the performance of the \gls*{kme} in this scenario and propose a practical data-driven choice.

\section{Method}

\label{sec:method}

Section~\ref{sec:theoretical-results} introduces the main theoretical results.
Section~\ref{sec:computation} and Section~\ref{sec:parameter-optimization} address the two steps of our \gls*{kme} proposal.
Finally, Section~\ref{sec:illustrative-example} illustrates our method with a \textit{toy} example.

\subsection{Theoretical results}

\label{sec:theoretical-results}

For a fixed bandwidth $\bandwidth > 0$, the population-wide modal location statistic corresponding to~\eqref{eq:kde-definition} as $\samplesizegoestoinfty$, denoted $\smoothedmode$, is obtained by maximizing the \textit{smoothed \gls*{pdf}} $\smoothedpdf$~\citep[Eq. 2.4]{Wand1995} defined by
\begin{equation}
    \label{eq:smoothed-pdf-definition}
    \smoothedpdfatx
    =
    \expectedvalueof{\kdeatx}
    =
    \intoverreals
    \kernelh(\x - \y)
    \pdf(\y)
    \ \dy
    =
    (\convolutionbetween{\kernelh}{\pdf})(\x)
    \,,
\end{equation}
which does not depend on the sample size $\samplesize$.
If $\kernel$ is a proper \gls*{pdf}, so it is~\eqref{eq:smoothed-pdf-definition}.

As noted in Section \ref{sec:background}, the \textit{smoothed} mode $\smoothedmode$ does not generally coincide with the true mode $\mode$ of $\pdf$.
To change that situation, we need to assume the following two hypotheses.

\begin{definition}[Symmetry]
    An absolutely continuous random variable $\samplerandomvar$ with \gls*{pdf} $\pdf$ is \textit{symmetric} about a center of symmetry $\symmetrycenter \in \reals$ if $\centeredsamplerandomvar$ and $\symmetrycenter - \samplerandomvar$ have the same distribution, or, equivalently, if $\pdfatmuplusx = \pdfatmuminusx$ for all $\xinreals$.
\end{definition}

\begin{definition}[Unimodality]
    \label{def:unimodality}
    An absolutely continuous random variable $\samplerandomvar$ with \gls*{pdf} $\pdf$ is \textit{unimodal} about a mode $\mode \in \reals$ if $\untranslatedpdf(\x) = \pdfatmplusx$ is a strictly decreasing function of $\absvalof{\x}$ over its support.
\end{definition}

The definitions of symmetry and unimodality can be found in p. 22 and p. 228 of~\citet{Maronna2019}, respectively.
In both cases, we shall extend these terms to \glspl*{pdf} rather than only random variables.
In particular, the kernel $\kernel$ in~\eqref{eq:kde-definition} will typically be, in our context, a symmetric, unimodal \gls*{pdf} about zero.
There is a wide array of symmetric, unimodal \glspl*{pdf}, as shown in \figurename~\ref{fig:symmetric-unimodal-test-beds}.
Indeed, symmetry arises as a natural assumption in the location model $\locationmodel$, where $\randomerror$ represents the error variable, to formalize the idea that there are no systematic errors~\citep[p. 17]{Maronna2019}.
Moreover, symmetry and unimodality of the error term appear as key assumptions in modern modal regression~\citep{Wang2024}.

The symmetry assumption implies $\mean = \median = \symmetrycenter$~\citep[p. 22]{Maronna2019}.
However, note that $\mean$ is not always defined, whereas $\median$ is.
This happens, for instance, in the case of the Cauchy distribution, that is, Student's t with one degree of freedom.
If, in addition to symmetry, the maximizer~\eqref{eq:mode-definition} is assumed to be unique, then it is easy to see that $\mode = \symmetrycenter$ as well.
Next, we show that the stronger condition of unimodality ensures that $\smoothedmode = \symmetrycenter$ for every $\bandwidth > 0$, making the \gls*{kme} $\samplemode$ an M-estimator of $\symmetrycenter$, as suggested by~\citet{Chernoff1964,Eddy1982}.

\begin{figure}
    \centering
    \includegraphics[width=0.5\textwidth]{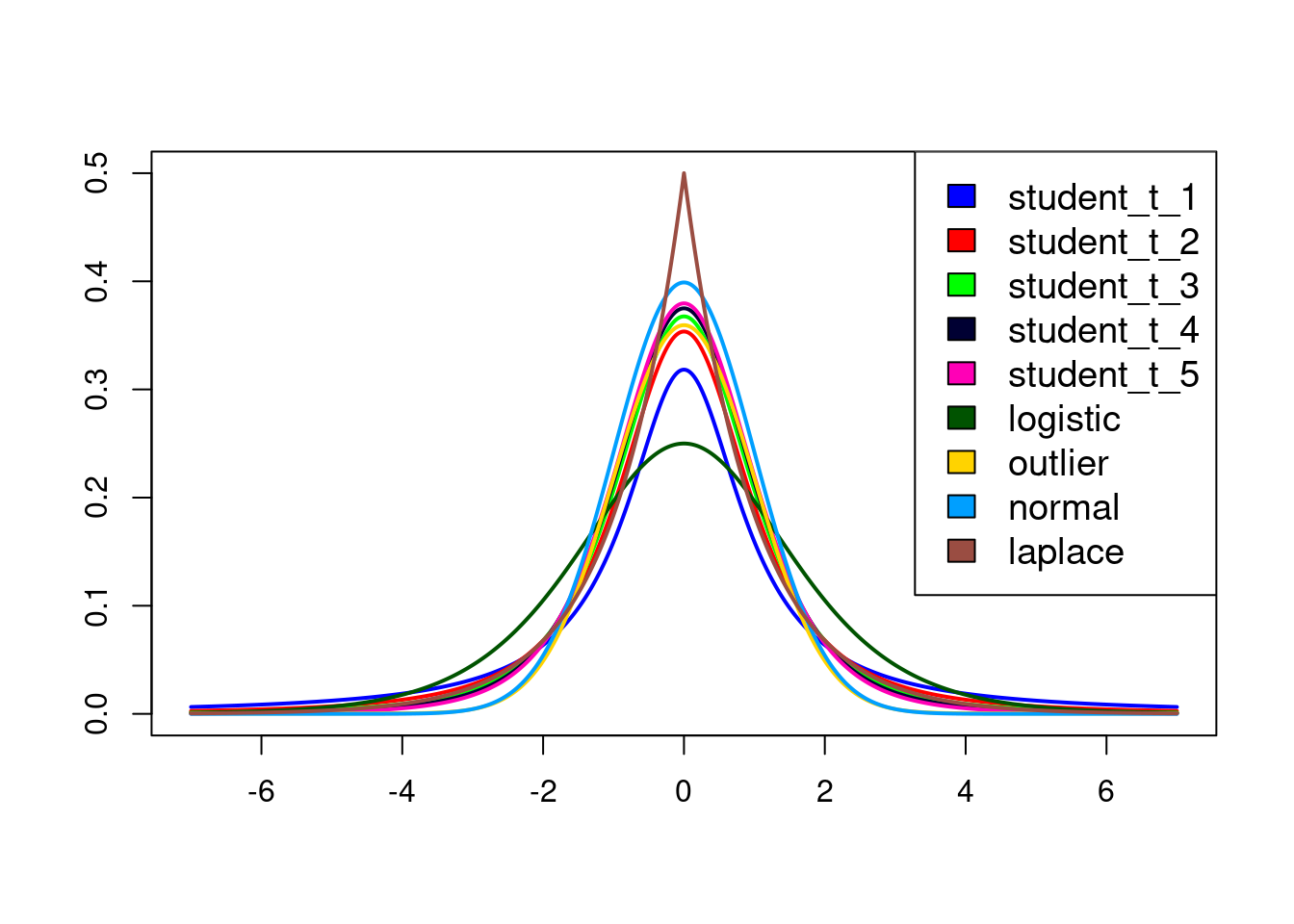}
    \caption{%
        Several examples of symmetric, unimodal \glspl*{pdf} (about zero).
        Instances of the Student's t family~\citep{Maronna2019} with $\studenttdegreesoffreedom \in \{1, ..., 5\}$ appear in dark blue, red, light green, black, and magenta, respectively.
        The \textit{standard} normal, Laplace, and logistic \glspl*{pdf} appear in light blue, brown, and dark green, respectively.
        See~\citet{Maronna2019} for the former two and~\citet[Example 3.13]{Huber2009} for the latter.
        Finally, the \textit{outlier} \gls*{pdf}, shown in yellow, is defined as a Gaussian mixture with two components having the same mean (zero) but very different variances.
        For further details, see Section~\ref{sec:simulation-study} below, where all these \glspl*{pdf} will be used as test-beds in a simulation study.
    }
    \label{fig:symmetric-unimodal-test-beds}
\end{figure}

\begin{theorem}
    \label{th:symmetric-unimodal-smoothed-pdf}
    Assume that a \gls*{pdf} $\pdf$ is symmetric and unimodal about $\symmetrycenter$.
    Also, consider a differentiable \gls*{pdf} kernel $\kernel$ with bounded $\derivativeofkernel$ that is symmetric and unimodal about zero.
    Then, the smoothed \gls*{pdf} $\smoothedpdf$ is also symmetric and unimodal about $\symmetrycenter$ for any $\bandwidth > 0$.
\end{theorem}

The fact that the \textit{convolution} of two symmetric and unimodal distributions---the case of~\eqref{eq:smoothed-pdf-definition}---is also symmetric and unimodal is well known~\citep{Purkayastha1998}.
However, the concept of unimodality employed in that reference slightly differs from the usual one in robust statistics~\citep[see][Theorem 10.2]{Maronna2019}, covering a broader range of distributions but producing a weaker result than required.
Theorem~\ref{th:symmetric-unimodal-smoothed-pdf} is tailored to our needs and has a more straightforward elementary proof.

The role of the bandwidth $\bandwidth$ in Theorem~\ref{th:symmetric-unimodal-smoothed-pdf} seems minor, permanently attached to the kernel in $\kernelh$.
One could think of $\bandwidth$ as incorporated \textit{inside} $\kernel$.
Carrying $\bandwidth$ responds to historical reasons in the kernel smoothing literature.
However, as we shall see, fine-tuning $\bandwidth$ will make the difference in obtaining efficient kernel estimators of the center of symmetry.

Theorem~\ref{th:symmetric-unimodal-smoothed-pdf} ensures that the population-wide target maximizing~\eqref{eq:kde-definition} for a fixed bandwidth $\bandwidth$ is unique and equal to the center of symmetry $\symmetrycenter$.
According to~\citet[Theorem 10.5]{Maronna2019}, the latter guarantees consistency, namely convergence in probability of the \gls*{kme} $\samplemode$ to $\symmetrycenter$ as $\samplesizegoestoinfty$.
The following intuitive result shows that the \gls*{kme} is also \textit{on target} for every finite sample size $\samplesize$.

\begin{theorem}
    \label{th:symmetric-kme}
    In the hypotheses of Theorem~\ref{th:symmetric-unimodal-smoothed-pdf}, if we assume that ties are broken uniformly at random in the condition~\eqref{eq:sample-mode-condition}, then $\samplemode$ is symmetric about $\symmetrycenter$ for every $\samplesize \in \naturals$ and $\bandwidth > 0$.
    Consequently, if $\samplemode$ is integrable, then it is an unbiased estimator of $\symmetrycenter$, i.e., $\expectedvalueof{\samplemode} = \symmetrycenter$.
\end{theorem}

Let us now turn to the asymptotic distribution of the \gls*{kme}.
To make arguments shorter and simpler, we shall work with \textit{bell-shaped} kernels.

\begin{definition}[Bell-shaped kernel]
    \label{def:bell-shaped-kernel}
    Consider a twice continuously differentiable \gls*{pdf} kernel $\kernel$ that is symmetric and unimodal about zero, with bounded $\derivativeofkernel$ and $\secondderivativeofkernel$.
    We say that $\kernel$ is \textit{bell-shaped} if $\kernel$ has exactly two inflection points at $\pm \inflectionpoint$, for some $\inflectionpoint > 0$.
    That is, $\pm \inflectionpoint$ are the only points where $\secondderivativeofkernel$ changes its sign.
\end{definition}

Definition~\ref{def:bell-shaped-kernel} adds further regularity conditions to the ones imposed in Theorem~\ref{th:symmetric-unimodal-smoothed-pdf}.
We require additional smoothness and assume that $\kernel$ is concave over $\absvalof{\x} < \inflectionpoint$ (i.e., $\secondderivativeofkernel(\x) \leq 0$) and convex over $\absvalof{\x} > \inflectionpoint$ (i.e., $\secondderivativeofkernel(\x) \geq 0$).
One can easily see that the convexity and unimodality of $\kernel$ forces $\limxinf \derivativeofkernel(\pm\x) = 0$, making the \gls*{kme} with a bell-shaped kernel a redescending M-estimator.
The most common smooth kernels, such as the Gaussian, are bell-shaped, so Definition~\ref{def:bell-shaped-kernel} is not too restrictive.
Also, note that if $\kernel$ is bell-shaped, so it is $\kernelh$ for any $\bandwidth > 0$.

Using Definition~\ref{def:bell-shaped-kernel}, we can state the following asymptotic normality result.
It can be seen as a particular case of a general asymptotic distribution result for M-estimators ~\citep[see][Theorem 10.7]{Maronna2019}, adapted for \glspl*{kme} in their specific setting.
The notations $\derivativeofkernelh$ and $\secondderivativeofkernelh$ should be interpreted as successive derivatives of $\kernelh$.

\begin{theorem}
    \label{th:sample-mode-asymptotic-normality}
    Assume that $\samplerandomvar \followsdistr \pdf$ is symmetric and unimodal about $\symmetrycenter$.
    Consider a bell-shaped kernel $\kernel$.
    Then, for any $\bandwidth > 0$,
    $
        \weaklyconverges{\sqrtsamplesize(\centeredkme)}{\normaldistribution{0}{\squaredsmoothedstddevmode}}
    $
    as $\samplesizegoestoinfty$, with variance
    \begin{equation}
        \label{eq:sample-mode-asymptotic-variance}
        \squaredsmoothedstddevmode
        =
        \frac
        {\expectedvalueof{\derivativeofkernelh(\centeredsamplerandomvar)^2}}
        {\expectedvalueof{\secondderivativeofkernelh(\centeredsamplerandomvar)}^2}
        =
        \frac
        {\intoverreals \derivativeofkernelh(\centeredx)^2 \pdfatx \ \dx}
        {[\intoverreals \secondderivativeofkernelh(\centeredx) \pdfatx \ \dx]^2}
        =
        \convolutionvariance
        \,.
    \end{equation}
\end{theorem}

Although the result of Theorem \ref{th:sample-mode-asymptotic-normality} is not particularly surprising, its implications are significant in understanding and optimizing the performance of the \gls*{kme}.
The variance~\eqref{eq:sample-mode-asymptotic-variance} has an interesting behavior as $\bandwidthgoestopositivezero$ and $\bandwidthgoestoinfty$ that has not been addressed in the current robust statistics literature.

\begin{theorem}
    \label{th:limit-variance-h}
    In what follows, assume the hypotheses of Theorem~\ref{th:sample-mode-asymptotic-normality}.
    \begin{enumerate}
        \item \label{it:limit-variance-h-tends-to-zero} Let $\pdf$ be twice continuously differentiable with bounded $\derivativeofpdf$ and $\secondderivativeofpdf$ integrable.
              Let $\kernel$ be such that $\limxinf \x \kernel(\x) = 0$ and $\kernelroughness = \kernelderivativesquareltwonorm$ is finite. Then, $\limhzero \squaredsmoothedstddevmode = \infty$.
              Moreover, if $\curvatureatsymmetrycenter < 0$, and if we define $\squarestddevmode = \squarestddevmodeexpanded$, then $\squaredsmoothedstddevmode \asymptoticallyequivalent \squarestddevmode \bandwidth^{-3}$ as $\bandwidthgoestopositivezero$.
        \item \label{it:limit-variance-h-tends-to-infty} If $\pdf$ has finite variance $\squarestddev$, and $\secondderivativeofkernel(0) < 0$, then $\lim_{\bandwidthgoestoinfty} \squaredsmoothedstddevmode = \squarestddev$.
    \end{enumerate}
\end{theorem}

Part~\ref{it:limit-variance-h-tends-to-zero} of Theorem~\ref{th:limit-variance-h} is compatible with the observation by~\citet{Huber2009} that redescending M-estimators are more sensitive to wrong \textit{scaling}, as a too-small $\bandwidth$ dangerously increases the variance.
Part~\ref{it:limit-variance-h-tends-to-infty} of Theorem~\ref{th:limit-variance-h} shows that \glspl*{kme} are asymptotically nearly as efficient as the sample mean if $\bandwidth$ is sufficiently large.

Theorem~\ref{th:sample-mode-asymptotic-normality} is a fixed-bandwidth version of~\citet[Theorem 2.1]{Romano1988a}.
He states, roughly speaking, that
$
    \weaklyconverges
    {\sqrtsamplesizebandwidththree(\samplemode - \smoothedmode)}
    {\normaldistribution{0}{\squarestddevmode}}
$
as $\samplesizegoestoinfty$ and $\bandwidthgoestopositivezero$, where $\squarestddevmode$ is as in Theorem~\ref{th:limit-variance-h} but replacing $\symmetrycenter$ with $\mode$.
The connection between our result and~\citeauthor{Romano1988a}'s derives from the fact that $\smoothedmode = \mode = \symmetrycenter$ under symmetry and unimodality.

Heavy-tailed distributions exist in which the \gls*{kme} is guaranteed to outperform the sample mean.
To model the tail behavior of the distribution we rely on the concept of regularly varying functions~\citep{Seneta1976}.

\begin{definition}[Regularly varying function]
    A function $\functiondef{\function}{[\regularvariationdomainlow, \infty)}{(0, \infty)}$, for $\regularvariationdomainlow > 0$, is \textit{regularly varying} with index $\regularvariationindex \in \reals$ if, for all $\regularvariationconstant > 0$,
    \begin{equation*}
        \limxinf
        \frac
        {\function(\regularvariationconstant \x)}
        {\function(\x)}
        =
        \regularvariationconstant^{\regularvariationindex}
        \,.
    \end{equation*}
    If $\regularvariationindex = 0$, $\function$ is said to be \textit{slowly varying}.
\end{definition}

In symmetric \glspl*{pdf}, the regular variation behavior is the same in both tails and requires $\regularvariationindex < -1$ since otherwise $\pdf$ would not be integrable~\citep[see the assumptions in][Section 9.3, Theorem 2]{Devroye1985}.
An interesting consequence of the regularly varying definition is the universal representation $\function(\x) = \xtoregularvariationindex \slowlyvaryingcomponent(\x)$, where $\slowlyvaryingcomponent$ is a slowly varying function~\citep{Seneta1976}.
The most paradigmatic regularly varying distributions belong to Student's t family~\citep[Eq. 2.9]{Maronna2019}, with \gls*{pdf}
\begin{equation}
    \label{eq:student-t-pdf}
    \pdfatx
    \propto
    \left(
    1 + \frac{\x^2}{\studenttdegreesoffreedom}
    \right)^{-(\studenttdegreesoffreedom + 1) / 2}
    \,,
\end{equation}
parameterized by the \textit{degrees of freedom} $\studenttdegreesoffreedom > 0$, corresponding to the regular variation index $\regularvariationindex = -(\studenttdegreesoffreedom + 1)$.
The limit of~\eqref{eq:student-t-pdf} as $\nugoestoinfty$ is the standard normal \gls*{pdf}, which is not regularly varying.

Theorem~\ref{th:regular-variation} below specifies conditions under which regularly varying \glspl*{pdf}, particularly Student's t \glspl*{pdf}, have a \gls*{kme} that is more efficient than the sample mean.

\begin{theorem}
    \label{th:regular-variation}
    Assume that a symmetric and unimodal \gls*{pdf} $\pdf$, with center $\symmetrycenter$, is such that $\untranslatedpdf(\x) = \pdfatmuplusx$ is regularly varying with index $\regularvariationindex < -3$ and bounded slowly varying component $\slowlyvaryingcomponent$ satisfying $\limhinf \slowlyvaryingcomponent(\bandwidth) = \slowlyvaryinglimit > 0$.
    Suppose that a kernel $\kernel$ satisfies the hypotheses of Theorem~\ref{th:limit-variance-h}, part~\ref{it:limit-variance-h-tends-to-infty}, and the following two integral conditions: (i)
    $
        \integralsquarekernelprimenotation
        =
        \intoverpositivereals
        \xtoregularvariationindex
        [\absvalof{\secondderivativeofkernel(0)}^{-2} \derivativeofkernel(\x)^2 - \x^2]
        \
        \dx
    $ is finite and negative, and (ii) $
        \intoverpositivereals
        \xtoregularvariationindex
        [1 + \absvalof{\secondderivativeofkernel(0)}^{-1} \secondderivativeofkernel(\x)]
        \
        \dx
    $ is finite.
    Then, $\limhinf \bandwidth^{-(\regularvariationindexplusthree)} (\squaredsmoothedstddevmode - \squarestddev) = \negativelimit$.
    Consequently, for every sufficiently large $\bandwidth$, we have $\squaredsmoothedstddevmode < \squarestddev$.
\end{theorem}

\begin{remark}
    \label{th:remark-integral-conditions}
    In Theorem~\ref{th:regular-variation}, if $\kernel$ has compact support $[-1, 1]$, we have:
    \begin{enumerate}[label=(\roman*)]
        \item For the first integral condition, if $
                  \integralsquarekernelprimenotationbis
                  =
                  \intoverunitinterval
                  \xtoregularvariationindex
                  [\absvalof{\secondderivativeofkernel(0)}^{-2} \derivativeofkernel(\x)^2 - \x^2]
                  \
                  \dx
              $,
              then $
                  \integralsquarekernelprimenotation
                  =
                  \integralsquarekernelprimenotationbis
                  +
                  (\regularvariationindexplusthree)^{-1}
              $.
              Consequently, it suffices to check that $\integralsquarekernelprimenotationbis$ is finite and $\integralsquarekernelprimenotationbis < -(\regularvariationindexplusthree)^{-1}$.
              Since $\regularvariationindex < -3$ by hypothesis, the latter inequality is satisfied if $\integralsquarekernelprimenotationbis < 0$.
        \item For the second integral condition, since $\regularvariationindex < -3$ by hypothesis, it suffices to check that $
                  \intoverunitinterval
                  \xtoregularvariationindex
                  [1 + \absvalof{\secondderivativeofkernel(0)}^{-1}\secondderivativeofkernel(\x)]
                  \
                  \dx
              $ is finite.
    \end{enumerate}
\end{remark}

Theorem~\ref{th:limit-variance-h} and Theorem~\ref{th:regular-variation} together allow searching for an $\bandwidth$ that minimizes the variance $\squaredsmoothedstddevmode$, following a similar scheme to that of~\citet{Chacon2007}.

\begin{corollary}
    \label{th:corollary-optimal-h}
    Assume all the hypotheses in Theorem~\ref{th:limit-variance-h} and Theorem~\ref{th:regular-variation}.
    Define $\functiondef{\variancefunction}{(0, \infty)}{(0, \infty)}$ given by $\variancefunction(\bandwidth) = \squaredsmoothedstddevmode$.
    There exists $\optmialbandwidth \in (0, \infty)$ such that $\variancefunction(\optmialbandwidth) = \inf_{\bandwidth \in (0, \infty)} \variancefunction(\bandwidth) < \squarestddev$.
\end{corollary}

\begin{remark}
    Of course, if $\squarestddev = \infty$, then any \gls*{kme} is more efficient than the sample mean, regardless of $\bandwidth$.
\end{remark}

Finding a kernel that satisfies the hypotheses of Theorem~\ref{th:regular-variation} is not trivial, especially one covering every $\regularvariationindex < -3$.
The Epanechnikov kernel, defined as $\epanechnikovkernel(\x) = 3(1 - \x^2) / 4 \cdot \indicatoroversupportatx$~\citep[see][p. 875]{Eddy1980}, accomplishes the two integral conditions, but it is not differentiable.
Though differentiable, other polynomial kernels, such as the \textit{triweight} ~\citep[or Tukey's \textit{biweight} $\mestimatorlossderivative$, in the M-estimator parlance of][Eq. 4.92]{Huber2009}, only cover a limited range of $\regularvariationindex$.
In turn, the household Gaussian kernel is bell-shaped but only satisfies the first integral condition.
We present now a parametric family of kernels compatible with the assumptions of Theorem~\ref{th:regular-variation} and whose parameter can be tuned to reach every $\regularvariationindex$.

\begin{theorem}
    \label{th:chacon-serrano-kernel}
    Define the family of bump-like functions $\bumplikefunction$, with $\chaconserranokernelparam > 0$,
    \begin{equation}
        \label{eq:bump-like-function}
        \bumplikefunction(\x)
        =
        \begin{cases}
            e^{-1 / (1 - \absvalof{\x}^{\chaconserranokernelparam})}
            \,,
            \  & \mathrm{if} \
            \absvalof{\x} < 1
            \\
            0
            \,,
            \  & \mathrm{otherwise}
        \end{cases}
        \,.
    \end{equation}
    Then, consider the family of \gls*{pdf} kernels $\chaconserranokernel$ whose derivative is given by
    \begin{equation}
        \label{eq:chacon-serrano-kernel-derivative}
        \derivativeofchaconserranokernel(\x)
        \propto
        -\x \bumplikefunction(\x)
        \,.
    \end{equation}
    If $\absvalof{\regularvariationindex + 1} < \chaconserranokernelparam$, then the kernel $\chaconserranokernel$ satisfies the hypotheses of Corollary~\ref{th:corollary-optimal-h}.
\end{theorem}

See in \figurename~\ref{fig:chacon-serrano-kernel-family} several instances of $\chaconserranokernel$ showcasing the flexibility of the new family.
The inspiration for~\eqref{eq:chacon-serrano-kernel-derivative} comes from the redescending M-estimator perspective, as we derive $\mestimatorloss$ from $\mestimatorlossderivative$ and not vice versa.
Essentially, we multiply $-\x$ by~\eqref{eq:bump-like-function} to obtain a smooth version of $\derivativeof{\epanechnikovkernel}$, which, except for the points where $\epanechnikovkernel$ is not differentiable ($\pm 1$), corresponds to $\mestimatorlossderivative(\x) \propto \x \cdot \indicatoroversupportatx$.
The latter produces the most efficient redescending M-estimator for a \gls*{pdf} $\pdf$ satisfying $-\derivativeof{\untranslatedpdf}(\x) / \untranslatedpdf(\x) \propto \x$ among those $\mestimatorlossderivative$ with support $[-1, 1]$~\citep[Eq. 4.88]{Huber2009}.
One can easily see that such $\pdf$ is either a Gaussian or a truncated Gaussian.

\begin{figure}
    \centering
    \begin{subfigure}[b]{\figurewidth}
        \centering
        \includegraphics[width=\subfigurewidth]{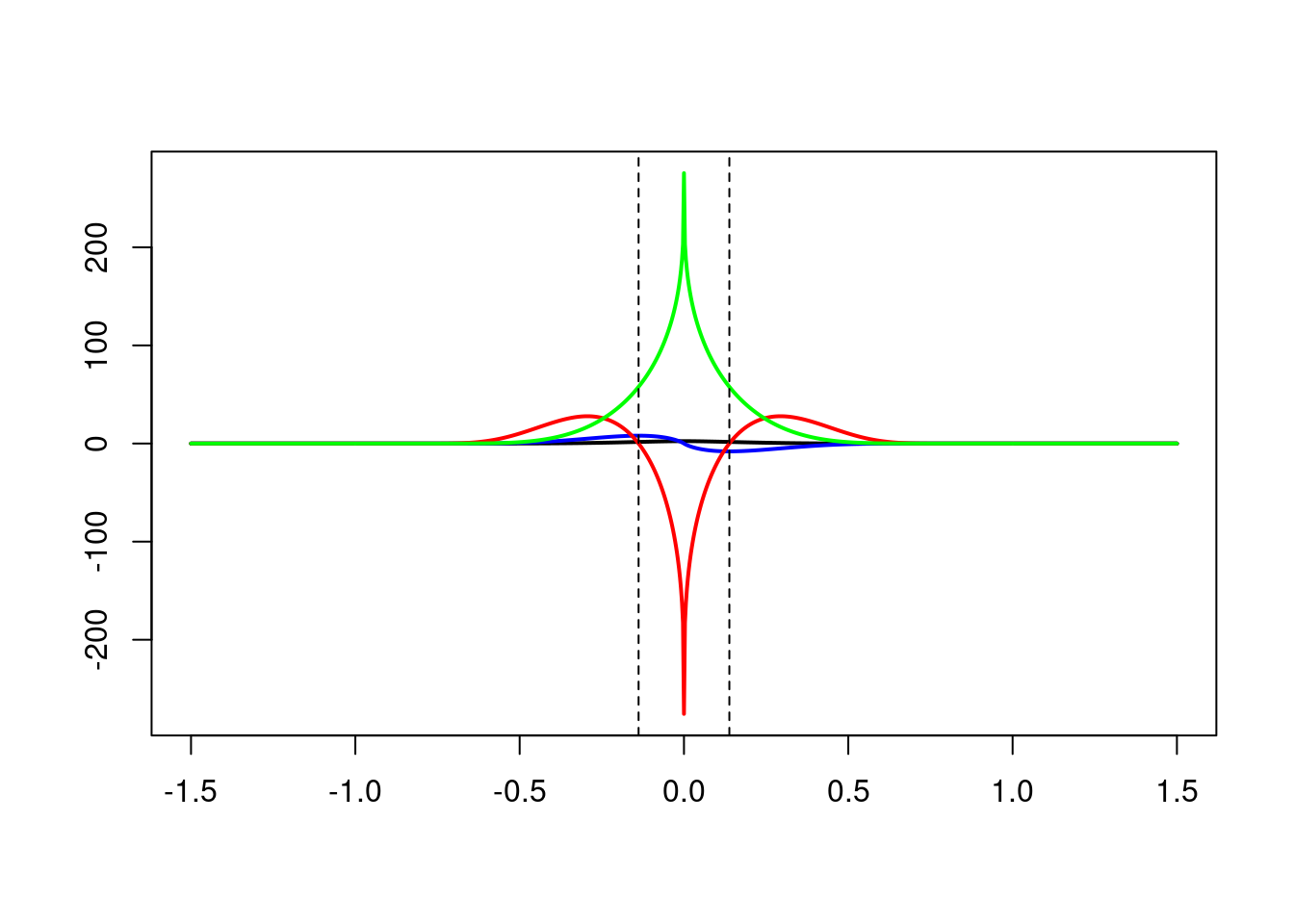}
        \caption{$\chaconserranokernelparam = 1/4$}
    \end{subfigure}
    \begin{subfigure}[b]{\figurewidth}
        \centering
        \includegraphics[width=\subfigurewidth]{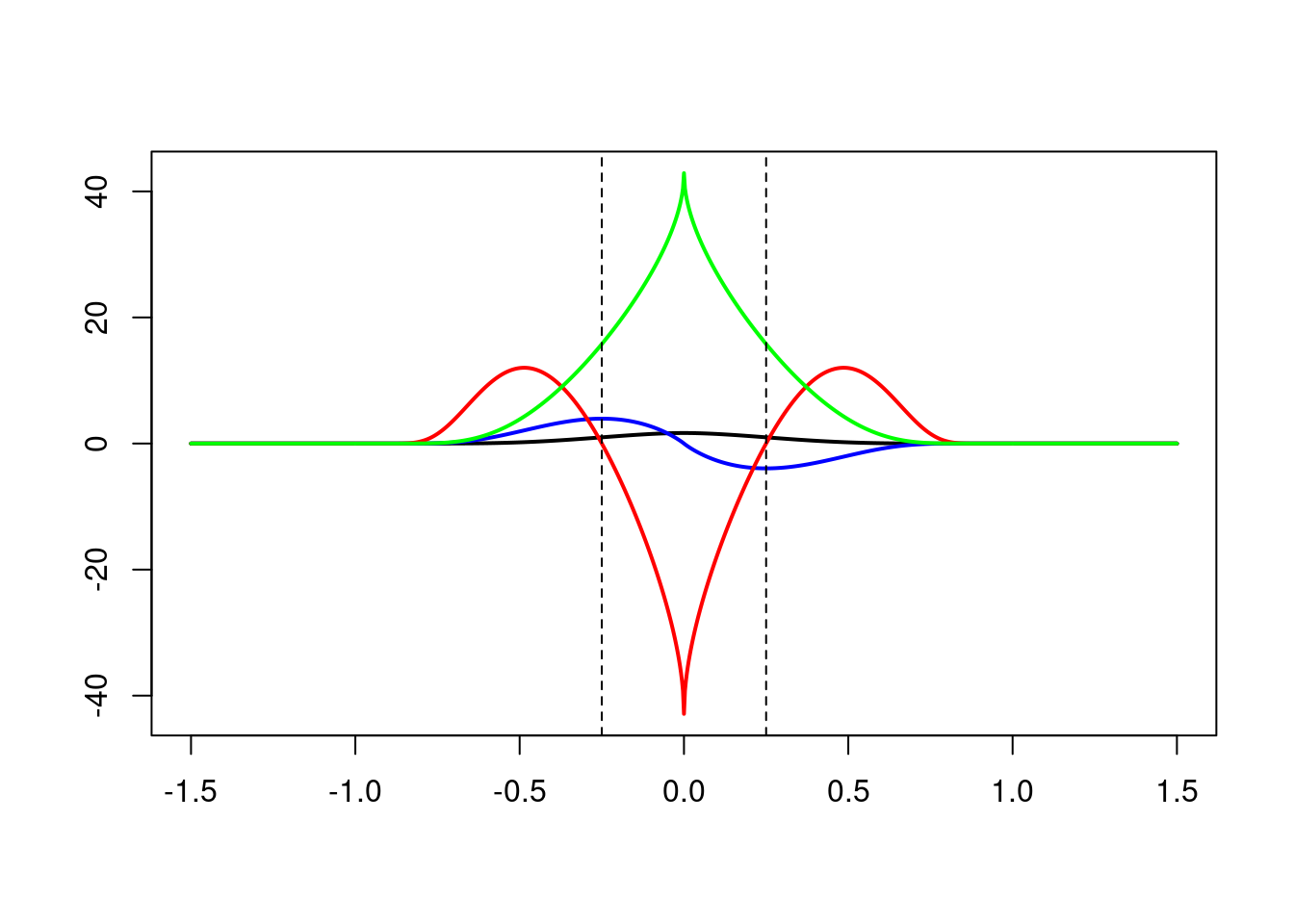}
        \caption{$\chaconserranokernelparam = 1/2$}
    \end{subfigure}
    \begin{subfigure}[b]{\figurewidth}
        \centering
        \includegraphics[width=\subfigurewidth]{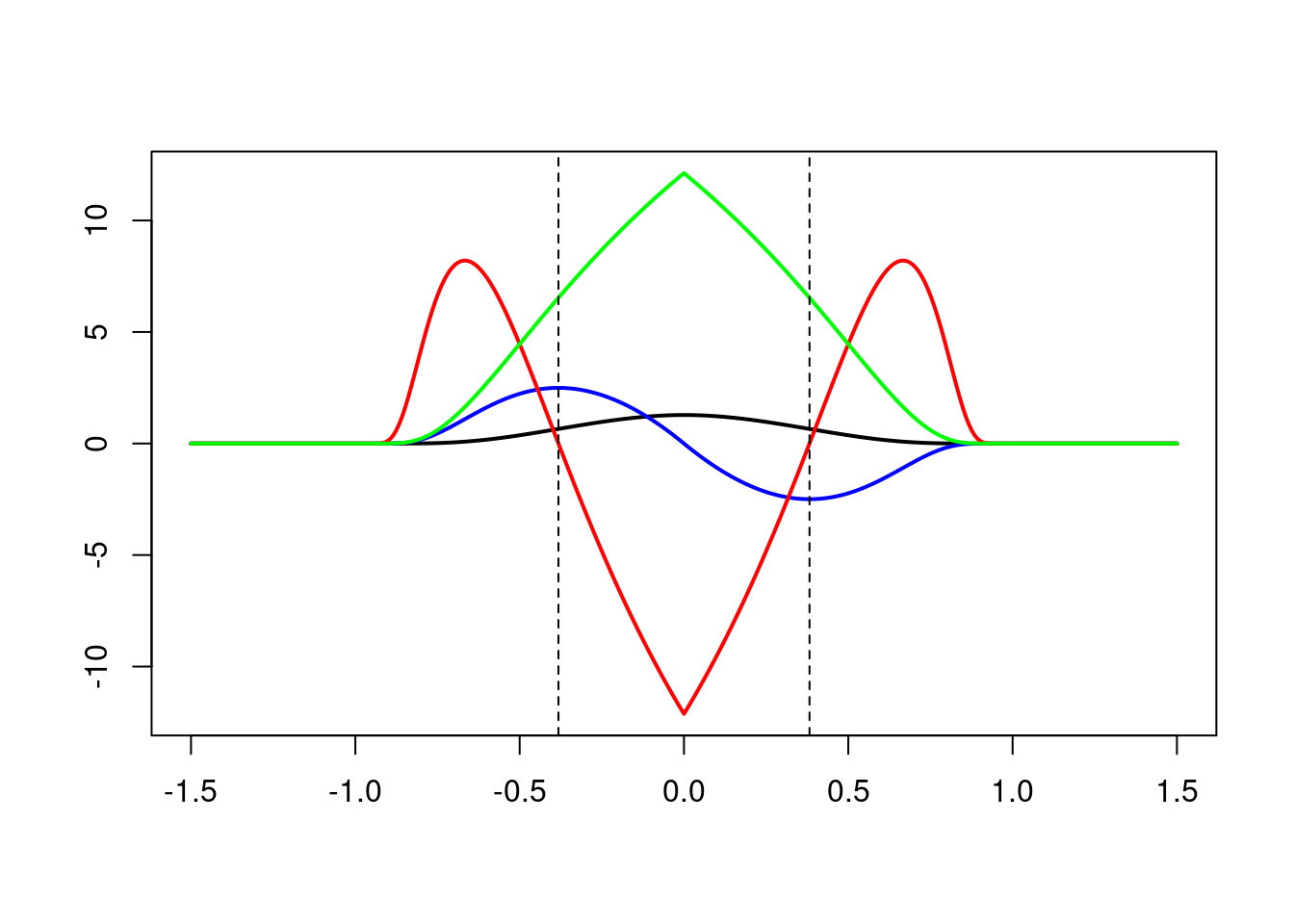}
        \caption{$\chaconserranokernelparam = 1$}
    \end{subfigure}
    \begin{subfigure}[b]{\figurewidth}
        \centering
        \includegraphics[width=\subfigurewidth]{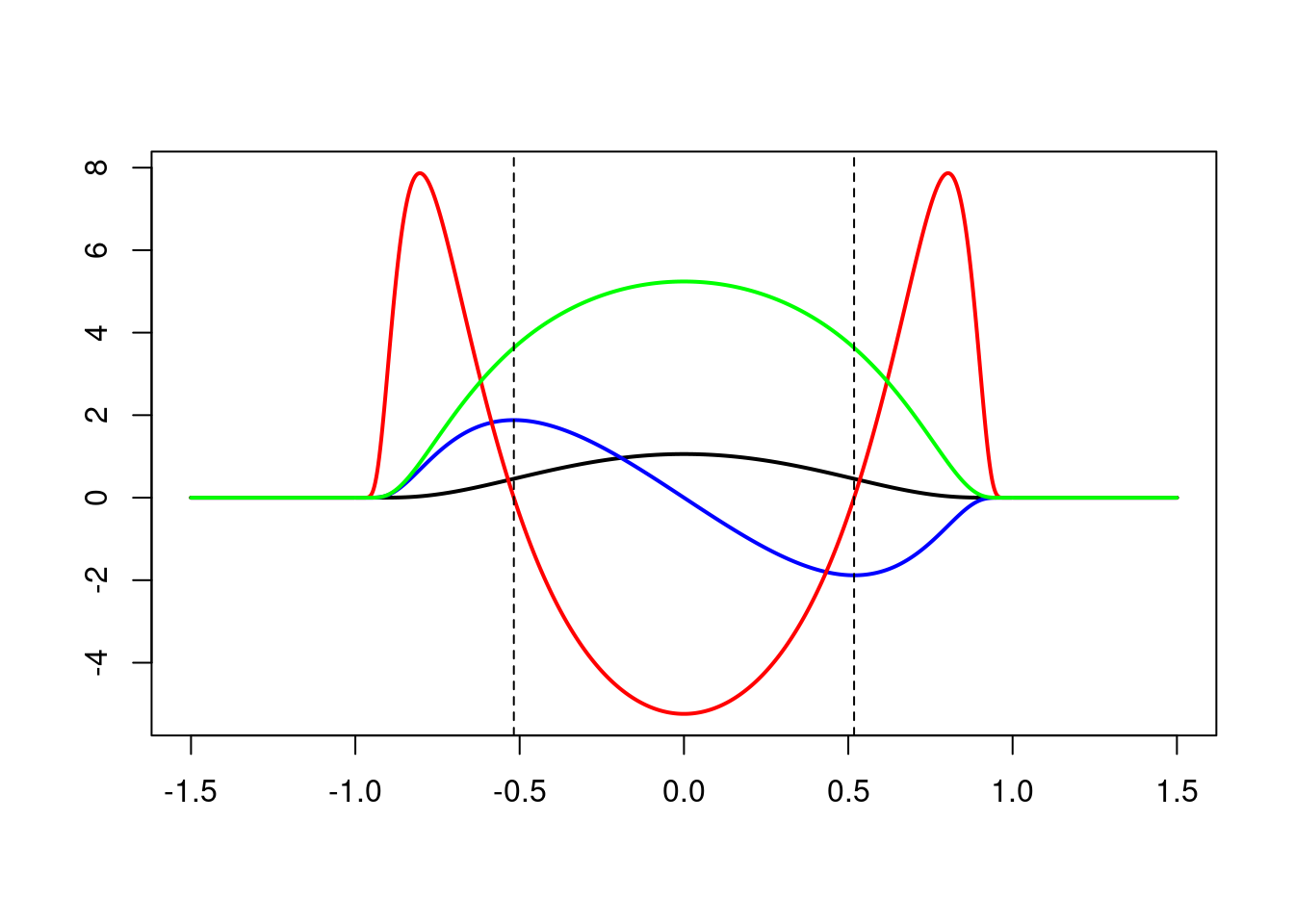}
        \caption{$\chaconserranokernelparam = 2$}
    \end{subfigure}
    \begin{subfigure}[b]{\figurewidth}
        \centering
        \includegraphics[width=\subfigurewidth]{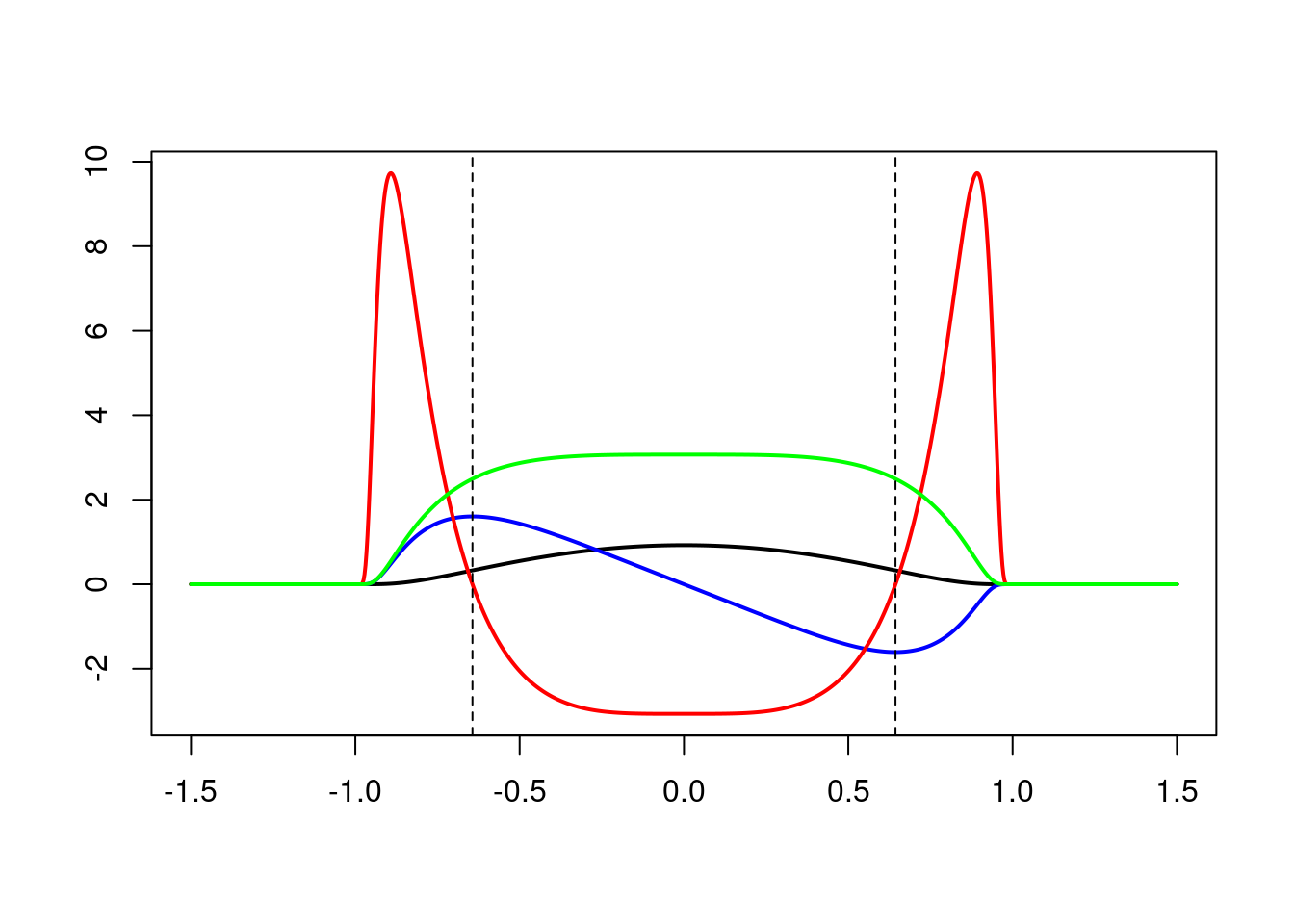}
        \caption{$\chaconserranokernelparam = 4$}
    \end{subfigure}
    \begin{subfigure}[b]{\figurewidth}
        \centering
        \includegraphics[width=\subfigurewidth]{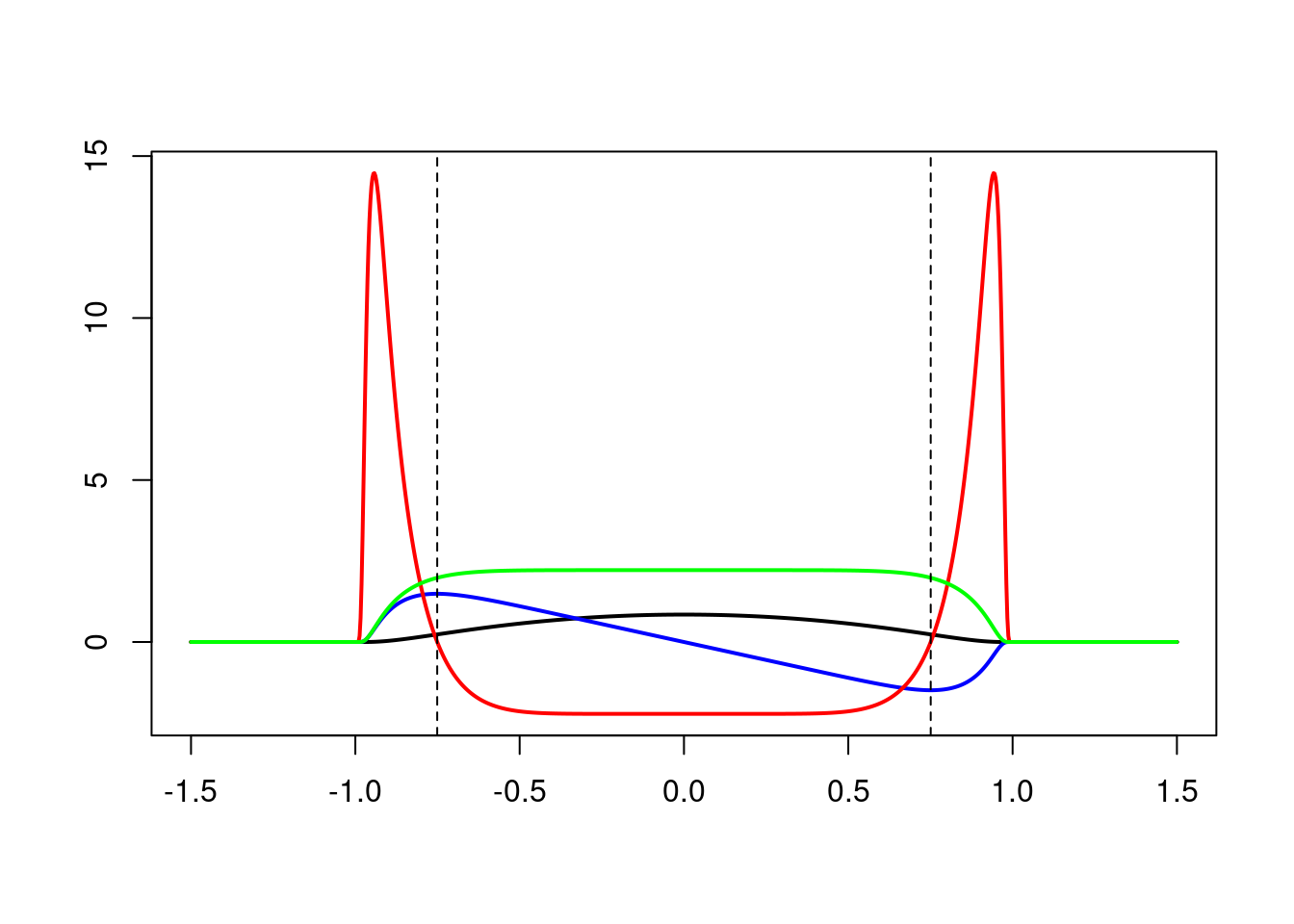}
        \caption{$\chaconserranokernelparam = 8$}
    \end{subfigure}
    \caption{%
        Several instances of the parametric kernel family $\chaconserranokernel$ defined in Theorem~\ref{th:chacon-serrano-kernel}, along with some derived associated functions.
        The black, blue, and red lines are $\chaconserranokernel$, $\derivativeofchaconserranokernel$, and $\secondderivativeofchaconserranokernel$, respectively.
        The green line corresponds to the weight function~\eqref{eq:weight-function}, which is equal to the $\bumplikefunction$~\eqref{eq:bump-like-function} in the case of $\chaconserranokernel$.
        The inflection points of $\chaconserranokernel$ appear as vertical dashed lines.
        As $\chaconserranokernelparam$ grows, the inflection points diverge towards the edges of the support, while $\derivativeofchaconserranokernel$ and $\bumplikefunction$ approach an oblique straight line and a rectangular function, respectively.
        Consequently, $\secondderivativeofchaconserranokernel$ \textit{redescends} increasingly more sharply.
        In turn, as $\chaconserranokernelparam$ decreases to zero, the inflection points converge towards the center, making the slope of $\derivativeofchaconserranokernel$ increasingly \textit{steep} near zero.
        Indeed, for $\chaconserranokernelparam \leq 1$, we see that $\secondderivativeofchaconserranokernel$ and $\bumplikefunction$ are not differentiable at zero.
    }
    \label{fig:chacon-serrano-kernel-family}
\end{figure}

\begin{figure}
    \centering
    \begin{subfigure}[b]{\figurewidth}
        \centering
        \includegraphics[width=\subfigurewidth]{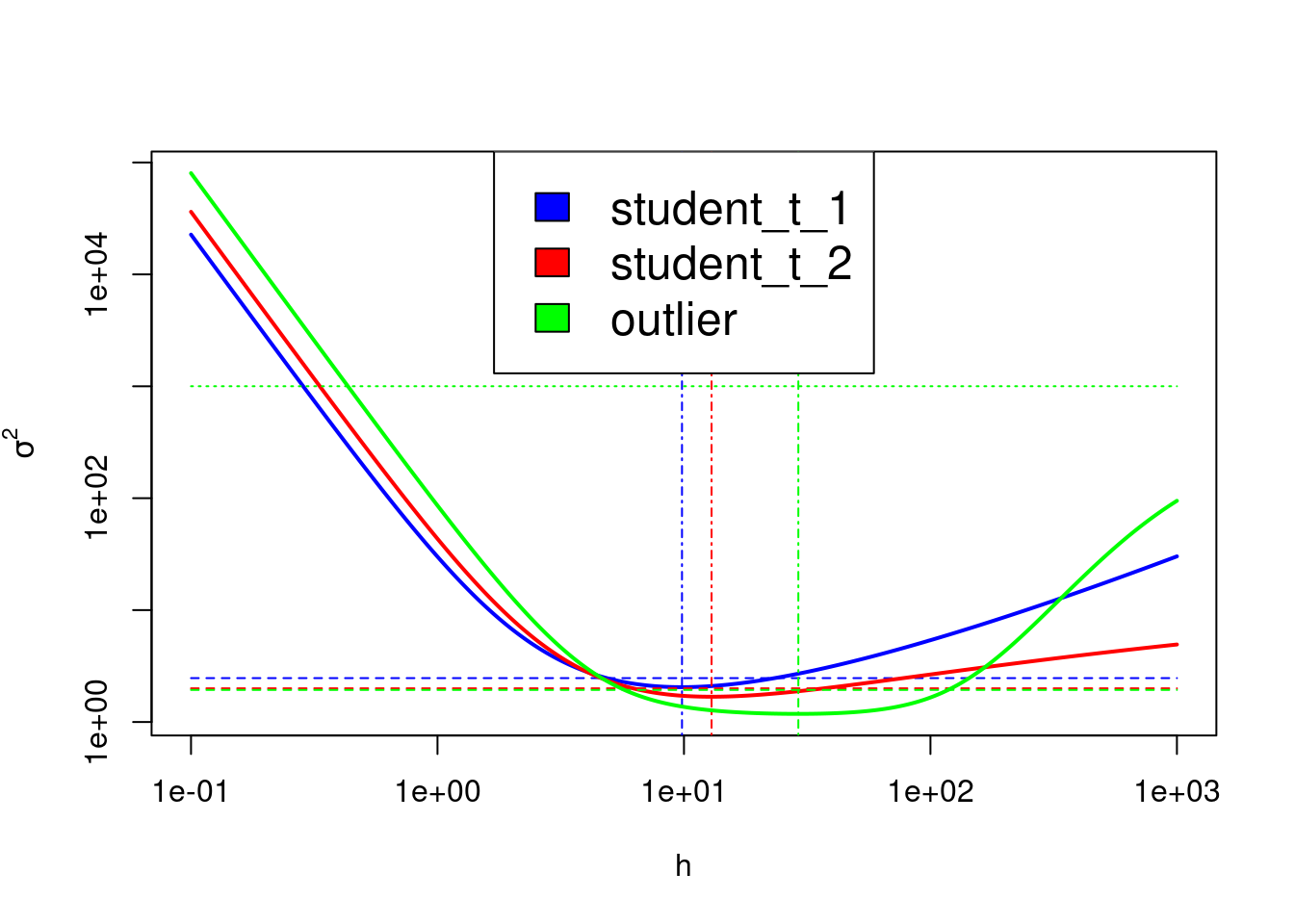}
        \caption{$\chaconserranokernelparam = 1/4$}
        \label{fig:kme-variance-beta-0-25-fat}
    \end{subfigure}
    \begin{subfigure}[b]{\figurewidth}
        \centering
        \includegraphics[width=\subfigurewidth]{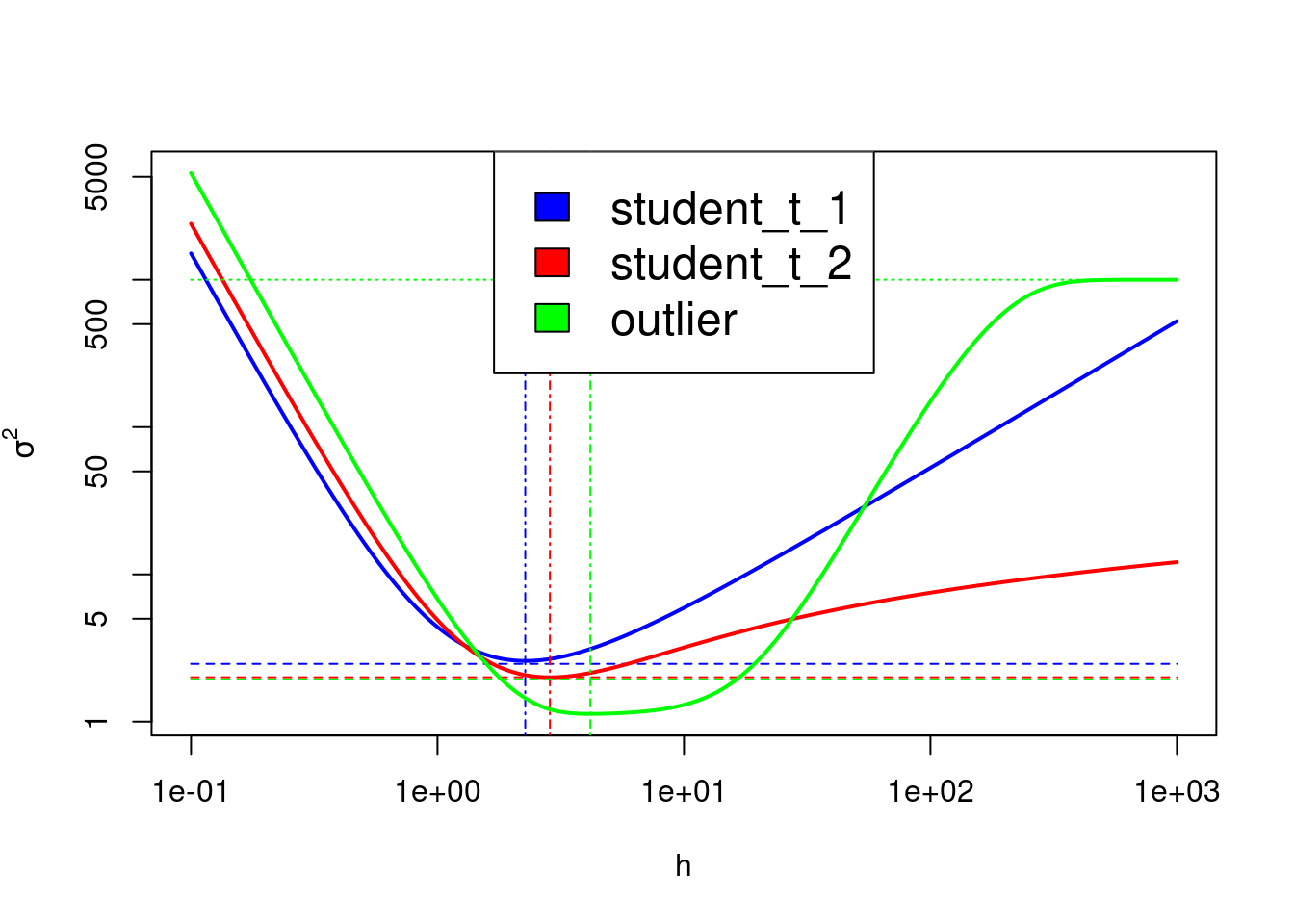}
        \caption{$\chaconserranokernelparam = 8$}
        \label{fig:kme-variance-beta-8-fat}
    \end{subfigure}
    \begin{subfigure}[b]{\figurewidth}
        \centering
        \includegraphics[width=\subfigurewidth]{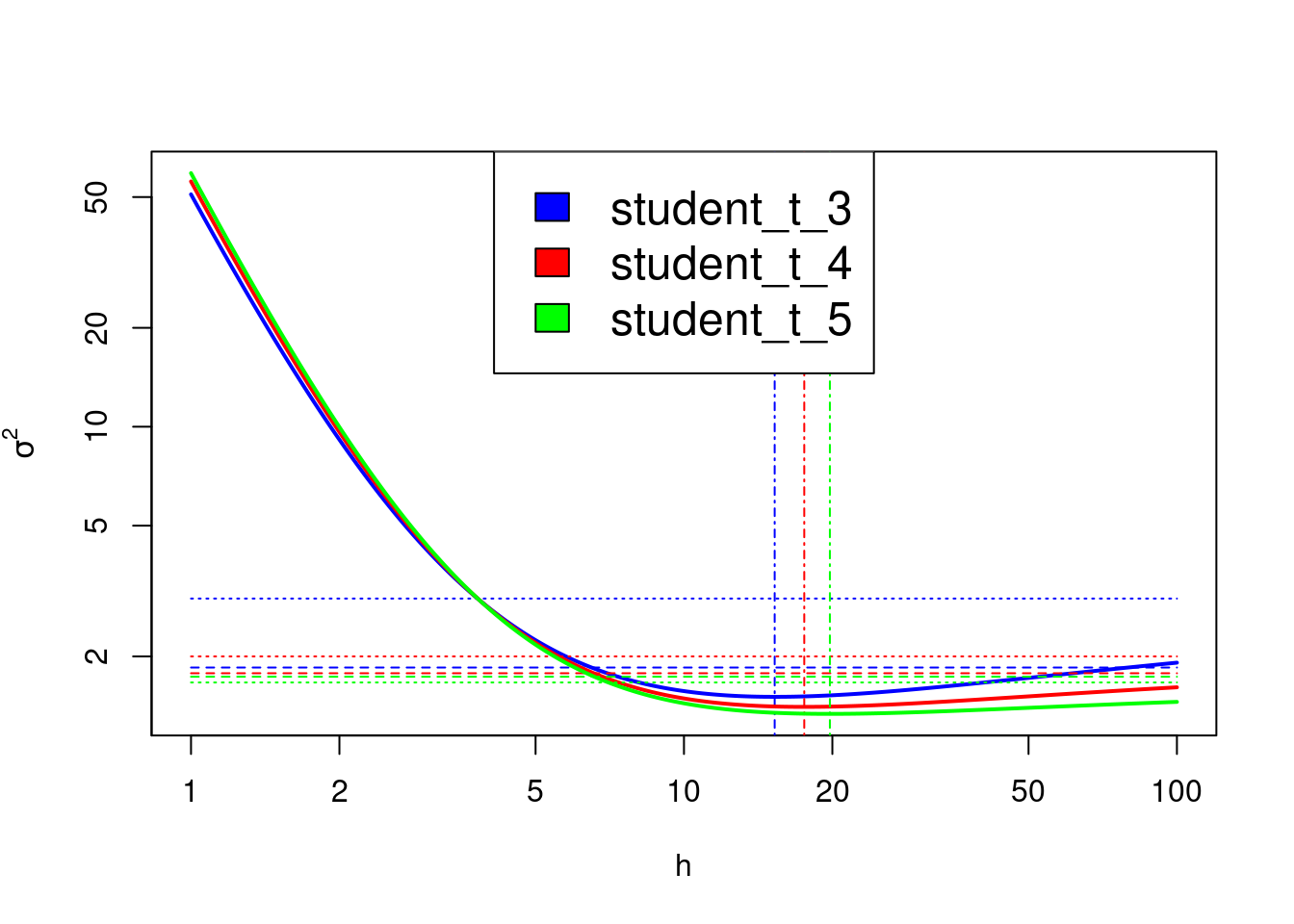}
        \caption{$\chaconserranokernelparam = 1/4$}
    \end{subfigure}
    \begin{subfigure}[b]{\figurewidth}
        \centering
        \includegraphics[width=\subfigurewidth]{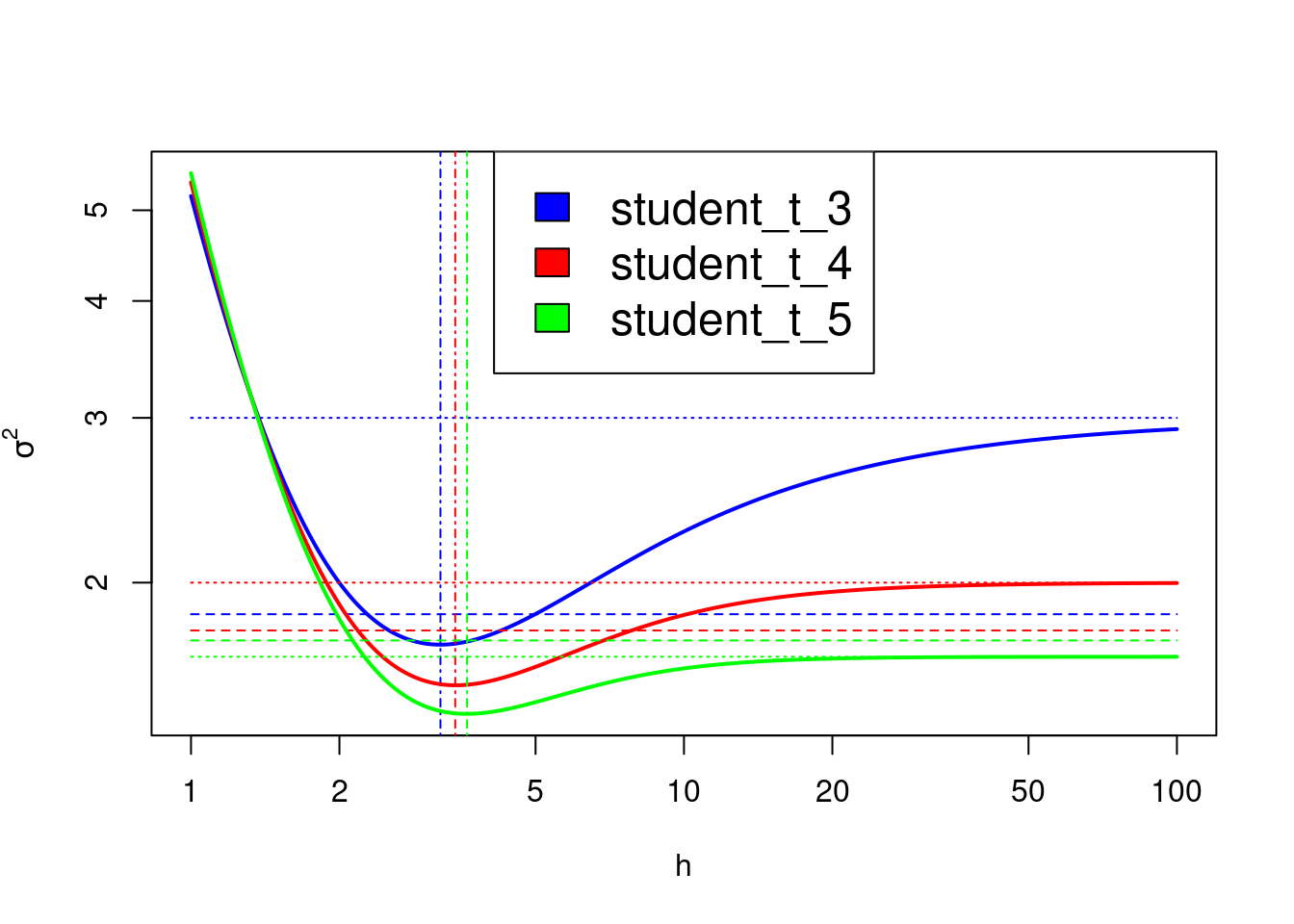}
        \caption{$\chaconserranokernelparam = 8$}
    \end{subfigure}
    \begin{subfigure}[b]{\figurewidth}
        \centering
        \includegraphics[width=\subfigurewidth]{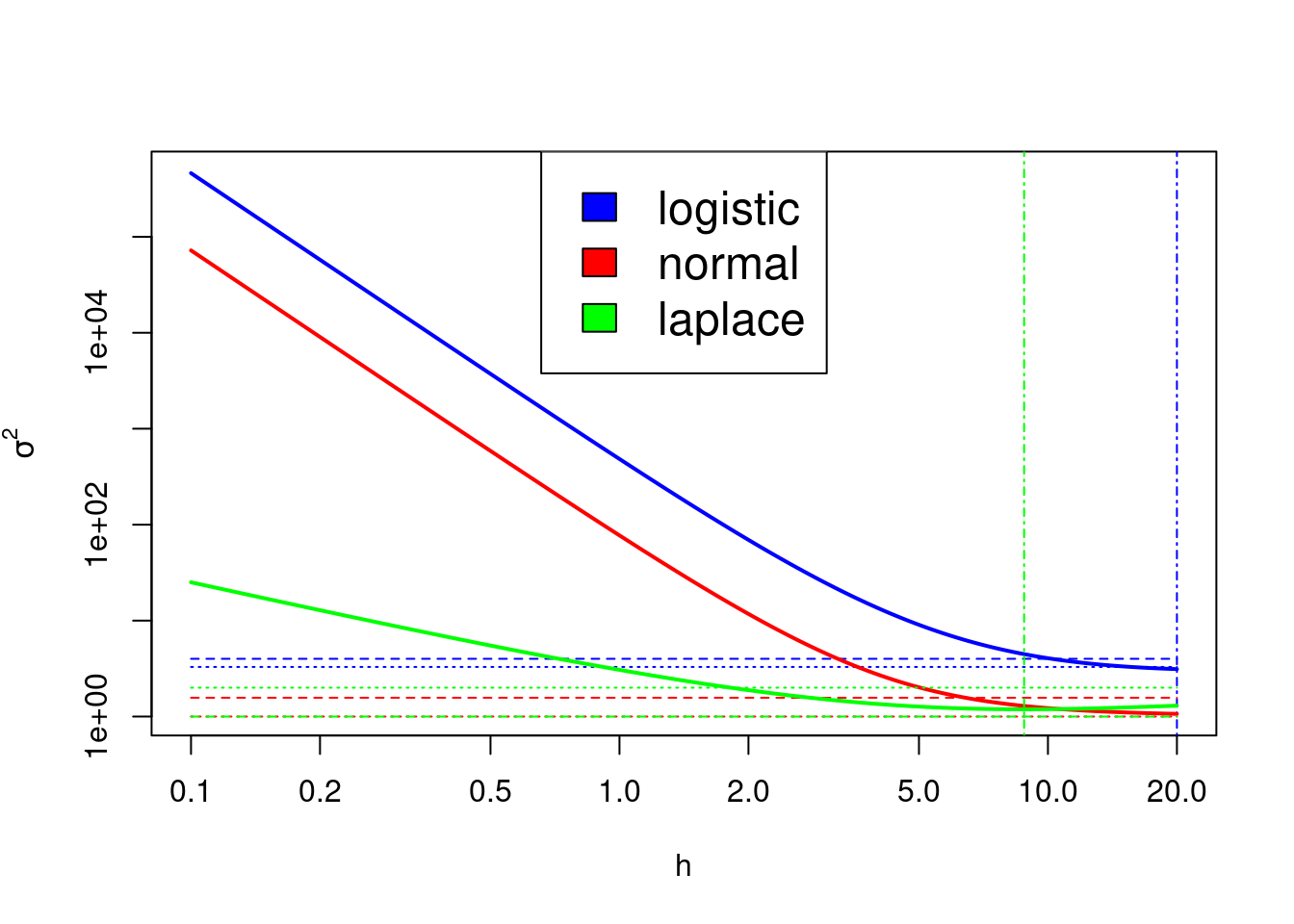}
        \caption{$\chaconserranokernelparam = 1/4$}
    \end{subfigure}
    \begin{subfigure}[b]{\figurewidth}
        \centering
        \includegraphics[width=\subfigurewidth]{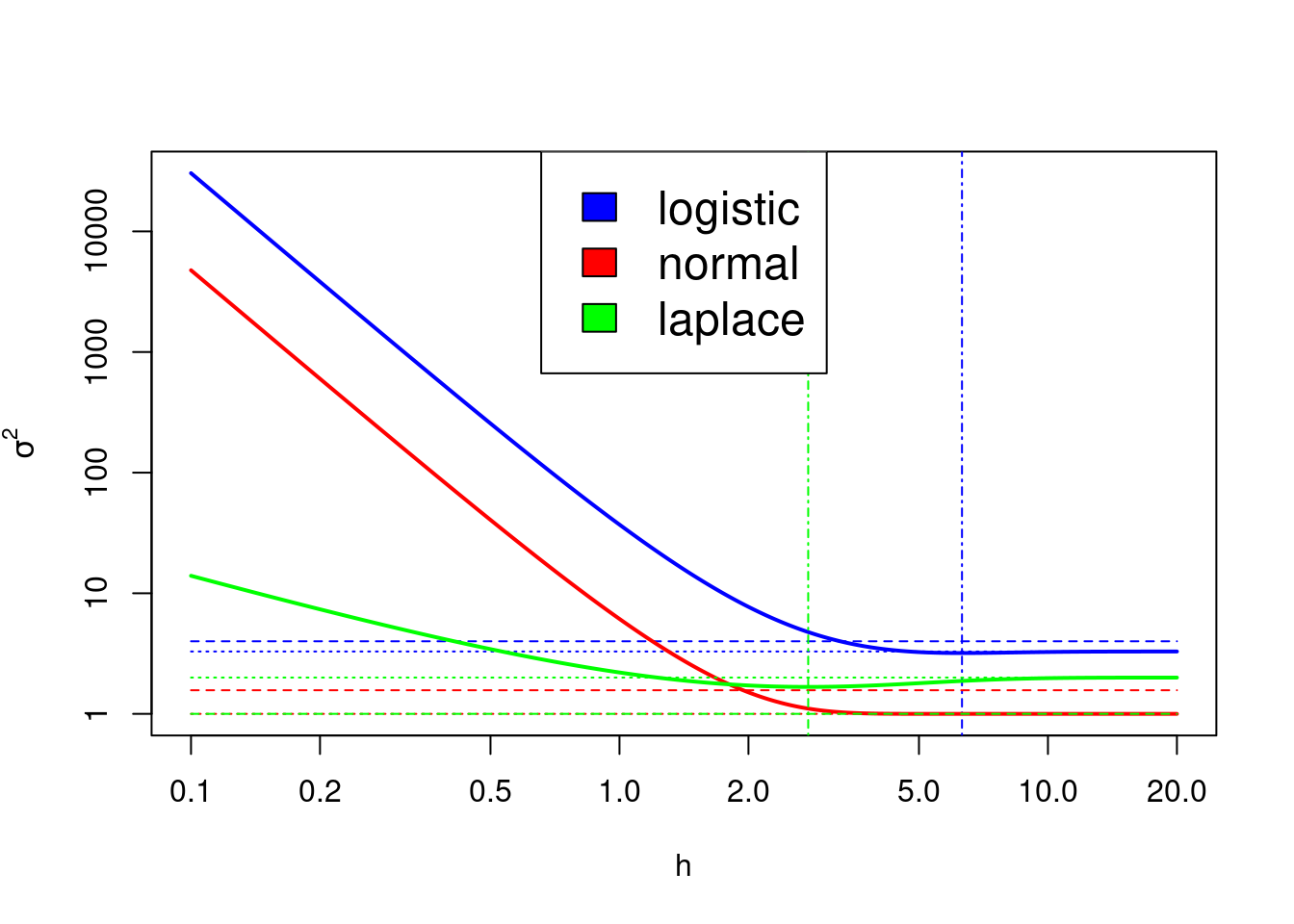}
        \caption{$\chaconserranokernelparam = 8$}
        \label{fig:kme-variance-beta-8-thin}
    \end{subfigure}
    \caption{%
        Several curves of the function $\lambdafunction{\bandwidth}{\squaredsmoothedstddevmode}$ for various combinations of \gls*{pdf} $\pdf$ and parameter $\chaconserranokernelparam$, taking $\kernel \equiv \chaconserranokernel$.
        Both the horizontal and vertical axes are on a logarithmic scale.
        The subfigures from the left column correspond to $\chaconserranokernelparam = 1/4$, while those on the right consider $\chaconserranokernelparam = 8$.
        The \glspl*{pdf} are those in \figurename~\ref{fig:symmetric-unimodal-test-beds}.
        The top row includes two instances from the Student's t, with $\studenttdegreesoffreedom \in \{1, 2\}$, and the \textit{outlier} \gls*{pdf}.
        The mid row gathers the remaining representatives of the Student's t, with $\studenttdegreesoffreedom \in \{3, 4, 5\}$.
        The bottom row comprises the normal, logistic, and Laplace \glspl*{pdf}.
        All the subfigures have the same structure.
        The elements related to a given \gls*{pdf} appear in the same color whenever defined and finite.
        The solid curve is the main variance function.
        The values $\squarestddev = \samplesize \varianceof{\samplemeanestimator}$ and $[2 \pdf(\median)]^{-2} \sim \samplesize \varianceof{\samplemedianestimator}$ appear as the horizontal dotted and dashed lines, respectively.
        Finally, the optimal $\bandwidth$ value minimizing $\squaredsmoothedstddevmode$ is a vertical line.
    }
    \label{fig:kme-variance}
\end{figure}

If we consider a Student's t \gls*{pdf} with $\studenttdegreesoffreedom > 2$ degrees of freedom, Theorem~\ref{th:chacon-serrano-kernel} says we should take $\chaconserranokernelparam > \studenttdegreesoffreedom$.
Therefore, as $\nugoestoinfty$, approaching Gaussianity, we must also take $\chaconserranokernelparam \goestoinfty$, and we see that $\lim_{\chaconserranokernelparam \goestoinfty} \derivativeofchaconserranokernel(\x) \propto -\x \cdot \indicatoroversupportatx$ for all $\x \in \reals \setminus \{\pm 1\}$.
Consequently, $\lim_{\chaconserranokernelparam \goestoinfty} \chaconserranokernel = \epanechnikovkernel$ uniformly.
In that sense, the kernel~\eqref{eq:chacon-serrano-kernel-derivative} asymptotically approaches the optimal behavior of the Epanechnikov kernel for Gaussian data.
Interestingly, the \gls*{kme} $\samplemode$ with kernel $\epanechnikovkernel$ and bandwidth $\bandwidth$ is equivalent to the M-estimator with
\begin{equation}
    \label{eq:trimmed-mean-m-estimator-rho}
    \mestimatorloss(\x)
    =
    \begin{cases}
        \x^2
        \,,
        \  & \mathrm{if} \
        \absvalof{\x} < \bandwidth
        \\
        \squarebandwidth
        \,,
        \  & \mathrm{otherwise}
    \end{cases}
    \,.
\end{equation}
Therefore, in that case, $\samplemode$ has a straightforward interpretation as a sort of ``trimmed mean'', i.e., an average of those $\ithsamplevar$ satisfying $\absvalof{\ithsamplevarminussamplemode} < \bandwidth$~\citep[p. 79]{Huber1964}.

Before moving on to the computational aspects of \glspl*{kme}, the essential theoretical results from this section are graphically summed up in \figurename~\ref{fig:kme-variance}.
There, we can see several plots of the variance function $\lambdafunction{\bandwidth}{\variancefunction(\bandwidth)}$ for various combinations of true \gls*{pdf} $\pdf$ and kernel $\chaconserranokernel$.
First, the two parts of Theorem~\ref{th:limit-variance-h} are confirmed.
Increasingly small values of $\bandwidth$ make the variance explode to infinity, while the variance of the \gls*{kme} approaches that of the true \gls*{pdf} as $\bandwidthgoestoinfty$, including the case of the Student's t for $\studenttdegreesoffreedom \in \{1, 2\}$, where $\squarestddev = \infty$.
In the case of the Laplace \gls*{pdf}, not differentiable at $\symmetrycenter$, the explosion as $\bandwidthgoestopositivezero$ turns out to be less dramatic.

Secondly, \figurename~\ref{fig:kme-variance} depicts the existence of a bandwidth $\optimalbandwidth$ minimizing the variance $\variancefunction(\bandwidth)$, even for \glspl*{pdf} that are not regularly varying, falling out of the hypotheses of Corollary~\ref{th:corollary-optimal-h}.
The only case in \figurename~\ref{fig:kme-variance} for which the variance does not have a finite minimizer is for the normal distribution.
Of course, this is because the sample mean $\samplemeanestimator$ coincides with the \gls*{mle} in this case, and, hence, it cannot be improved.
On the other hand, in most of the remaining cases, the \gls*{kme} with bandwidth $\optimalbandwidth$ not only improves over the sample mean, but it also presents lower variance than the sample median $\samplemedianestimator$, which has $\varianceof{\samplemedianestimator} \asymptoticallyequivalent [4\samplesize\pdf(\symmetrycenter)^2]^{-1}$~\citep[Eq. 14]{Lai1983}.
The only exception is for the Laplace \gls*{pdf}, for which, again, the sample median is also the \gls*{mle} of its location parameter~\citep[p. 23]{Maronna2019}.
All in all, despite having some narrow margins for improvement (see, e.g., \figurename~\ref{fig:kme-variance-beta-8-thin}), the \gls*{kme} should theoretically outperform the sample mean and median over these test-beds whenever possible.

Lastly, \figurename~\ref{fig:kme-variance} exposes an intuitive yet not trivial fact from~\eqref{eq:sample-mode-asymptotic-variance}: the \textit{shape} of the kernel $\kernel$ impacts the efficiency of the \gls*{kme}, in addition to the bandwidth.
In some extreme cases, a poorly chosen kernel can make the \gls*{kme} severely underperform, having an optimal minimum variance that exceeds that of a much simpler method like the sample median.
That is precisely the case in \figurename~\ref{fig:kme-variance-beta-8-fat} for $\chaconserranokernelparam = 8$ and Student's t with $\studenttdegreesoffreedom = 1$.
By contrast, in \figurename~\ref{fig:kme-variance-beta-0-25-fat}, the kernel $\chaconserranokernel$ with $\chaconserranokernelparam = 1/4$ outperforms the sample median for a range of values of $\bandwidth$.

\subsection{Computation}

\label{sec:computation}

There is no closed-form antiderivative for~\eqref{eq:chacon-serrano-kernel-derivative}.
Nonetheless, we shall see that there is a straightforward iterative algorithm to calculate $\samplemode$ without evaluating $\chaconserranokernel$ or even $\derivativeofchaconserranokernel$.
Again, the kernel~\eqref{eq:chacon-serrano-kernel-derivative} has a convenient interpretation when the \gls*{kme} is regarded as an M-estimator.
Adapting~\citet[Section 2.3.3]{Maronna2019} to our context, the weight function
\begin{equation}
    \label{eq:weight-function}
    \weightfunctionh(\x)
    =
    \begin{cases}
        -\derivativeofkernelh(\x) / \x
        \,,
        \  & \mathrm{if} \
        \x \neq 0
        \\
        -\secondderivativeofkernelh(0)
        \,,
        \  & \mathrm{if} \
        \x = 0
    \end{cases}
    \,,
\end{equation}
which is always positive, allows expressing the \gls*{kme} as the weighted mean
\begin{equation}
    \label{eq:kme-as-weighted-mean}
    \samplemode
    =
    \frac
    {%
        \sumoverallsample
        \weightfunctionh(\ithsamplevarminussamplemode)
        \ithsamplevar
    }
    {%
        \sumoverallsample
        \weightfunctionh(\ithsamplevarminussamplemode)
    }
    \,,
\end{equation}
where the weights on the right-hand side also depend on $\samplemode$.

Equation~\eqref{eq:kme-as-weighted-mean} implies that $\samplemode$ is \textit{one} of the fixed points of the function $\meanshiftaverage$ given by
\begin{equation}
    \label{eq:mean-shift-sample-average}
    \meanshiftaverage(\x)
    =
    \frac
    {%
        \sumoverallsample
        \weightfunctionh(\ithsamplevar - \x)
        \ithsamplevar
    }
    {%
        \sumoverallsample
        \weightfunctionh(\ithsamplevar - \x)
    }
    \,,
\end{equation}
when the denominator in~\eqref{eq:mean-shift-sample-average} does not vanish, and $\meanshiftaverage(\x) = \x$ otherwise.
This observation suggests an iterative procedure for computing $\samplemode$, known as \gls*{irw}.
Such an algorithm turns out to be numerically more stable than solving~\eqref{eq:m-estimator-argmin} or~\eqref{eq:m-estimator-equation-root} through customary optimization or root-finding algorithms, respectively~\citep[Section 2.10.5.1]{Maronna2019}.

The subindex $\bandwidth$ on the left-hand side of~\eqref{eq:weight-function} is a convenient notation to express the dependence on $\bandwidth$.
Despite $\weightfunctionh$ being similar to a kernel in the \gls*{kde} sense, it is generally \textit{not} a \gls*{pdf}.
Also, if we denote $\weightfunction$ the weight function~\eqref{eq:weight-function} corresponding to $\bandwidth = 1$, we would have $\weightfunctionh(\x) \neq \weightfunction(\x / \bandwidth) / \bandwidth$ for all $\bandwidth \neq 1$.
To obtain \gls*{pdf}-like scaling behavior, we should take $\lambdafunction{\x}{\squarebandwidth \weightfunctionh(\x)}$.

The following discussion of \gls*{irw} remains valid for other kernels, but we shall focus on $\chaconserranokernel$.
Considering $\kernel \equiv \chaconserranokernel$ in~\eqref{eq:weight-function}, we get $\weightfunctionh(\x) \propto \bumplikefunction(\inversebandwidth \x)$, where the exact value of the proportionality constant does not matter, as it cancels out due to the normalizing denominator in~\eqref{eq:kme-as-weighted-mean}.
Taking $\chaconserranokernelparam \goestoinfty$, we get a \textit{flat} weight function, i.e., $\weightfunctionh(\x) \propto 1$, if $\absvalof{x} < \bandwidth$, and zero elsewhere, retrieving the ``trimmed mean'' discussed above.

Let us assume that a choice of $\chaconserranokernelparam$ and $\bandwidth$ has been made.
The \gls*{irw} procedure is implemented as follows.
First, define the score function
\begin{equation}
    \label{eq:chacon-serrano-score-function}
    \scorefunction(\x)
    =
    \begin{cases}
        -
        (1 - \absvalof{\x}^{\chaconserranokernelparam})^{-1}
        \,,
        \  & \mathrm{if} \
        \absvalof{\x} < 1
        \\
        -\infty
        \,,
        \  & \mathrm{otherwise}
    \end{cases}
    \,.
\end{equation}
Starting from an initial guess $\samplemodeinitialguess$ at $\samplemode$, the transition from the $\samplemodekthiteration$-th approximation $\previousitersamplemode$ to the $(\samplemodekthiteration + 1)$-th approximation $\currentitersamplemode$, for $\samplemodekthiteration \geq 0$, is given by
\begin{equation}
    \label{eq:irw-update}
    \currentitersamplemode
    =
    \meanshiftaverage(\previousitersamplemode)
    =
    \sumoverallsample
    \currentweightvector_{\genindexi}
    \ithsamplevar
    \,,
\end{equation}
where the $\samplemodekthiteration$-th weight vector $\currentweightvector$ depends on $\previousitersamplemode$ through
\begin{equation*}
    \currentweightvector
    =
    \softmax
    \left[
        \scorefunction
        \left(
        \frac
        {\samplerandomvarindex{1} - \previousitersamplemode}
        {\bandwidth}
        \right),
        \dots,
        \scorefunction
        \left(
        \frac
        {\samplerandomvarindex{\samplesize} - \previousitersamplemode}
        {\bandwidth}
        \right)
        \right]
    \,,
\end{equation*}
and $\softmax$ is the \textit{softmax} function mapping $\scorevector = (\scorevector_1, \dots, \scorevector_{\samplesize})$ to $\softmax(\scorevector) = (\ithsoftmaxscorevector{1}, \dots, \ithsoftmaxscorevector{\samplesize})$, for
$\ithsoftmaxscorevector{\genindexi} = e^{\scorevector_{\genindexi}} / \sum_{\genindexj = 1}^{\samplesize} e^{\scorevector_{\genindexj}}$, using the convention $e^{-\infty} = 0$.
More abstractly, for $\samplemodekthiteration \geq 1$, if we define
$
    \kthmeanshiftaverage{\samplemodekthiteration}
    =
    \meanshiftaverage \circ \dots \circ \meanshiftaverage
    \
    (\samplemodekthiteration \ \mathrm{times})
$,
then
$
    \previousitersamplemode
    =
    \kthmeanshiftaverage{\samplemodekthiteration}(\samplemodeinitialguess)
$.
Iterations stop on some convergence criterion, such as
$
    \smallabsvalof{
        \currentitersamplemode
        -
        \previousitersamplemode
    }
    \leq
    \irwtolerance
    \bandwidth
$,
where $\irwtolerance > 0$ is a typically small tolerance parameter~\citep[Section 2.8.1]{Maronna2019}.

Following \citet[Section 2.8.1]{Maronna2019}, the convergence of \gls*{irw} in our context is guaranteed because $\chaconserranokernel$ is sufficiently smooth, and its corresponding $\weightfunctionh$ is bounded and monotonically decreasing as a function of $\absvalof{\x}$.
If there is a unique maximizer $\samplemode$ without other local maxima, then $
    \lim_{\samplemodekthiteration \goestoinfty}
    \previousitersamplemode
    =
    \samplemode
$,
regardless of the starting point $\samplemodeinitialguess$.
Otherwise, $\samplemodeinitialguess$ must be close to $\samplemode$ to avoid ``bad solutions''.
A standard choice to prevent the worst-case scenario is $\samplemodeinitialguess = \samplemedianestimator$, the sample median.
The \gls*{kme} calculated this way also inherits the optimal breakdown point of the sample median~\citep[p. 55]{Huber2009}.
Nevertheless, at the expense of a computational cost increase, an even \textit{safer}, minimal-risk choice is $\samplemodeinitialguess = \samplerandomvarindex{\genindexj}$, where $\genindexj = \argmax_{1 \leq \genindexi \leq \samplesize} \kdeat{\ithsamplevar}$.
However, the latter requires explicitly evaluating the \gls*{kde}, which \gls*{irw} was meant to avoid.

\begin{remark}
    It is worth noting here the link between the \gls*{irw} algorithm and modal clustering~\citep{Chacon2015}.
    Indeed, the update scheme~\eqref{eq:irw-update} is known in clustering as the \textit{mean shift} algorithm~\citep{Fukunaga1975}, which iteratively translates any initial point through the steepest density ascent path until it reaches a local maximum, and then clusters together all the points that converge to the same local density mode after such a translation.
\end{remark}

\subsection{Parameter optimization}

\label{sec:parameter-optimization}

As mentioned above, carefully choosing $\chaconserranokernelparam$ and $\bandwidth$ is critical for obtaining an efficient \gls*{kme} $\samplemode$.
We propose optimizing $\chaconserranokernelparam$ and $\bandwidth$ based on the data.
Let us denote $\targetvariance$ the asymptotic variance~\eqref{eq:sample-mode-asymptotic-variance} considering a kernel $\kernel \equiv \chaconserranokernel$.
Then, the optimal shape and bandwidth parameters are those that minimize $\targetvariance$, i.e.,
$
    (\optimalchaconserranokernelparam, \optimalbandwidth)
    =
    \argmin_{\chaconserranokernelparam > 0, \bandwidth > 0}
    \targetvariance
$.
Let us define a \textit{standardized} version of the random variable $\randomvarx \followsdistr \pdf$ centered on its actual center of symmetry $\symmetrycenter$ as
$
    \standardrandomvar
    =
    (\centeredvar{\randomvarx}) / \bandwidth
$,
where the bandwidth $\bandwidth$ plays the role of a scaling parameter.
Also, denoting
$
    \chaconserranokernelpsix
    =
    -\x \bumplikefunction(\x)
$,
if we define the auxiliary functions
$
    \ithsymmetricfunction{1}{\x}
    =
    \chaconserranokernelpsix^2
$
and
$
    \ithsymmetricfunction{2}{\x}
    =
    \derivativeofwrt{\chaconserranokernelpsix}{\x}
$,
we have
\begin{equation}
    \label{eq:target-variance}
    \targetvariance
    =
    \squarebandwidth
    \
    \frac
    {\ithsymmetricexpectation{1}}
    {\ithsymmetricexpectation{2}^2}
    \,.
\end{equation}
Using the symmetry of the auxiliary functions and the target \gls*{pdf}, and the compact support $[-1, 1]$ of the former, inherited from $\bumplikefunction$, we get, for $\auxiliaryfunctionindex \in \{1, 2\}$,
\begin{equation}
    \label{eq:iths-ymmetric-expectation}
    \ithsymmetricexpectation{\auxiliaryfunctionindex}
    =
    2\bandwidth
    \intoverunitinterval
    \ithsymmetricfunction{\auxiliaryfunctionindex}{\x}
    \
    \untranslatedpdf(\bandwidth\x)
    \ \dx
    \,,
\end{equation}
where we recall that $\untranslatedpdf(\x) = \pdfatmuplusx$ is the centered version of $\pdf$ in Definition~\ref{def:unimodality}.
Both integrals~\eqref{eq:iths-ymmetric-expectation} can be accurately computed using standard numerical methods.

Minimizing~\eqref{eq:target-variance}, a two-dimensional optimization problem, is more complicated than the one-dimensional scenario depicted in \figurename~\ref{fig:kme-variance}, where $\chaconserranokernelparam$ was fixed.
We propose employing a gradient-free optimization algorithm such as~\citet{Nelder1965}, which produces more than satisfactory results, as Section~\ref{sec:simulation-study} will demonstrate.
In this respect, we should emphasize that, rather than necessarily finding the global optimum, our goal is to \textit{improve} on some default sensible parameter guesses $\chaconserranokernelparam = 1$ and $\bandwidth = \madn$, where the latter is the normalized \textit{median absolute deviation about the median} defined in~\citet[p. 5]{Maronna2019}, a robust scale estimator in the context of M-estimators.

Since $\pdf$ is a priori unknown in~\eqref{eq:target-variance}, we propose employing a \textit{plug-in}-type estimator, replacing $\untranslatedpdf$ with a convenient estimate $\estimateofuntranslatedpdf$ in~\eqref{eq:iths-ymmetric-expectation}.
Specifically, following~\citet{Meloche1991}, we can take $\estimateofuntranslatedpdf(\x) = \symmetrickde(\samplemedianestimator + \x)$, where
\begin{equation}
    \label{eq:symmetric-kde}
    \symmetrickde(\x)
    =
    \frac{\kdesecondary(x) + \kdesecondary(2\samplemedianestimator - \x)}{2}
    \,,
\end{equation}
and $\kdesecondary$ is the \gls*{kde} in~\eqref{eq:kde-definition} but with a custom bandwidth $\bandwidthsecondary$ independent of the $\bandwidth$ in~\eqref{eq:target-variance}.
Since $\estimateofuntranslatedpdf$ will be evaluated many times when computing~\eqref{eq:iths-ymmetric-expectation}, we recommend employing, especially for large samples, a \textit{binned} interpolated approximation~\citep[Appendix D.2]{Wand1995} for $\kdesecondary$ over $[\samplemedianestimator - \bandwidthupperbound, \samplemedianestimator + \bandwidthupperbound]$ using a fine grid, where $\bandwidthupperbound > 0$ is some reasonable upper bound for the optimal $\bandwidth$.

The function~\eqref{eq:symmetric-kde} is a modified version of the household \gls*{kde} that considers the symmetry of $\pdf$ about $\symmetrycenter$.
Indeed,~\eqref{eq:symmetric-kde} is a symmetric \gls*{pdf} about the sample median $\samplemedianestimator$, an estimator of $\symmetrycenter$.
Alternatively,~\eqref{eq:symmetric-kde} can be seen as the \gls*{kde} with an \textit{augmented} sample $(\randomsample, 2\samplemedianestimator - \samplerandomvarindex{1}, \dots, 2\samplemedianestimator - \samplerandomvarindex{\samplesize})$, establishing connections with similar procedures in robust statistics such as~\citet{Mehrotra1991} \citep[see also][Example 1]{Chacon2009}.
Results by~\citet{Meloche1991} show that, under mild assumptions, symmetrization of the \gls*{kde} about a well-behaved estimator of $\symmetrycenter$, such as the sample median, halves the variance term in the \textit{mean integrated square error} for \gls*{pdf} estimation.

Given the stringent assumptions of symmetry and unimodality, a basic and fast \textit{rule-of-thumb} bandwidth selection criterion, such as that of~\citet[p. 48]{Silverman1986}, should provide a reasonable estimate for $\bandwidthsecondary$.
Such a procedure relies on relatively minor departures from normality, usually producing over-smoothed \glspl*{kde} in general multimodality settings~\citep[Section 3.2.1]{Wand1995}.
Nonetheless, over-smoothing should play to our advantage, producing tails in $\estimateofuntranslatedpdf$ that are robust against outliers and enforce unimodality.
Even though other more sophisticated bandwidth selectors could potentially yield better results at ensuring unimodality,~\citeauthor{Silverman1986}'s bandwidth copes well with this particular scenario while keeping the computational cost to a minimum.

\subsection{Illustrative example}

\label{sec:illustrative-example}

Let us illustrate our \gls*{kme} proposal through a synthetic data example.
Consider the sample realization given by $(\randomvarobsindex{1}, \dots, \randomvarobsindex{7}) = (-2, -1, 0, 1, 2, 10, 11)$, with $\samplesize = 7$.
The sample mean and median realizations are $\samplemeanestimator = 3$ and $\samplemedianestimator = 1$, respectively.
None of them makes a convincing candidate for the center of symmetry, as the pilot \gls*{kde} $\kdesecondary$ shown in red in \figurename~\ref{fig:kme-synthetic-data} suggests a bimodal structure with two subpopulations.
The largest one, corresponding to the first five observations, reaches its maximum density near zero, which would be a perfect center of symmetry if we removed the last two observations ($10$ and $11$) in the second cluster.
Therefore, finding a good center of symmetry entails separating \textit{bulk} data from \textit{outlying} data.
In other circumstances, if unimodality were not assumed, those outliers would be worth analyzing.

The symmetric \gls*{kde}~\eqref{eq:symmetric-kde} is depicted in blue in \figurename~\ref{fig:kme-synthetic-data}.
The median $\samplemedianestimator$ about which the original \gls*{kde} $\kdesecondary$ is \textit{mirrored} shows as the blue vertical dashed line.
A slight shift in the maximum of $\symmetrickde$ is the price for symmetric and thinner tails, more compatible with our assumptions.
The optimal parameters for our \gls*{kme} are $\optimalparameters = (1.765101, 9.199545)$.
Then, the corresponding \textit{rescaled} kernel centered on the median, i.e., $\lambdafunction{\x}{\chaconserranokernel[\bandwidth^{-1} (\x - \samplemedianestimator)] / \bandwidth}$, is displayed as the dotted black line in \figurename~\ref{fig:kme-synthetic-data}.
Finally, the \gls*{kde} underlying the \gls*{kme}, built from the previous kernel, is shown in solid green, while the \gls*{kme} is given by the green vertical dashed line.
As we can see, the optimal \gls*{kde} is slightly smoother than the original $\kdesecondary$, but both ultimately reach their peak near zero.

\begin{figure}
    \centering
    \includegraphics[width=0.5\textwidth]{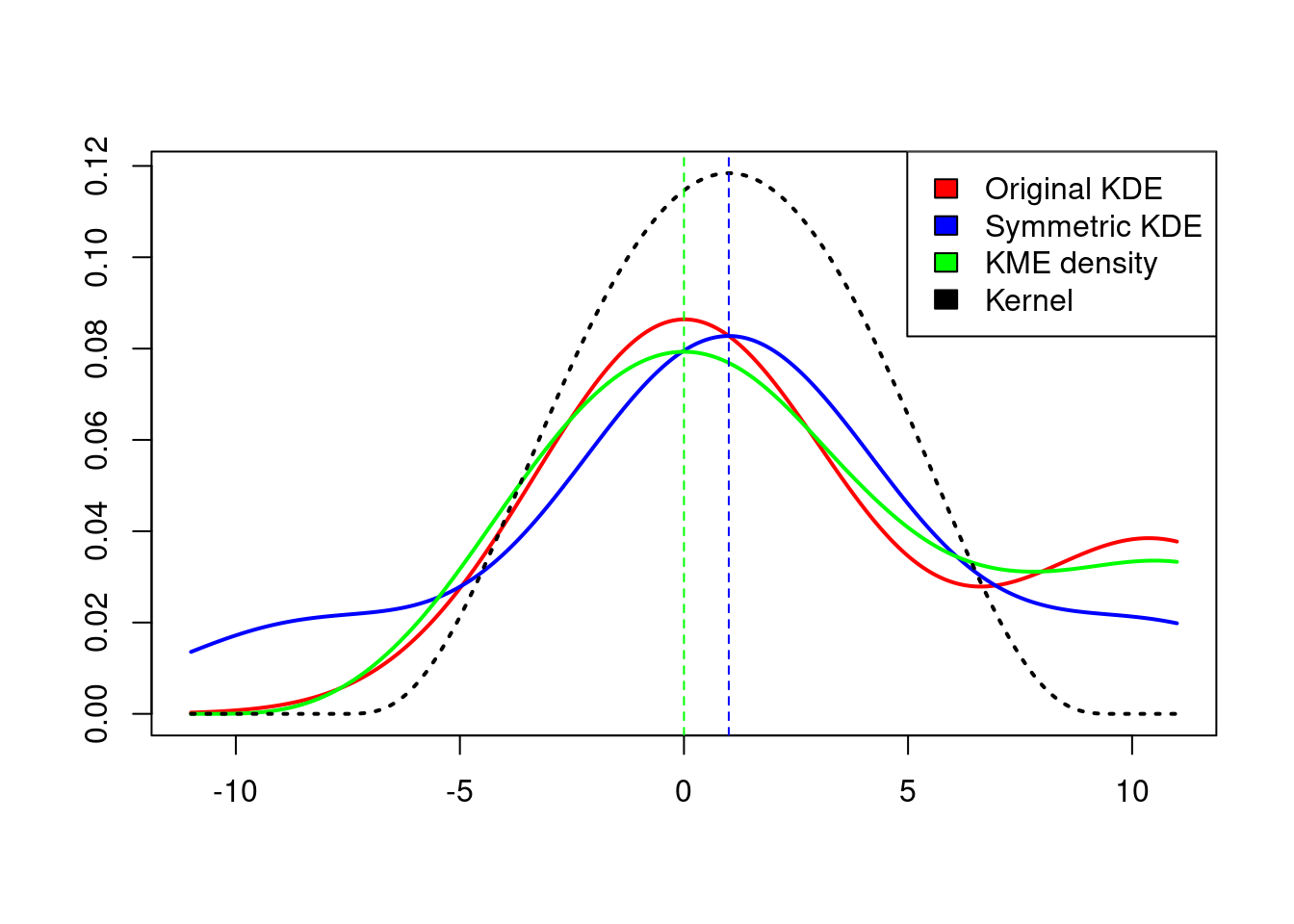}
    \caption{%
        Plot summary of all the relevant functions involved in \gls*{irw} for the synthetic data example in Section~\ref{sec:illustrative-example}.
        The pilot \gls*{kde} $\kdesecondary$ and its symmetrized version $\symmetrickde$ about the median, represented as a blue vertical dashed line, are shown in red and blue, respectively.
        The black dotted curve is the rescaled kernel with optimal parameters $\optimalparameters$.
        The subsequent \gls*{kde} $\kde$ underlying $\samplemode$ shows in green, while the green vertical dashed line corresponds to the \gls*{irw} computation result.
    }
    \label{fig:kme-synthetic-data}
\end{figure}

Table~\ref{tab:irw-synthetic-data} gathers the final and intermediate results from the \gls*{irw} algorithm.
Convergence was reached after eight iterations.
From the first iteration, the observation $\randomvarobsindex{7}$ has zero weight, while the weight of $\randomvarobsindex{6}$ vanishes from the second one onwards.
As the algorithm progresses, the weights stabilize symmetrically about $\randomvarobsindex{3} = 0$, which has the most mass.
Ultimately, the value of $\samplemode$ is practically indistinguishable from zero.

\begin{table}
    \centering
    \footnotesize
    \begin{tabular}{ccccccccc}
    \toprule
    $k$ & $\mathbf{w}^{(k)}_1$ & $\mathbf{w}^{(k)}_2$ & $\mathbf{w}^{(k)}_3$ & $\mathbf{w}^{(k)}_4$ & $\mathbf{w}^{(k)}_5$ & $\mathbf{w}^{(k)}_6$ & $\mathbf{w}^{(k)}_7$ & $\hat{\mathsf{m}}_{k + 1}$ \\
    \midrule
    0   & \texttt{1.80E-01}    & \texttt{1.96E-01}    & \texttt{2.07E-01}    & \texttt{2.11E-01}    & \texttt{2.07E-01}    & \texttt{2.09E-12}    & \texttt{0}           & \texttt{6.89E-02}          \\
    1   & \texttt{1.92E-01}    & \texttt{2.03E-01}    & \texttt{2.07E-01}    & \texttt{2.04E-01}    & \texttt{1.94E-01}    & \texttt{0}           & \texttt{0}           & \texttt{4.67E-03}          \\
    2   & \texttt{1.93E-01}    & \texttt{2.03E-01}    & \texttt{2.07E-01}    & \texttt{2.03E-01}    & \texttt{1.93E-01}    & \texttt{0}           & \texttt{0}           & \texttt{3.17E-04}          \\
    3   & \texttt{1.93E-01}    & \texttt{2.03E-01}    & \texttt{2.07E-01}    & \texttt{2.03E-01}    & \texttt{1.93E-01}    & \texttt{0}           & \texttt{0}           & \texttt{2.15E-05}          \\
    4   & \texttt{1.93E-01}    & \texttt{2.03E-01}    & \texttt{2.07E-01}    & \texttt{2.03E-01}    & \texttt{1.93E-01}    & \texttt{0}           & \texttt{0}           & \texttt{1.46E-06}          \\
    5   & \texttt{1.93E-01}    & \texttt{2.03E-01}    & \texttt{2.07E-01}    & \texttt{2.03E-01}    & \texttt{1.93E-01}    & \texttt{0}           & \texttt{0}           & \texttt{9.92E-08}          \\
    6   & \texttt{1.93E-01}    & \texttt{2.03E-01}    & \texttt{2.07E-01}    & \texttt{2.03E-01}    & \texttt{1.93E-01}    & \texttt{0}           & \texttt{0}           & \texttt{6.73E-09}          \\
    7   & \texttt{1.93E-01}    & \texttt{2.03E-01}    & \texttt{2.07E-01}    & \texttt{2.03E-01}    & \texttt{1.93E-01}    & \texttt{0}           & \texttt{0}           & \texttt{4.57E-10}          \\
    \bottomrule
\end{tabular}

    \caption{%
        \Gls*{irw} weights and mode approximations from the update scheme~\eqref{eq:irw-update} for the synthetic data example.
    }
    \label{tab:irw-synthetic-data}
\end{table}

\section{Case study}

\label{sec:case-study}

The estimation of the center of symmetry naturally arises in physics.
A measurement $\samplerandomvar$ can be described by the location model $\locationmodel$, where $\symmetrycenter$ represents an unknown parameter of interest and $\randomerror$ is a random variable accounting for the measurement error~\citep[Section 2.1]{Maronna2019}.
In this context, to dismiss the existence of systematic errors, the error variable $\randomerror$ is usually assumed to be symmetric about zero, physically meaning that overestimating and underestimating are equally likely~\citep{Taylor1997}.

A classic example to illustrate the robust estimation of location is Simon Newcomb's experiment from 1882 for measuring the speed of light~\citetext{\citealp[pp. 66--67]{Gelman2013};~\citealp[Example 1.2]{Maronna2019};~\citealp{Stigler1977}}.
Newcomb measured the time it takes light to cover a distance of 7,442 meters.
The recorded unique values and their number of repetitions are collected in the first two columns of Table~\ref{tab:c-measurements}.
See~\citet[Figure 3.1]{Gelman2013} for a sample histogram.
A total of $\samplesize = 66$ measurements were taken, the lowest two of which ($-44$ and $-2$) are outliers.
The sample mean, $26.2$, is much more affected by the two outliers than the sample median, $27$.
Despite the latter being a more reasonable centrality measure, the value $28$ is still the most repeated in the sample, i.e., the \textit{discrete} mode.
For that matter, if we removed the two low outliers, the value $28$ would be at the same distance from the new minimum (i.e., $16$) and the maximum (i.e., $40$).
We shall see that the \gls*{kme} provides an elegant solution to this problem closer to $28$.

The optimal shape and bandwidth parameters for the \gls*{kme} using our kernel proposal~\eqref{eq:chacon-serrano-kernel-derivative} are $\optimalparameters = (97.03537, 21.23523)$.
Then, the \gls*{irw} algorithm yields $\samplemode = 27.75$.
The corresponding \gls*{irw} unitary and total weights at the last iteration for each sample value are shown in the third and fourth columns of Table~\ref{tab:c-measurements}.
As we can see, the two outliers have no weight, while the rest have the same unit weight of $1 / 64$.
Therefore, \gls*{irw} computes the ``trimmed mean'' M-estimator corresponding to~\eqref{eq:trimmed-mean-m-estimator-rho}, equivalent to the \gls*{kme} with Epanechnikov kernel, employing a bandwidth $\bandwidth = \optimalbandwidth$.
As it turns out, considering the finite computer precision, the score function $\scorefunction$ in~\eqref{eq:chacon-serrano-score-function} with $\chaconserranokernelparam$ equal to the large optimal $\optimalchaconserranokernelparam$ above is numerically indistinguishable from the constant $-1$ over all the standardized random variables $(\ithsamplevarminussamplemode) / \bandwidth$ underlying~\eqref{eq:kme-as-weighted-mean}.
Finally, the fifth column in Table~\ref{tab:c-measurements} shows the value of the \gls*{kde} behind the \gls*{kme} at each sample observation, reaching its maximum at $28$, the discrete mode.

\begin{table}
    \centering
    \footnotesize
    \begin{tabular}{rcccc}
    \toprule
    Value        & Count      & Unit weight          & Total weight         & Density              \\
    \midrule
    \texttt{-44} & \texttt{1} & \texttt{0}           & \texttt{0}           & \texttt{5.41522E-04} \\
    \texttt{-2}  & \texttt{1} & \texttt{0}           & \texttt{0}           & \texttt{8.30601E-04} \\
    \texttt{16}  & \texttt{2} & \texttt{1.56250E-02} & \texttt{3.12500E-02} & \texttt{2.22093E-02} \\
    \texttt{19}  & \texttt{1} & \texttt{1.56250E-02} & \texttt{1.56250E-02} & \texttt{2.66347E-02} \\
    \texttt{20}  & \texttt{1} & \texttt{1.56250E-02} & \texttt{1.56250E-02} & \texttt{2.79265E-02} \\
    \texttt{21}  & \texttt{2} & \texttt{1.56250E-02} & \texttt{3.12500E-02} & \texttt{2.90680E-02} \\
    \texttt{22}  & \texttt{2} & \texttt{1.56250E-02} & \texttt{3.12500E-02} & \texttt{3.00520E-02} \\
    \texttt{23}  & \texttt{3} & \texttt{1.56250E-02} & \texttt{4.68750E-02} & \texttt{3.08786E-02} \\
    \texttt{24}  & \texttt{5} & \texttt{1.56250E-02} & \texttt{7.81250E-02} & \texttt{3.15478E-02} \\
    \texttt{25}  & \texttt{5} & \texttt{1.56250E-02} & \texttt{7.81250E-02} & \texttt{3.20595E-02} \\
    \texttt{26}  & \texttt{5} & \texttt{1.56250E-02} & \texttt{7.81250E-02} & \texttt{3.24138E-02} \\
    \texttt{27}  & \texttt{6} & \texttt{1.56250E-02} & \texttt{9.37500E-02} & \texttt{3.26106E-02} \\
    \texttt{28}  & \texttt{7} & \texttt{1.56250E-02} & \texttt{1.09375E-01} & \texttt{3.26499E-02} \\
    \texttt{29}  & \texttt{5} & \texttt{1.56250E-02} & \texttt{7.81250E-02} & \texttt{3.25318E-02} \\
    \texttt{30}  & \texttt{3} & \texttt{1.56250E-02} & \texttt{4.68750E-02} & \texttt{3.22563E-02} \\
    \texttt{31}  & \texttt{2} & \texttt{1.56250E-02} & \texttt{3.12500E-02} & \texttt{3.18233E-02} \\
    \texttt{32}  & \texttt{5} & \texttt{1.56250E-02} & \texttt{7.81250E-02} & \texttt{3.12329E-02} \\
    \texttt{33}  & \texttt{2} & \texttt{1.56250E-02} & \texttt{3.12500E-02} & \texttt{3.04850E-02} \\
    \texttt{34}  & \texttt{1} & \texttt{1.56250E-02} & \texttt{1.56250E-02} & \texttt{2.95797E-02} \\
    \texttt{36}  & \texttt{4} & \texttt{1.56250E-02} & \texttt{6.25000E-02} & \texttt{2.72967E-02} \\
    \texttt{37}  & \texttt{1} & \texttt{1.56250E-02} & \texttt{1.56250E-02} & \texttt{2.59271E-02} \\
    \texttt{39}  & \texttt{1} & \texttt{1.56250E-02} & \texttt{1.56250E-02} & \texttt{2.29097E-02} \\
    \texttt{40}  & \texttt{1} & \texttt{1.56250E-02} & \texttt{1.56250E-02} & \texttt{2.11794E-02} \\
    \bottomrule
\end{tabular}

    \caption{%
        Newcomb's measurements of the speed of light as deviations from 24,800 nanoseconds.
        The first two columns represent the unique measured values and the number of times they occur, respectively.
        The third column shows the \gls*{irw} unitary weights corresponding to the first column values.
        Then, the fourth column is the unit weight times the number of repetitions of each unique value.
        Finally, the last column gives the value of the \gls*{kde} underlying the \gls*{kme} at each observation.
    }
    \label{tab:c-measurements}
\end{table}

We can better assess the situation by looking at \figurename~\ref{fig:kme-case-study}.
In comparison with the synthetic data example in Section~\ref{sec:illustrative-example} and \figurename~\ref{fig:kme-synthetic-data}, the isolated local modes of $\kdesecondary$, corresponding to the two outliers, are so much less pronounced that they are almost entirely removed from $\kde$.
Therefore, the risk of getting trapped in any local maxima of the \gls*{kde} is minimal.
We also see that the original $\kdesecondary$ is already mostly symmetric about the sample median, nearly coinciding with $\symmetrickde$.
Moreover, the optimal rescaled kernel $\chaconserranokernel$, closely resembling the Epanechnikov, supports the Gaussianity hypothesis in~\citet{Gelman2013}, except for the two outliers.

\begin{figure}
    \centering
    \includegraphics[width=0.5\textwidth]{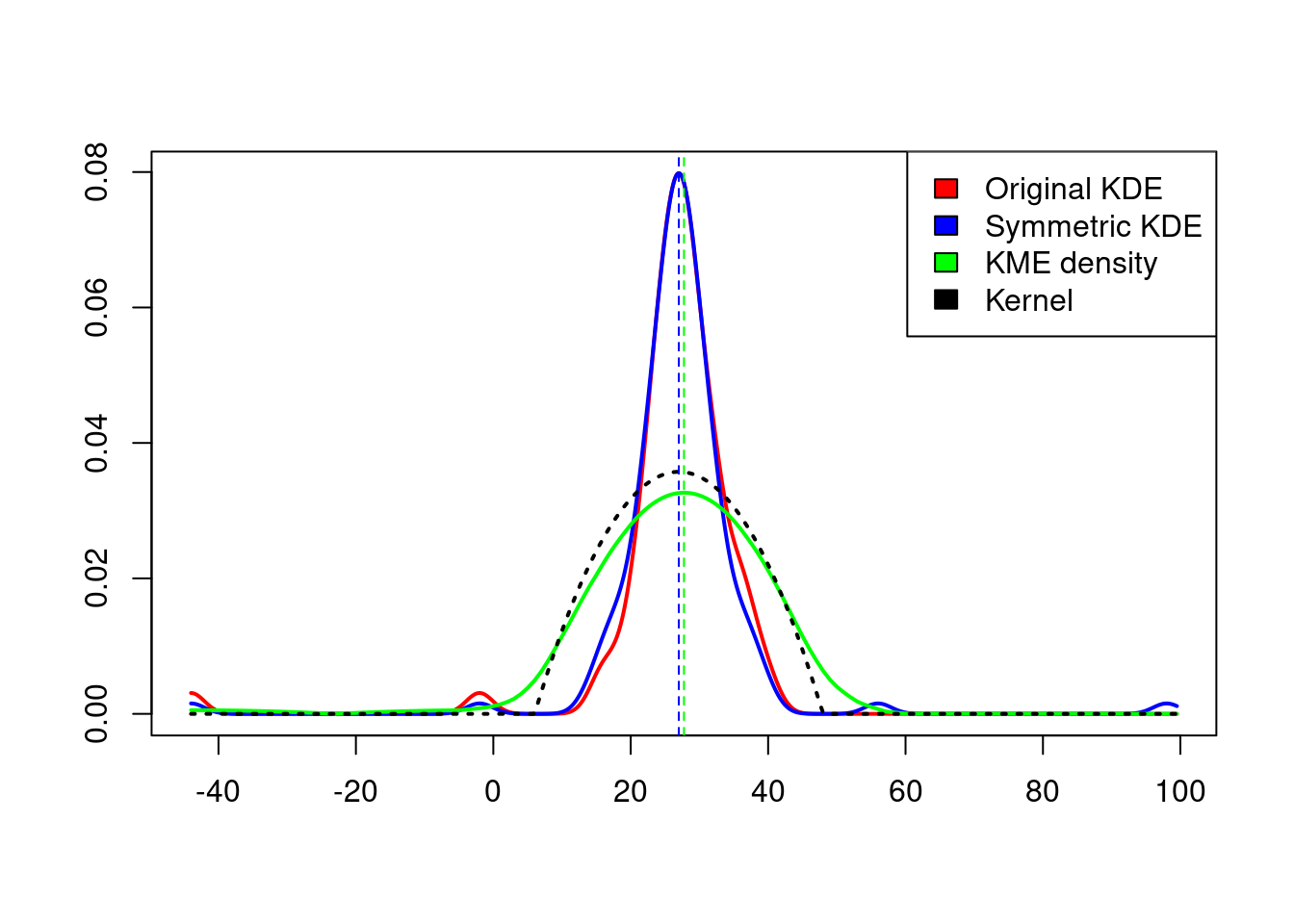}
    \caption{%
        Plot summary of all the relevant functions involved in \gls*{irw} for the case study.
        The structure is the same as in \figurename~\ref{fig:kme-synthetic-data}.
    }
    \label{fig:kme-case-study}
\end{figure}

Detecting and completely discarding the two outliers is an interesting feat of our \gls*{kme} proposal.
Choosing a trimming level $\trimlevel = 2/66$ in a (symmetrically) trimmed mean would also affect the upper tail observations $39$ and $40$, which, though extreme, are not outliers.
As a result, the trimmed mean would be $27.37$, which is lower than our \gls*{kme}.
In any case, despite the numerous tools for estimating the center of symmetry, the actual value of the speed of light known today is $33$, far from any of the estimates we have seen.
Therefore, as~\citeauthor{Gelman2013} remind us, data analysis can only be as good as the experiment that produces the data.

\section{Simulation study}

\label{sec:simulation-study}

This section demonstrates the practical effectiveness of \glspl*{kme} for estimating the center of symmetry, with particular attention paid to our proposal.
The latter comprises the parameter optimization in Section~\ref{sec:parameter-optimization} for the novel kernel family~\eqref{eq:chacon-serrano-kernel-derivative}, followed by a run of the \gls*{irw} algorithm in Section~\ref{sec:computation}.
We shall compare our \gls*{kme} with two classic M-estimators, the sample mean and median, and two redescending M-estimators, Tukey's \textit{biweight} and Andrew's \textit{sine}~\citep[p. 100]{Huber2009}.
For completeness, we also include in the study two classic L-estimators, the trimmed and winsorized means~\citep[Section 2.4]{Maronna2019}.

The two considered redescending M-estimators are also \glspl*{kme}.
The kernel corresponding to Tukey's biweight is the \textit{triweight}~\citep[p. 31]{Wand1995}
$
    \kernelatx
    \propto
    (1 - \x^2)^3 \cdot \indicatoroversupportatx
$,
while that of Andrew's sine is the \textit{raised cosine}
$
    \kernelatx
    \propto
    [1 + \cos(\pi \x)] \cdot \indicatoroversupportatx
$.
Both kernels and their respective associated functions are shown in \figurename~\ref{fig:redescending-m-estimators}.
The scale parameter (bandwidth) $\bandwidth$ for these M-estimators is usually tuned by assuming data from a \textit{contaminated} Gaussian distribution.
This leads to bandwidths that are \textit{prefixed} multiples of some robust scale estimate, contrary to the full data-driven optimization in Section~\ref{sec:parameter-optimization}.
Hence, these methods are considered \textit{non-adaptive}~\citep{Hogg1974}.
Specifically, letting $\madnvar = \madn$ be as in Section~\ref{sec:parameter-optimization}, we consider $\bandwidth = 6\madnvar$ for Tukey's biweight, and $\bandwidth = 2.1\pi\madnvar$ for Andrew's sine.
The factor $\madnvar$ is included because of the recommendation in~\citet[Section 2.8.1]{Maronna2019}~\citep[see also][]{Hogg1974}, whereas the ``magical'' constants $6$ and $2.1\pi$, taken from~\citet{Stigler1977}, were initially proposed by Tukey and Andrew themselves.
The same \gls*{irw} algorithm in Section~\ref{sec:computation} was used to compute both redescending M-estimators.

For the trimmed and winsorized means, users typically select one of the widespread values $\trimlevel \in \{0.1, 0.15, 0.25\}$ for the trimming level, as in~\citet{Stigler1977}.
Here, however, we shall attempt to optimize $\trimlevel \in [0, 1/2)$, allowing both L-estimators to range between a \textit{mean-like} ($\trimlevel = 0$) or a \textit{median-like} ($\trimlevel \lesssim 1/2$) behavior, depending on the data.
To do so, we implemented the straightforward \textit{bootstrap} variance-minimizing procedure by~\citet{Mehrotra1991} that yielded good results for trimmed means over finite samples.
Specifically, we employed their augmented sample strategy about the median (which they call Estimate 3), analogous to the one behind~\eqref{eq:symmetric-kde}.
These methods are considered \textit{adaptive} by~\citet{Hogg1974}, who praises the adaptive version of the trimmed mean for the symmetric case.

Simulated data will be drawn from each of the nine symmetric, unimodal \glspl*{pdf} about $\symmetrycenter = 0$ shown in \figurename~\ref{fig:symmetric-unimodal-test-beds}.
These are the following:
\begin{itemize}
    \item \texttt{normal}:
          Standard Gaussian \gls*{pdf} $\pdf(\x) = \standardgaussianpdf(\x) \propto e^{-\x^2 / 2}$.
    \item \texttt{logistic}:
          Logistic \gls*{pdf} $\pdf(\x) = (e^{\x / 2} + e^{-x / 2})^{-2}$.
    \item \texttt{laplace}:
          Laplace \gls*{pdf} $\pdf(\x) \propto e^{-\absvalof{\x}}$.
    \item \texttt{student\_t\_<$\studenttdegreesoffreedom$>} ($\studenttdegreesoffreedom \in \{1, \dots, 5\}$):
          Student's t \gls*{pdf}~\eqref{eq:student-t-pdf} with $\studenttdegreesoffreedom \in \{1, \dots, 5\}$.
    \item \texttt{outlier}:
          Gaussian mixture \gls*{pdf} $\pdf(\x) = (9/10)\standardgaussianpdf(\x) + (1/10)\standardgaussianpdf(\x/100)/100$.
\end{itemize}
\citet{Mehrotra1991} previously used the normal, logistic, and Laplace (double exponential) distributions as test-beds.
These authors also included Student's t with $\studenttdegreesoffreedom = 1$, the Cauchy distribution.
In turn, the so-called \textit{outlier} distribution is a variant of the classic homonym \gls*{pdf} in~\citet{Marron1992}, where the original $10\%$ component of the mixture was $\standardgaussianpdf(\x/10)/10$, i.e., a Gaussian with $\stddev = 10$ instead of our choice of $\stddev = 100$.

Compared to the normal distribution, the rest of the test-bed \glspl*{pdf} exhibit \textit{heavy-tailedness}.
The logistic, Laplace, and outlier distributions have a larger \textit{kurtosis}~\citep[p. 228]{Maronna2019} than the normal, meaning they have a moderate but more significant proportion of outliers.
Then, all instances from Student's t family are regularly varying, with tails that decay according to a \textit{power law}, which implies a slower rate than the exponential-like tails of the rest.
In particular, when $\studenttdegreesoffreedom = 1$, the tails are so heavy that even the \gls*{pdf} expectation is undefined.
The expectation does exist for $\studenttdegreesoffreedom = 2$, as well as the variance, but the latter is still infinite.

\begin{figure}
    \centering
    \begin{subfigure}[b]{\figurewidth}
        \centering
        \includegraphics[width=\subfigurewidth]{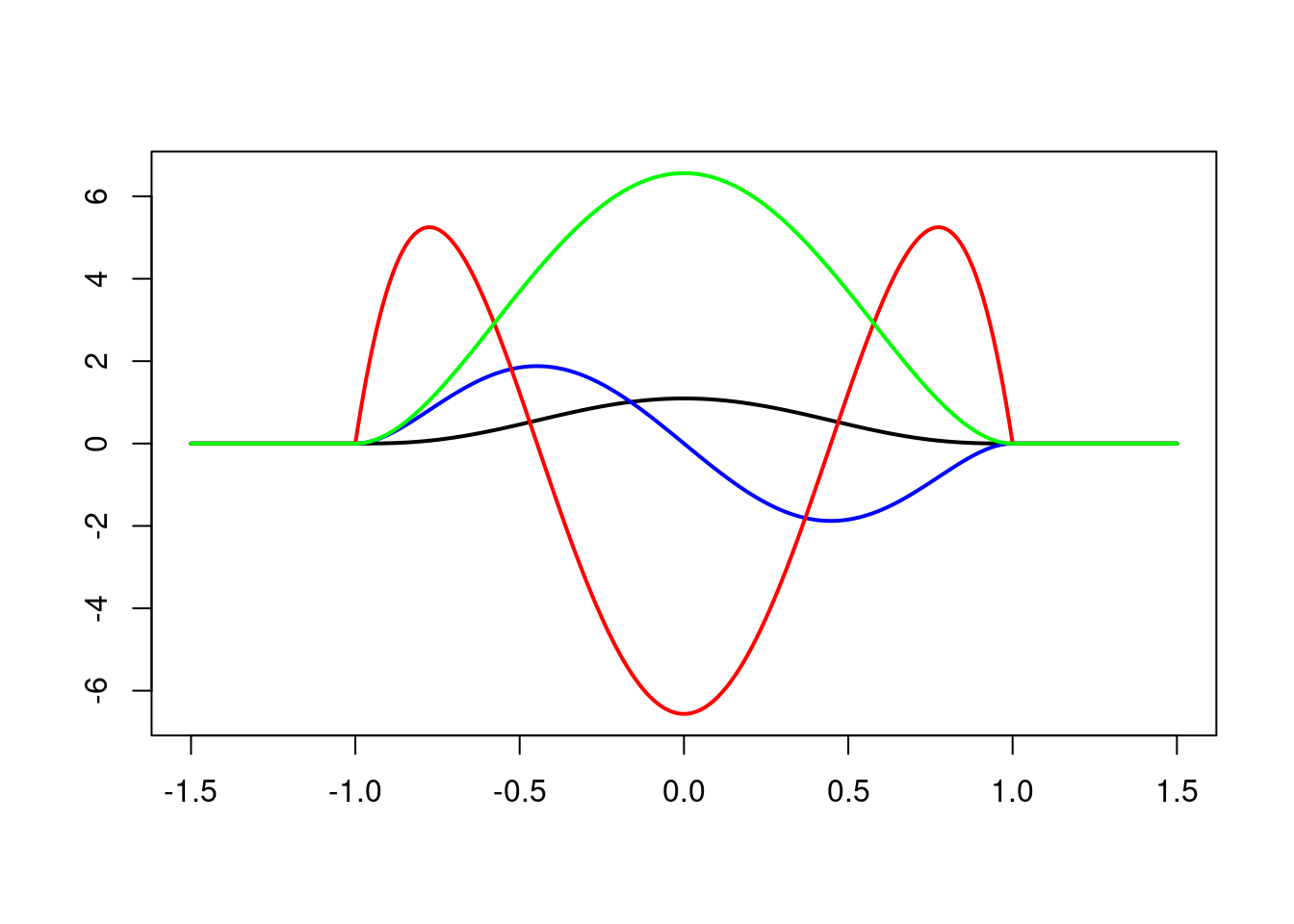}
        \caption{Triweight kernel}
    \end{subfigure}
    \begin{subfigure}[b]{\figurewidth}
        \centering
        \includegraphics[width=\subfigurewidth]{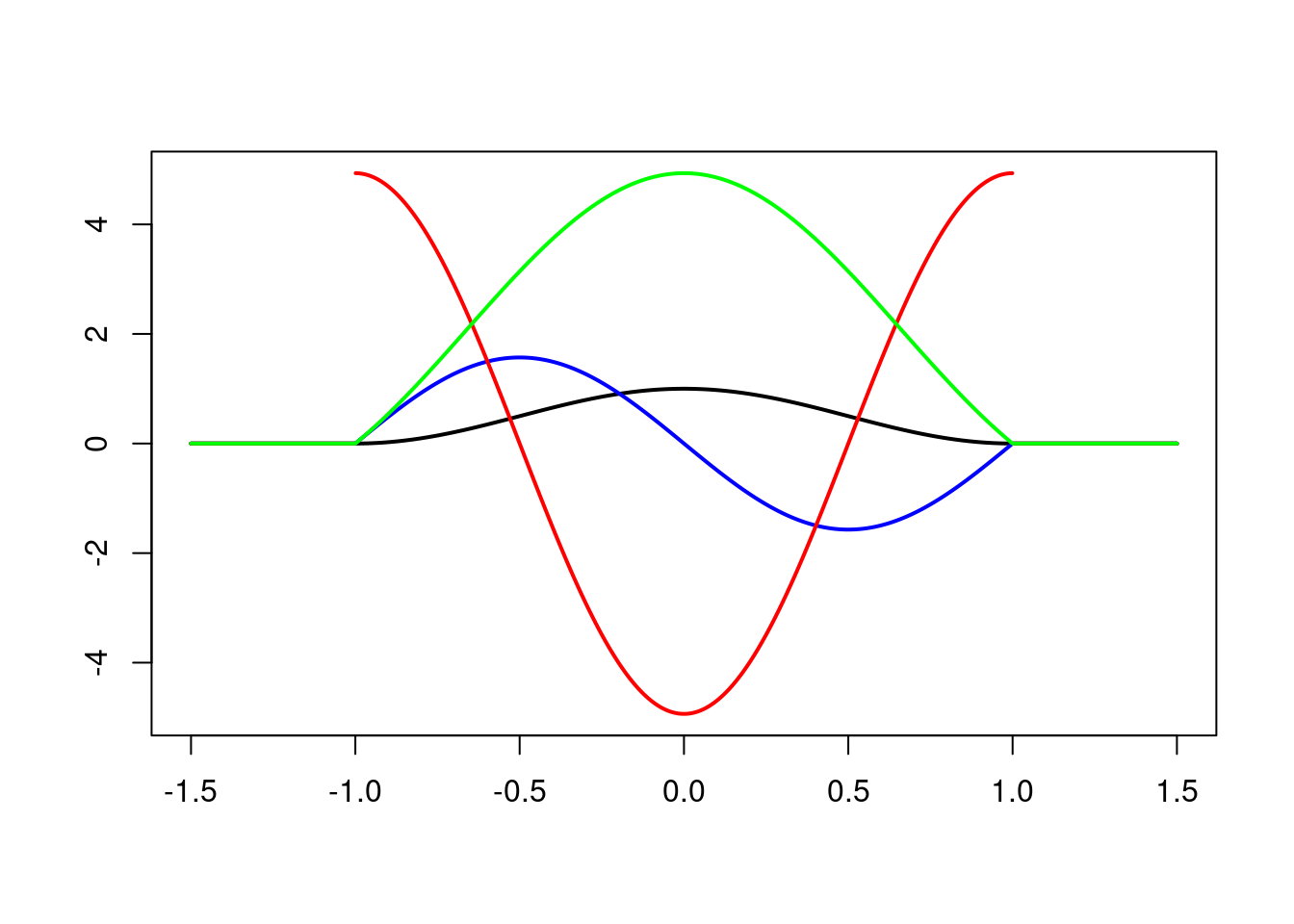}
        \caption{Raised cosine kernel}
    \end{subfigure}
    \caption{%
        Kernels and associated functions corresponding to the redescending M-estimators Tukey's biweight (left) and Andrew's sine (right) used in the simulation study.
        The plot structure is the same as in \figurename~\ref{fig:chacon-serrano-kernel-family}, excluding the vertical lines that marked the inflection points.
        The left and right sets of functions are similar, except for the second derivative (red), which is discontinuous at $\absvalof{\x} = 1$ for the raised cosine.
    }
    \label{fig:redescending-m-estimators}
\end{figure}

Each of the nine test-bed \glspl*{pdf} will be paired with three sample sizes, $\samplesize = 100$ (small), $\samplesize = 1,000$ (medium-sized), and $\samplesize = 10,000$ (large), giving rise to twenty-seven sampling configurations.
Then, $\numrepetitions = 1,000$ repetitions of each random experiment will be carried out to obtain significant results.
For each $\genindexi$-th replication and each contending estimator $\method$, we shall compute an estimate $\methodestimateofsymmetrycenter_{\genindexi}$ of the center of symmetry $\symmetrycenter = 0$.
Let us assume that $\methodone$ is our \gls*{kme} proposal while $\methodtwo$ is some other method.
Then, letting $\mseofmethod = \numrepetitions^{-1} \sum_{\genindexi = 1}^{\numrepetitions} (\methodestimateofsymmetrycenter_{\genindexi} - \symmetrycenter)^2$ be the \gls*{mse} of the estimator $\method$, a relative efficiency measure comparing two methods $\methodone$ and $\methodtwo$ is given by $\mserrorof{\methodone} / \mserrorof{\methodtwo}$.
Values less than one favor $\methodone$, while those greater than one give the advantage to $\methodtwo$.
On the other hand, for evaluating on a \textit{per-sample} basis, in addition to counting how many experiments satisfy
$
    \smallabsvalof{\methodoneestimateofsymmetrycenter_{\genindexi} - \symmetrycenter}
    <
    \smallabsvalof{\methodtwoestimateofsymmetrycenter_{\genindexi} - \symmetrycenter}
$,
a matched-samples Wilcoxon test can be performed to reject the \textit{null} hypothesis that the sequence
$
    (\smallabsvalof{\methodtwoestimateofsymmetrycenter_{\genindexi} - \symmetrycenter})_{\genindexi = 1}^{\numrepetitions}
$
has an equal or lower \textit{median} value than the sequence
$
    (\smallabsvalof{\methodoneestimateofsymmetrycenter_{\genindexi} - \symmetrycenter})_{\genindexi = 1}^{\numrepetitions}
$.
A low $p$-value will reject the null hypothesis in favor of the \textit{alternative} that our \gls*{kme} has a lower median value.

The results from the simulation study are split across three tables with the same structure.
Let us analyze each of them separately.
First, Table~\ref{tab:simulation-results-1} compares our \gls*{kme} with the two classic M-estimators: the sample mean and median.
The overall results largely favor our proposal, especially with the three most heavy-tailed test-beds: \texttt{student\_t\_1}, \texttt{student\_t\_2}, and \texttt{outlier}.
There are only two understandable exceptions, the \texttt{normal} and \texttt{laplace}, for which, as mentioned above, the sample mean and median are \glspl*{mle}, respectively.
Even so, the performance gap reduces for $\samplesize = 10,000$, especially with respect to the sample mean, which, though not significantly, is outperformed.

The results in Table~\ref{tab:simulation-results-2} are the least advantageous for our proposal, which nonetheless prevails.
The trimmed and winsorized means, especially the former, perform similarly to our \gls*{kme} in many settings.
For most Student's t \glspl*{pdf}, namely $\studenttdegreesoffreedom \in \{2, 3, 4\}$, our \gls*{kme} is neither significantly better nor worse than the modified means.
The same could be said about \texttt{logistic}.
Our method only performs consistently better with \texttt{student\_t\_1}, \texttt{student\_t\_5}, and \texttt{outlier}.
Surprisingly, the modified means do not behave like the sample mean with the test-bed $\texttt{normal}$, losing the advantage to our \gls*{kme} in small and medium-sized samples.
In turn, the trimmed mean performs especially well with \texttt{laplace} and $\samplesize < 10,000$, contrary to the winsorized mean.

Lastly, Table~\ref{tab:simulation-results-3} compares the redescending M-estimators.
Both Tukey's biweight and Andrew's sine have similar performances.
The results show that our \gls*{kme} is more versatile, outperforming the other two in most settings, namely, in all the \glspl*{pdf} of Student's t family and the \texttt{laplace}.
In particular, our \gls*{kme} outperforms the recommended method by~\citeauthor{Maronna2019} for the Cauchy distribution, which is the biweight~\citep[p. 65]{Maronna2019}.
Notwithstanding, the redescending M-estimators demonstrate some proficiency in those scenarios that are more akin to robust statistics, where our \gls*{kme} proposal shows a similar performance.
These cases include the \texttt{normal} and the similarly-shaped \texttt{logistic}, as well as the \texttt{outlier}, which is a \textit{contaminated} normal.

\begin{table}
    \centering
    \footnotesize
    \begin{tabular}{llllll}
    \toprule
    \multicolumn{1}{c}{ }  & \multicolumn{1}{c}{ } & \multicolumn{2}{c}{Mean} & \multicolumn{2}{c}{Median}                                                 \\
    \cmidrule(l{3pt}r{3pt}){3-4} \cmidrule(l{3pt}r{3pt}){5-6}
    Test-bed               & $n$                   & MSE                      & Paired samples             & MSE               & Paired samples            \\
    \midrule
    \texttt{student\_t\_1} & \texttt{100}          & \texttt{6.43E-06}        & \texttt{.925 (p < .001)}   & \texttt{9.10E-01} & \texttt{.539 (p = .0016)} \\

                           & \texttt{1000}         & \texttt{2.52E-06}        & \texttt{.972 (p < .001)}   & \texttt{7.98E-01} & \texttt{.581 (p < .001)}  \\

                           & \texttt{10000}        & \texttt{2.98E-08}        & \texttt{.988 (p < .001)}   & \texttt{8.52E-01} & \texttt{.545 (p < .001)}  \\
    \addlinespace
    \texttt{student\_t\_2} & \texttt{100}          & \texttt{1.55E-01}        & \texttt{.705 (p < .001)}   & \texttt{8.99E-01} & \texttt{.565 (p < .001)}  \\

                           & \texttt{1000}         & \texttt{1.57E-01}        & \texttt{.757 (p < .001)}   & \texttt{8.50E-01} & \texttt{.547 (p < .001)}  \\

                           & \texttt{10000}        & \texttt{1.10E-01}        & \texttt{.776 (p < .001)}   & \texttt{8.48E-01} & \texttt{.568 (p < .001)}  \\
    \addlinespace
    \texttt{student\_t\_3} & \texttt{100}          & \texttt{5.52E-01}        & \texttt{.610 (p < .001)}   & \texttt{8.69E-01} & \texttt{.564 (p < .001)}  \\

                           & \texttt{1000}         & \texttt{5.10E-01}        & \texttt{.659 (p < .001)}   & \texttt{8.28E-01} & \texttt{.563 (p < .001)}  \\

                           & \texttt{10000}        & \texttt{5.03E-01}        & \texttt{.649 (p < .001)}   & \texttt{8.31E-01} & \texttt{.567 (p < .001)}  \\
    \addlinespace
    \texttt{student\_t\_4} & \texttt{100}          & \texttt{8.19E-01}        & \texttt{.566 (p < .001)}   & \texttt{8.39E-01} & \texttt{.576 (p < .001)}  \\

                           & \texttt{1000}         & \texttt{7.16E-01}        & \texttt{.596 (p < .001)}   & \texttt{8.06E-01} & \texttt{.551 (p < .001)}  \\

                           & \texttt{10000}        & \texttt{7.05E-01}        & \texttt{.594 (p < .001)}   & \texttt{8.15E-01} & \texttt{.554 (p < .001)}  \\
    \addlinespace
    \texttt{student\_t\_5} & \texttt{100}          & \texttt{1.04E+00}        & \texttt{.536 (p = .0012)}  & \texttt{9.41E-01} & \texttt{.589 (p < .001)}  \\

                           & \texttt{1000}         & \texttt{7.68E-01}        & \texttt{.594 (p < .001)}   & \texttt{7.68E-01} & \texttt{.583 (p < .001)}  \\

                           & \texttt{10000}        & \texttt{7.82E-01}        & \texttt{.585 (p < .001)}   & \texttt{7.60E-01} & \texttt{.576 (p < .001)}  \\
    \addlinespace
    \texttt{logistic}      & \texttt{100}          & \texttt{1.00E+00}        & \texttt{.532 (p = .24)}    & \texttt{8.52E-01} & \texttt{.574 (p < .001)}  \\

                           & \texttt{1000}         & \texttt{9.21E-01}        & \texttt{.555 (p < .001)}   & \texttt{7.37E-01} & \texttt{.597 (p < .001)}  \\

                           & \texttt{10000}        & \texttt{9.27E-01}        & \texttt{.534 (p = .0033)}  & \texttt{7.13E-01} & \texttt{.560 (p < .001)}  \\
    \addlinespace
    \texttt{outlier}       & \texttt{100}          & \texttt{1.34E-03}        & \texttt{.972 (p < .001)}   & \texttt{6.44E-01} & \texttt{.641 (p < .001)}  \\

                           & \texttt{1000}         & \texttt{1.13E-03}        & \texttt{.975 (p < .001)}   & \texttt{5.60E-01} & \texttt{.640 (p < .001)}  \\

                           & \texttt{10000}        & \texttt{1.09E-03}        & \texttt{.986 (p < .001)}   & \texttt{5.68E-01} & \texttt{.647 (p < .001)}  \\
    \addlinespace
    \texttt{normal}        & \texttt{100}          & \texttt{1.67E+00}        & \texttt{.450 (p = 1)}      & \texttt{1.08E+00} & \texttt{.612 (p < .001)}  \\

                           & \texttt{1000}         & \texttt{1.01E+00}        & \texttt{.487 (p = .85)}    & \texttt{6.68E-01} & \texttt{.601 (p < .001)}  \\

                           & \texttt{10000}        & \texttt{9.99E-01}        & \texttt{.514 (p = .06)}    & \texttt{6.20E-01} & \texttt{.617 (p < .001)}  \\
    \addlinespace
    \texttt{laplace}       & \texttt{100}          & \texttt{6.42E-01}        & \texttt{.618 (p < .001)}   & \texttt{1.12E+00} & \texttt{.422 (p = 1)}     \\

                           & \texttt{1000}         & \texttt{5.50E-01}        & \texttt{.635 (p < .001)}   & \texttt{1.05E+00} & \texttt{.461 (p = .99)}   \\

                           & \texttt{10000}        & \texttt{5.17E-01}        & \texttt{.632 (p < .001)}   & \texttt{1.03E+00} & \texttt{.486 (p = .97)}   \\
    \bottomrule
\end{tabular}

    \caption{%
        Part one of the simulation study results, comparing our \gls*{kme} with the sample mean and median.
        The columns \gls*{mse} represent the ratio $\mserrorof{\methodone} / \mserrorof{\methodtwo}$, where $\methodone$ is our \gls*{kme} proposal, and $\methodtwo$ is either the sample mean or the sample median.
        The columns ``Paired samples'' represent a comparison on a \textit{per-sample} basis.
        The fractional number is the proportion of the $\numrepetitions$ replications of the random experiment in which
        $
            \smallabsvalof{\methodoneestimateofsymmetrycenter_{\genindexi} - \symmetrycenter}
            <
            \smallabsvalof{\methodtwoestimateofsymmetrycenter_{\genindexi} - \symmetrycenter}
        $.
        The $p$-value in parentheses derives from a Wilcoxon paired-samples test with the alternative hypothesis that our \gls*{kme} has a lower \textit{median} value of
        $
            \smallabsvalof{\methodestimateofsymmetrycenter_{\genindexi} - \symmetrycenter}
        $.
    }
    \label{tab:simulation-results-1}
\end{table}

\begin{table}
    \centering
    \footnotesize
    \begin{tabular}{llllll}
    \toprule
    \multicolumn{1}{c}{ }  & \multicolumn{1}{c}{ } & \multicolumn{2}{c}{Trimmed mean} & \multicolumn{2}{c}{Winsorized mean}                                                 \\
    \cmidrule(l{3pt}r{3pt}){3-4} \cmidrule(l{3pt}r{3pt}){5-6}
    Test-bed               & $n$                   & MSE                              & Paired samples                      & MSE               & Paired samples            \\
    \midrule
    \texttt{student\_t\_1} & \texttt{100}          & \texttt{9.03E-01}                & \texttt{.536 (p = .001)}            & \texttt{8.55E-01} & \texttt{.550 (p < .001)}  \\

                           & \texttt{1000}         & \texttt{8.65E-01}                & \texttt{.546 (p < .001)}            & \texttt{8.45E-01} & \texttt{.553 (p < .001)}  \\

                           & \texttt{10000}        & \texttt{9.02E-01}                & \texttt{.518 (p = .01)}             & \texttt{8.44E-01} & \texttt{.548 (p < .001)}  \\
    \addlinespace
    \texttt{student\_t\_2} & \texttt{100}          & \texttt{9.88E-01}                & \texttt{.537 (p = .02)}             & \texttt{9.74E-01} & \texttt{.518 (p = .04)}   \\

                           & \texttt{1000}         & \texttt{9.72E-01}                & \texttt{.512 (p = .43)}             & \texttt{9.44E-01} & \texttt{.516 (p = .08)}   \\

                           & \texttt{10000}        & \texttt{9.61E-01}                & \texttt{.497 (p = .04)}             & \texttt{9.24E-01} & \texttt{.549 (p < .001)}  \\
    \addlinespace
    \texttt{student\_t\_3} & \texttt{100}          & \texttt{1.03E+00}                & \texttt{.497 (p = .65)}             & \texttt{1.04E+00} & \texttt{.483 (p = .92)}   \\

                           & \texttt{1000}         & \texttt{1.01E+00}                & \texttt{.515 (p = .14)}             & \texttt{9.86E-01} & \texttt{.508 (p = .15)}   \\

                           & \texttt{10000}        & \texttt{9.81E-01}                & \texttt{.524 (p = .08)}             & \texttt{9.56E-01} & \texttt{.535 (p = .0093)} \\
    \addlinespace
    \texttt{student\_t\_4} & \texttt{100}          & \texttt{9.93E-01}                & \texttt{.493 (p = .78)}             & \texttt{1.01E+00} & \texttt{.495 (p = .83)}   \\

                           & \texttt{1000}         & \texttt{9.80E-01}                & \texttt{.499 (p = .41)}             & \texttt{9.78E-01} & \texttt{.507 (p = .35)}   \\

                           & \texttt{10000}        & \texttt{9.94E-01}                & \texttt{.484 (p = .6)}              & \texttt{9.72E-01} & \texttt{.510 (p = .07)}   \\
    \addlinespace
    \texttt{student\_t\_5} & \texttt{100}          & \texttt{1.13E+00}                & \texttt{.520 (p = .05)}             & \texttt{1.19E+00} & \texttt{.502 (p = .53)}   \\

                           & \texttt{1000}         & \texttt{9.57E-01}                & \texttt{.525 (p = .01)}             & \texttt{9.46E-01} & \texttt{.554 (p < .001)}  \\

                           & \texttt{10000}        & \texttt{9.67E-01}                & \texttt{.544 (p < .001)}            & \texttt{9.35E-01} & \texttt{.545 (p < .001)}  \\
    \addlinespace
    \texttt{logistic}      & \texttt{100}          & \texttt{1.01E+00}                & \texttt{.491 (p = .77)}             & \texttt{1.04E+00} & \texttt{.485 (p = .95)}   \\

                           & \texttt{1000}         & \texttt{9.86E-01}                & \texttt{.517 (p = .16)}             & \texttt{9.54E-01} & \texttt{.518 (p = .02)}   \\

                           & \texttt{10000}        & \texttt{1.00E+00}                & \texttt{.498 (p = .5)}              & \texttt{9.69E-01} & \texttt{.505 (p = .29)}   \\
    \addlinespace
    \texttt{outlier}       & \texttt{100}          & \texttt{7.77E-01}                & \texttt{.597 (p < .001)}            & \texttt{8.05E-01} & \texttt{.601 (p < .001)}  \\

                           & \texttt{1000}         & \texttt{7.33E-01}                & \texttt{.599 (p < .001)}            & \texttt{7.05E-01} & \texttt{.607 (p < .001)}  \\

                           & \texttt{10000}        & \texttt{7.52E-01}                & \texttt{.591 (p < .001)}            & \texttt{7.35E-01} & \texttt{.596 (p < .001)}  \\
    \addlinespace
    \texttt{normal}        & \texttt{100}          & \texttt{1.43E+00}                & \texttt{.516 (p = .0024)}           & \texttt{1.45E+00} & \texttt{.559 (p < .001)}  \\

                           & \texttt{1000}         & \texttt{9.75E-01}                & \texttt{.538 (p = .0023)}           & \texttt{9.19E-01} & \texttt{.534 (p < .001)}  \\

                           & \texttt{10000}        & \texttt{9.99E-01}                & \texttt{.508 (p = .7)}              & \texttt{9.48E-01} & \texttt{.511 (p = .01)}   \\
    \addlinespace
    \texttt{laplace}       & \texttt{100}          & \texttt{1.11E+00}                & \texttt{.438 (p = 1)}               & \texttt{9.94E-01} & \texttt{.511 (p = .27)}   \\

                           & \texttt{1000}         & \texttt{1.01E+00}                & \texttt{.455 (p = .98)}             & \texttt{8.71E-01} & \texttt{.589 (p < .001)}  \\

                           & \texttt{10000}        & \texttt{1.01E+00}                & \texttt{.503 (p = .52)}             & \texttt{8.91E-01} & \texttt{.558 (p < .001)}  \\
    \bottomrule
\end{tabular}

    \caption{%
        Part two of the simulation study results, comparing our \gls*{kme} with the trimmed and winsorized means, with the same structure as Table~\ref{tab:simulation-results-1}.
    }
    \label{tab:simulation-results-2}
\end{table}

\begin{table}
    \centering
    \footnotesize
    \begin{tabular}{llllll}
    \toprule
    \multicolumn{1}{c}{ }  & \multicolumn{1}{c}{ } & \multicolumn{2}{c}{Tukey's biweight} & \multicolumn{2}{c}{Andrew's sine}                                                 \\
    \cmidrule(l{3pt}r{3pt}){3-4} \cmidrule(l{3pt}r{3pt}){5-6}
    Test-bed               & $n$                   & MSE                                  & Paired samples                    & MSE               & Paired samples            \\
    \midrule
    \texttt{student\_t\_1} & \texttt{100}          & \texttt{6.90E-01}                    & \texttt{.614 (p < .001)}          & \texttt{5.91E-01} & \texttt{.626 (p < .001)}  \\

                           & \texttt{1000}         & \texttt{6.12E-01}                    & \texttt{.640 (p < .001)}          & \texttt{5.23E-01} & \texttt{.648 (p < .001)}  \\

                           & \texttt{10000}        & \texttt{6.32E-01}                    & \texttt{.626 (p < .001)}          & \texttt{5.48E-01} & \texttt{.636 (p < .001)}  \\
    \addlinespace
    \texttt{student\_t\_2} & \texttt{100}          & \texttt{9.28E-01}                    & \texttt{.544 (p = .0019)}         & \texttt{8.56E-01} & \texttt{.543 (p < .001)}  \\

                           & \texttt{1000}         & \texttt{8.85E-01}                    & \texttt{.556 (p < .001)}          & \texttt{8.26E-01} & \texttt{.581 (p < .001)}  \\

                           & \texttt{10000}        & \texttt{8.41E-01}                    & \texttt{.580 (p < .001)}          & \texttt{7.85E-01} & \texttt{.585 (p < .001)}  \\
    \addlinespace
    \texttt{student\_t\_3} & \texttt{100}          & \texttt{9.88E-01}                    & \texttt{.518 (p = .13)}           & \texttt{9.44E-01} & \texttt{.551 (p < .001)}  \\

                           & \texttt{1000}         & \texttt{9.26E-01}                    & \texttt{.557 (p < .001)}          & \texttt{8.89E-01} & \texttt{.555 (p < .001)}  \\

                           & \texttt{10000}        & \texttt{9.07E-01}                    & \texttt{.576 (p < .001)}          & \texttt{8.67E-01} & \texttt{.583 (p < .001)}  \\
    \addlinespace
    \texttt{student\_t\_4} & \texttt{100}          & \texttt{1.02E+00}                    & \texttt{.469 (p = .98)}           & \texttt{1.00E+00} & \texttt{.499 (p = .73)}   \\

                           & \texttt{1000}         & \texttt{9.63E-01}                    & \texttt{.551 (p < .001)}          & \texttt{9.27E-01} & \texttt{.552 (p < .001)}  \\

                           & \texttt{10000}        & \texttt{9.54E-01}                    & \texttt{.530 (p < .001)}          & \texttt{9.23E-01} & \texttt{.534 (p < .001)}  \\
    \addlinespace
    \texttt{student\_t\_5} & \texttt{100}          & \texttt{1.22E+00}                    & \texttt{.471 (p = .9)}            & \texttt{1.20E+00} & \texttt{.502 (p = .57)}   \\

                           & \texttt{1000}         & \texttt{9.87E-01}                    & \texttt{.517 (p = .06)}           & \texttt{9.58E-01} & \texttt{.546 (p = .0025)} \\

                           & \texttt{10000}        & \texttt{9.78E-01}                    & \texttt{.534 (p = .0089)}         & \texttt{9.54E-01} & \texttt{.538 (p = .0028)} \\
    \addlinespace
    \texttt{logistic}      & \texttt{100}          & \texttt{1.05E+00}                    & \texttt{.469 (p = .99)}           & \texttt{1.05E+00} & \texttt{.464 (p = .98)}   \\

                           & \texttt{1000}         & \texttt{9.93E-01}                    & \texttt{.506 (p = .1)}            & \texttt{9.85E-01} & \texttt{.536 (p = .02)}   \\

                           & \texttt{10000}        & \texttt{9.93E-01}                    & \texttt{.507 (p = .2)}            & \texttt{9.87E-01} & \texttt{.508 (p = .21)}   \\
    \addlinespace
    \texttt{outlier}       & \texttt{100}          & \texttt{1.10E+00}                    & \texttt{.530 (p = .03)}           & \texttt{1.11E+00} & \texttt{.514 (p = .2)}    \\

                           & \texttt{1000}         & \texttt{9.92E-01}                    & \texttt{.525 (p = .06)}           & \texttt{9.84E-01} & \texttt{.530 (p = .01)}   \\

                           & \texttt{10000}        & \texttt{1.00E+00}                    & \texttt{.492 (p = .59)}           & \texttt{1.01E+00} & \texttt{.489 (p = .54)}   \\
    \addlinespace
    \texttt{normal}        & \texttt{100}          & \texttt{1.64E+00}                    & \texttt{.520 (p = .18)}           & \texttt{1.66E+00} & \texttt{.507 (p = .55)}   \\

                           & \texttt{1000}         & \texttt{9.90E-01}                    & \texttt{.512 (p = .05)}           & \texttt{1.00E+00} & \texttt{.503 (p = .32)}   \\

                           & \texttt{10000}        & \texttt{9.92E-01}                    & \texttt{.492 (p = .36)}           & \texttt{9.99E-01} & \texttt{.486 (p = .56)}   \\
    \addlinespace
    \texttt{laplace}       & \texttt{100}          & \texttt{8.42E-01}                    & \texttt{.553 (p < .001)}          & \texttt{8.10E-01} & \texttt{.576 (p < .001)}  \\

                           & \texttt{1000}         & \texttt{7.30E-01}                    & \texttt{.613 (p < .001)}          & \texttt{6.93E-01} & \texttt{.613 (p < .001)}  \\

                           & \texttt{10000}        & \texttt{7.01E-01}                    & \texttt{.589 (p < .001)}          & \texttt{6.67E-01} & \texttt{.593 (p < .001)}  \\
    \bottomrule
\end{tabular}

    \caption{%
        Part three of the simulation study results, comparing our \gls*{kme} with the redescending M-estimators, with the same structure as Table~\ref{tab:simulation-results-1}.
    }
    \label{tab:simulation-results-3}
\end{table}

Apart from the simulation study results, \textit{debugging} the outputs of the parameter optimization process of Section~\ref{sec:parameter-optimization} is very instructive.
The underlying figures can be checked in Table~\ref{tab:simulation-study-parameter-optimization}, but we shall focus our comments on \figurename~\ref{fig:simulation-study-weight-functions}.
There, we can see the weight function~\eqref{eq:weight-function} for $\kernel \equiv \chaconserranokernel$, and $\parameters = \optimalparameters$.
The optimal parameters $\optimalparameters$ are obtained in \figurename~\ref{fig:simulation-study-weight-functions-true-pdfs} considering the true centered \gls*{pdf} $\untranslatedpdf$.
The ``optimal parameters'' in the remaining subfigures are an average of the $\numrepetitions$ optimal parameter vectors obtained considering an $\samplesize$-size sample estimate $\estimateofuntranslatedpdf$.

\begin{table}
    \centering
    \footnotesize
    \begin{tabular}{llllll}
    \toprule
    Test-bed               & Optimal $\beta$   & Optimal $h$       & $n$            & Mean $\beta$      & Mean $h$          \\
    \midrule
    \texttt{student\_t\_1} & \texttt{9.69E-02} & \texttt{3.04E+01} & \texttt{100}   & \texttt{1.93E+00} & \texttt{1.91E+01} \\

                           &                   &                   & \texttt{1000}  & \texttt{2.19E-01} & \texttt{2.25E+01} \\

                           &                   &                   & \texttt{10000} & \texttt{1.27E-01} & \texttt{2.47E+01} \\
    \addlinespace
    \texttt{student\_t\_2} & \texttt{1.57E-01} & \texttt{2.16E+01} & \texttt{100}   & \texttt{5.30E+00} & \texttt{5.43E+01} \\

                           &                   &                   & \texttt{1000}  & \texttt{2.90E-01} & \texttt{2.23E+01} \\

                           &                   &                   & \texttt{10000} & \texttt{1.87E-01} & \texttt{2.01E+01} \\
    \addlinespace
    \texttt{student\_t\_3} & \texttt{2.23E-01} & \texttt{1.73E+01} & \texttt{100}   & \texttt{8.23E+00} & \texttt{1.25E+02} \\

                           &                   &                   & \texttt{1000}  & \texttt{3.89E-01} & \texttt{2.07E+01} \\

                           &                   &                   & \texttt{10000} & \texttt{2.63E-01} & \texttt{1.64E+01} \\
    \addlinespace
    \texttt{student\_t\_4} & \texttt{2.91E-01} & \texttt{1.49E+01} & \texttt{100}   & \texttt{8.86E+00} & \texttt{1.86E+02} \\

                           &                   &                   & \texttt{1000}  & \texttt{4.74E-01} & \texttt{2.47E+01} \\

                           &                   &                   & \texttt{10000} & \texttt{3.34E-01} & \texttt{1.44E+01} \\
    \addlinespace
    \texttt{student\_t\_5} & \texttt{3.60E-01} & \texttt{1.33E+01} & \texttt{100}   & \texttt{7.97E+00} & \texttt{2.51E+02} \\

                           &                   &                   & \texttt{1000}  & \texttt{5.55E-01} & \texttt{2.20E+01} \\

                           &                   &                   & \texttt{10000} & \texttt{4.05E-01} & \texttt{1.33E+01} \\
    \addlinespace
    \texttt{logistic}      & \texttt{3.54E-01} & \texttt{2.54E+01} & \texttt{100}   & \texttt{9.14E+00} & \texttt{5.57E+02} \\

                           &                   &                   & \texttt{1000}  & \texttt{8.53E-01} & \texttt{2.70E+02} \\

                           &                   &                   & \texttt{10000} & \texttt{4.27E-01} & \texttt{2.69E+01} \\
    \addlinespace
    \texttt{outlier}       & \texttt{6.34E+00} & \texttt{4.31E+00} & \texttt{100}   & \texttt{1.39E+01} & \texttt{1.75E+01} \\

                           &                   &                   & \texttt{1000}  & \texttt{3.45E+01} & \texttt{5.19E+00} \\

                           &                   &                   & \texttt{10000} & \texttt{1.88E+01} & \texttt{4.55E+00} \\
    \addlinespace
    \texttt{normal}        & \texttt{5.34E+00} & \texttt{2.29E+01} & \texttt{100}   & \texttt{5.62E+00} & \texttt{2.13E+02} \\

                           &                   &                   & \texttt{1000}  & \texttt{5.45E+00} & \texttt{6.00E+01} \\

                           &                   &                   & \texttt{10000} & \texttt{4.72E+00} & \texttt{1.94E+01} \\
    \addlinespace
    \texttt{laplace}       & \texttt{7.07E-02} & \texttt{3.53E+01} & \texttt{100}   & \texttt{3.58E+00} & \texttt{3.31E+02} \\

                           &                   &                   & \texttt{1000}  & \texttt{5.50E-02} & \texttt{3.44E+02} \\

                           &                   &                   & \texttt{10000} & \texttt{3.33E-02} & \texttt{9.99E+02} \\
    \bottomrule
\end{tabular}

    \caption{%
        Optimal parameters $\parameters$ for each sampling configuration of test-bed and sample size $\samplesize$.
        The ``Optimal'' columns represent the optimal values of $\chaconserranokernelparam$ and $\bandwidth$ using the true $\untranslatedpdf$ in the optimization process.
        The ``Mean'' columns are averages of the optimal values of $\chaconserranokernelparam$ and $\bandwidth$ obtained for each random sample using the estimate $\estimateofuntranslatedpdf$ in the optimization.
    }
    \label{tab:simulation-study-parameter-optimization}
\end{table}

\begin{figure}
    \centering
    \begin{subfigure}[b]{\figurewidth}
        \centering
        \includegraphics[width=\subfigurewidth]{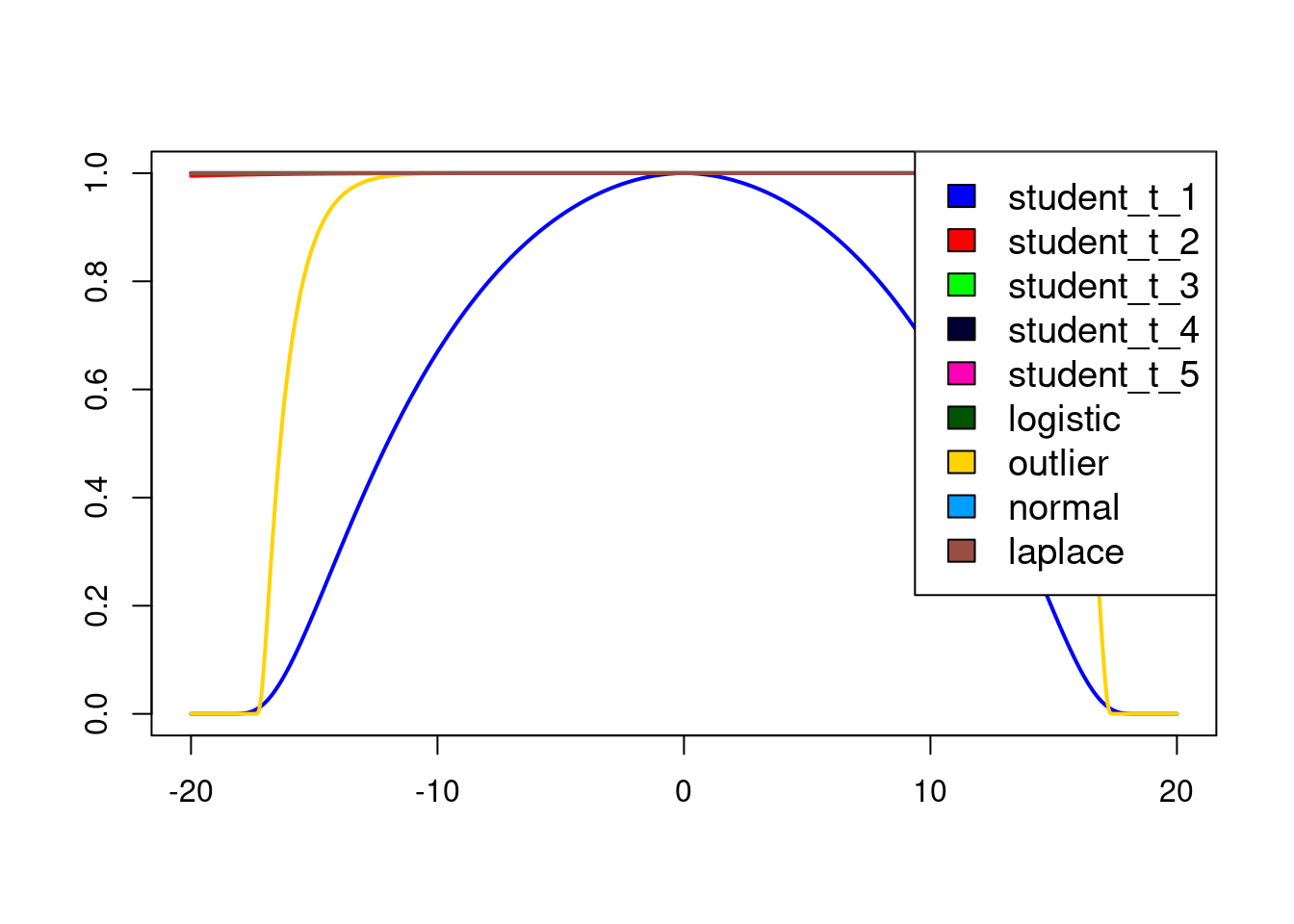}
        \caption{$\samplesize = 100$}
        \label{fig:simulation-study-weight-functions-100}
    \end{subfigure}
    \begin{subfigure}[b]{\figurewidth}
        \centering
        \includegraphics[width=\subfigurewidth]{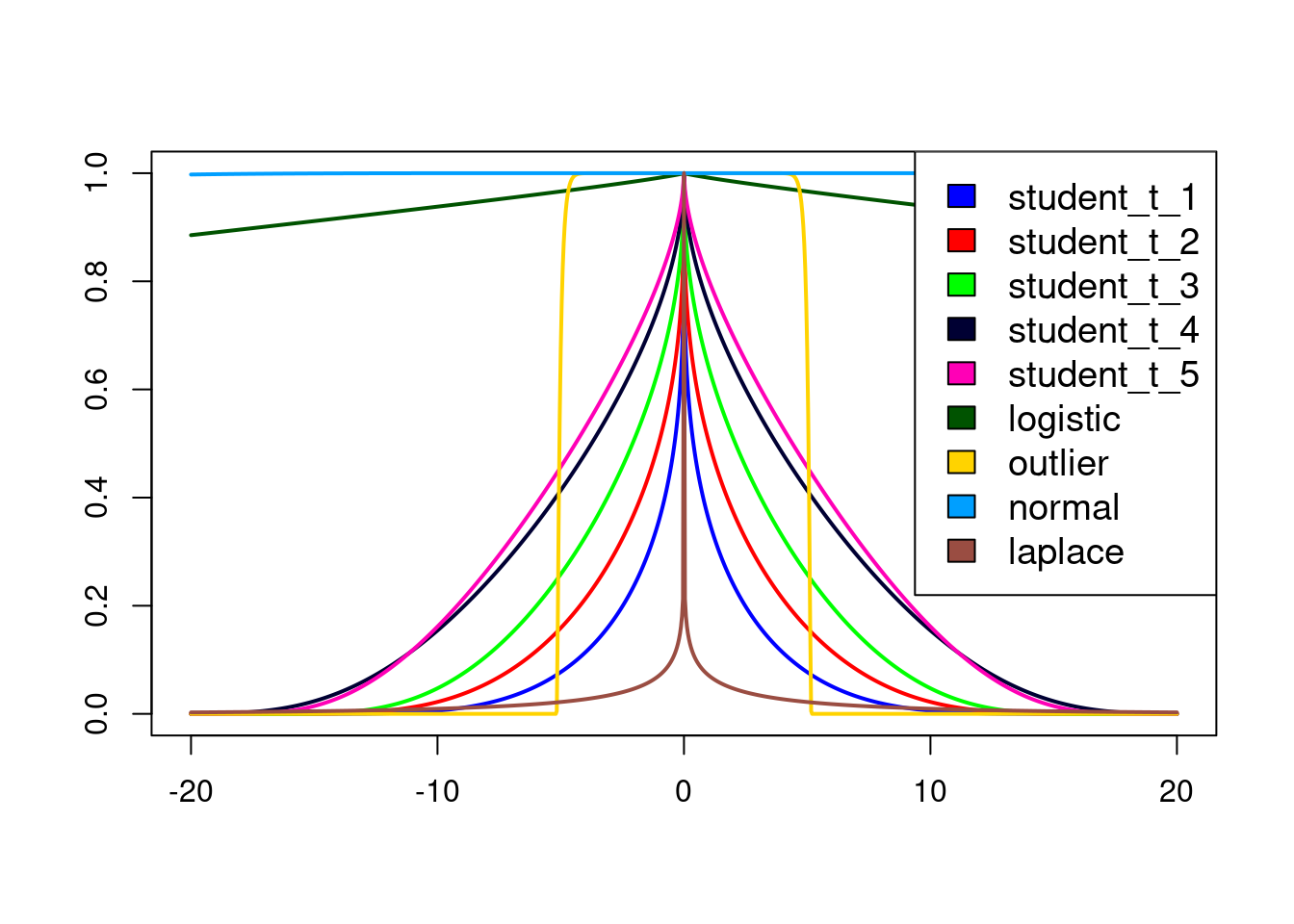}
        \caption{$\samplesize = 1,000$}
        \label{fig:simulation-study-weight-functions-1000}
    \end{subfigure}
    \begin{subfigure}[b]{\figurewidth}
        \centering
        \includegraphics[width=\subfigurewidth]{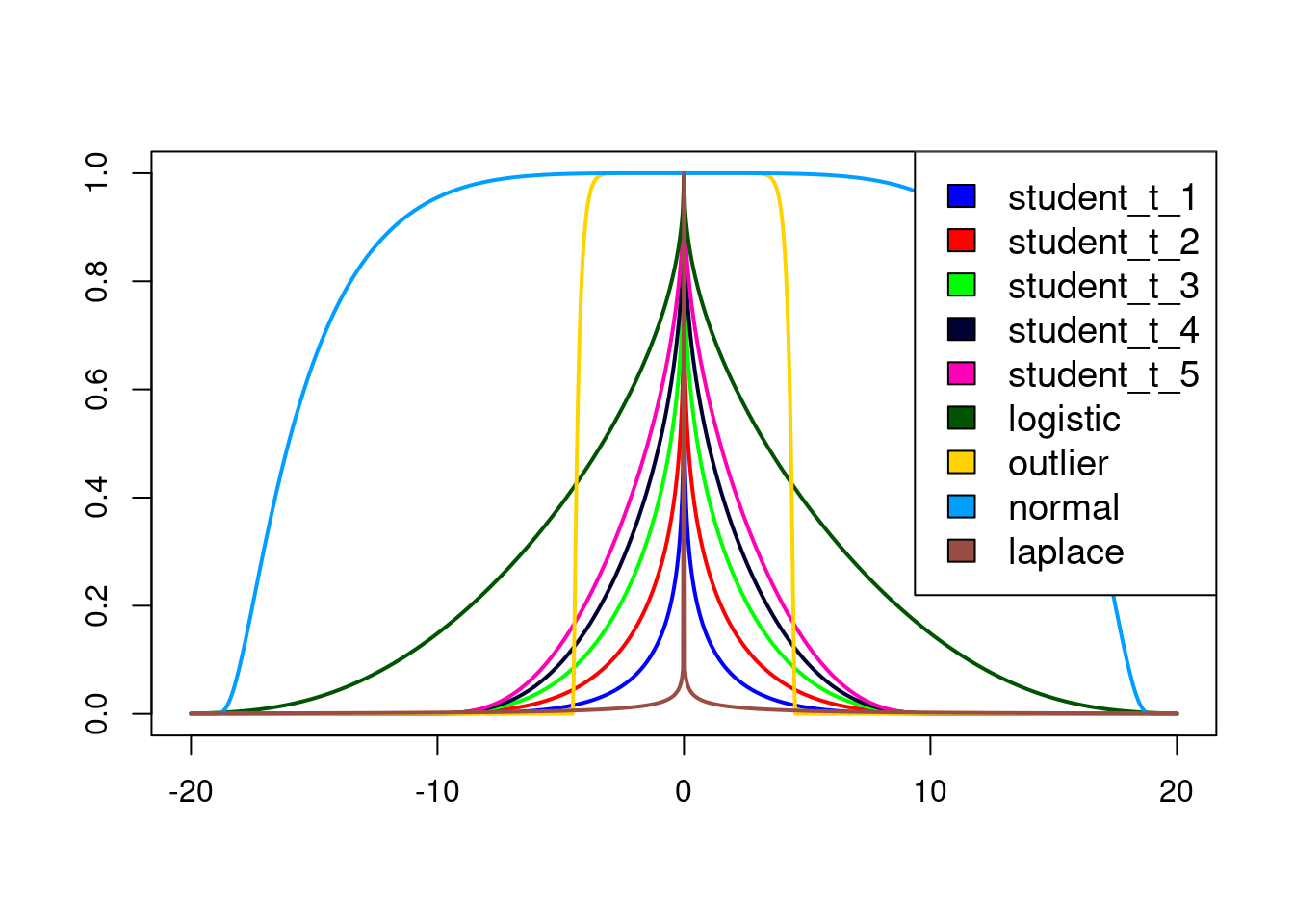}
        \caption{$\samplesize = 10,000$}
        \label{fig:simulation-study-weight-functions-10000}
    \end{subfigure}
    \begin{subfigure}[b]{\figurewidth}
        \centering
        \includegraphics[width=\subfigurewidth]{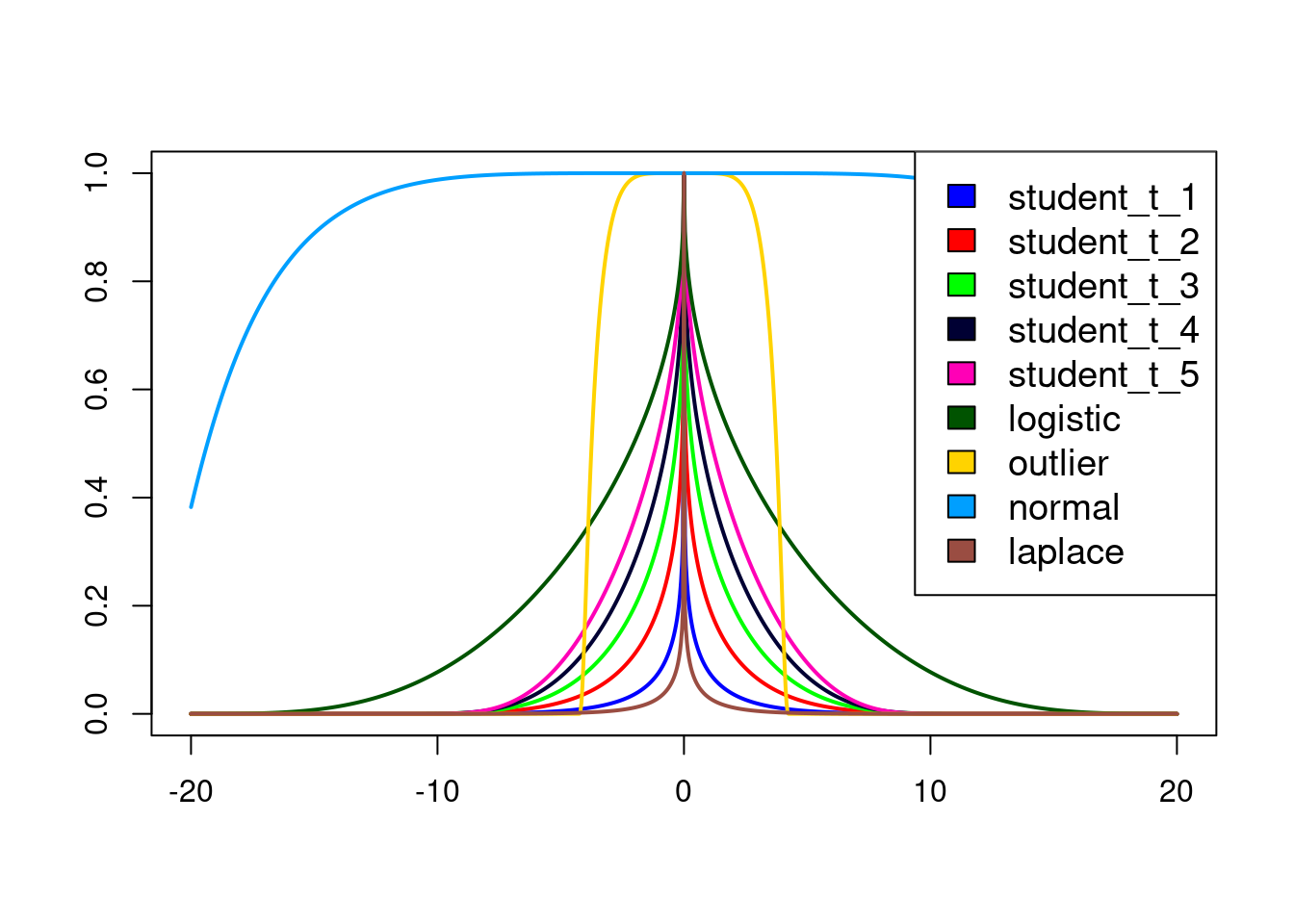}
        \caption{True \glspl*{pdf}}
        \label{fig:simulation-study-weight-functions-true-pdfs}
    \end{subfigure}
    \caption{%
        \gls*{irw} weight functions~\eqref{eq:weight-function}.
        For the bottom right subfigure, the underlying parameters $\parameters$ are the optimal ones $\optimalparameters$, considering the true centered \gls*{pdf} $\untranslatedpdf$.
        For the remaining subfigures, we computed an average of the optimal parameters for each $\samplesize$-size sample, considering the estimate $\estimateofuntranslatedpdf$.
    }
    \label{fig:simulation-study-weight-functions}
\end{figure}

The first aspect worth noticing in \figurename~\ref{fig:simulation-study-weight-functions} is the natural progression of the weight function estimates as $\samplesize$ grows, from \figurename~\ref{fig:simulation-study-weight-functions-100} to \figurename~\ref{fig:simulation-study-weight-functions-10000}, towards the true optimal weight functions in \figurename~\ref{fig:simulation-study-weight-functions-true-pdfs}.
In \figurename~\ref{fig:simulation-study-weight-functions-100}, our \gls*{kme} produces nearly \textit{flat} weight functions in most sampling configurations, effectively behaving like a sample mean.
Only the three most heavy-tailed test-beds (\texttt{student\_t\_1}, \texttt{student\_t\_2}, and \texttt{outlier}) make the estimates slightly \textit{bend}.
The situation starts changing in \figurename~\ref{fig:simulation-study-weight-functions-1000} and finally almost settles in \figurename~\ref{fig:simulation-study-weight-functions-10000}.
The second remarkable point about \figurename~\ref{fig:simulation-study-weight-functions} is that, except for the \texttt{normal} and the \texttt{outlier}, all test-beds \textit{prefer} a rapid weight decay ($\chaconserranokernelparam < 1$), being the most extreme case that of the \texttt{laplace}.

Last but not least, the optimal values $\chaconserranokernelparam$ in Table~\ref{tab:simulation-study-parameter-optimization} for the Student's t test-beds show that the condition imposed by Theorem~\ref{th:chacon-serrano-kernel} on $\chaconserranokernelparam$ is far stronger than necessary in practice.
Even though $\chaconserranokernelparam$ grows with $\studenttdegreesoffreedom$, we do not have $\chaconserranokernelparam > \studenttdegreesoffreedom$.
Indeed, Theorem~\ref{th:chacon-serrano-kernel} establishes a very stringent hypothesis on the kernel so that the sample mean is outperformed regardless of the bandwidth, providing that $\bandwidth$ is sufficiently large.
In our optimization scenario, however, we neither rely on a fixed kernel nor need an unbounded subset of bandwidths to outperform the sample mean.

\section{Discussion}

\label{sec:discussion}

This paper provides theoretical and practical results about estimating the center of symmetry from a modal point of view, covering a clear \textit{gap} between the nonparametrics and robust statistics communities.
Specifically, we adopt a purely nonparametric approach to a classic theme in robust statistics such as location estimation, assuming only symmetry and unimodality.
Along the way, interesting connections with \glspl*{kde} are enabled.

We have studied the efficiency of \glspl*{kme} in terms of the bandwidth $\bandwidth$, formalizing results without any parametric assumption, contrary to robust statistics.
On the one hand, a bandwidth too small dangerously increases the variance, supporting the \textit{empirical} caveats of~\citet{Huber2009} against wrong ``scaling''.
Conversely, increasingly large bandwidths make the \gls*{kme} behave like the sample mean.
However, it turns out that the \gls*{kme} can asymptotically outperform the sample mean with heavy-tailed, non-Gaussian data.
Namely, considering regularly varying \glspl*{pdf}, some kernels $\kernel$ allow finding an optimal bandwidth that minimizes the variance of the asymptotic distribution beyond that of the sample mean M-estimator.
In that sense, we also provide a novel parametric family of kernels $\chaconserranokernel$ whose parameter $\chaconserranokernelparam$, connected to the regular variation index $\regularvariationindex$, can be tuned to achieve that goal.

The theoretical efficiency calculations in \figurename~\ref{fig:kme-variance} demonstrate that, contrary to the widespread \textit{belief} shared by~\citet[p. 99]{Huber2009}, the shape of the kernel $\chaconserranokernel$, i.e., correctly choosing $\chaconserranokernelparam$, might be critical with specific data, to the point of rendering useless any optimization of $\bandwidth$ if the kernel is not suitable.
These considerations lead to a natural two-step \gls*{kme} procedure consisting of an \gls*{irw} run preceded by a double joint optimization of the shape parameter $\chaconserranokernelparam$ and the bandwidth $\bandwidth$, based on a \textit{plug-in} estimate of the \gls*{kme} variance.
Technical details on obtaining an estimate $\estimateofuntranslatedpdf$ and performing the two-dimensional optimization were also provided.

In addition to the above theoretical guarantees, the simulation study shed some light on the actual performance of our proposal compared to the sample mean and other more competitive methods.
The results are very favorable for our proposal.
The sample mean and median were easily outperformed whenever possible.
Then, two more sophisticated methods, the trimmed and winsorized means, though having very few weak points, were outperformed more often than not by our \gls*{kme},
especially with heavy-tailed \glspl*{pdf}.
Finally, the two redescending M-estimators, also \glspl*{kme}, were not nearly as versatile as our proposal, showing that tuning both the scale and shape parameters makes the difference.

The simulation and case studies revealed the compelling flexibility of the kernel family $\chaconserranokernel$.
For the more Gaussian-like \glspl*{pdf} and data, values $\chaconserranokernelparam > 1$ were chosen, producing smooth, nearly flat weight functions.
By contrast, for most test-beds, including all the regularly varying \glspl*{pdf}, values $\chaconserranokernelparam < 1$ were \textit{preferred}, producing a \textit{steep} slope near zero for the kernel derivative that resembles the $\mestimatorlossderivative$ of the sample median.
The latter shape, pervasive in our results, is \textit{rare} among usual redescending $\mestimatorlossderivative$ functions such as those in \figurename~\ref{fig:redescending-m-estimators}, designed with the typical robust statistics assumption of Gaussianity in mind.
However, in a fully nonparametric context, we see that those prefixed designs need not be optimal anymore.
All in all, the novel kernel family allows the \gls*{kme} to display a \textit{mean-like} or a \textit{median-like} behavior as needed.
These pieces of evidence deserve further investigation, reopening the quest for flexible redescending M-estimators.

Our results by no means invalidate those of robust statistics.
The latter framework is theoretically appealing and well-founded for many applications, where a default parametric model, typically Gaussian, can be assumed.
However, the robust statistics community has virtually ignored the relationship between redescending M-estimators and \glspl*{kme}.
On the other hand, the nonparametrics community has bypassed the symmetric, unimodal case so far, which is an interesting situation where asymptotically small bandwidths are paradoxically counterproductive.
By focusing on efficiency, beyond the robustness gained by employing a kernel with compact support (always welcomed), our research demonstrates that \glspl*{kme} provide an effective alternative as estimators of the center of symmetry.

\section*{Acknowledgements}

The research of the first author has been supported by the MICINN grant PID2021-124051NB-I00.
The second author would like to thank Professor Amparo Ba\'{i}llo Moreno for her advice as a doctoral counselor at the Autonomous University of Madrid.

\appendix

\section{Proofs}

\label{sec:proofs}

This appendix provides proof of the results in Section~\ref{sec:theoretical-results}.

\begin{proof}[Proof of Theorem~\ref{th:symmetric-unimodal-smoothed-pdf}]
    Let us first check the symmetry of $\smoothedpdf$ about $\symmetrycenter$.
    Using the symmetry of $\pdf$ and $\kernel$, we can easily see that if $\samplerandomvar \followsdistr \pdf$,
    \begin{equation*}
        \begin{split}
            \untranslatedsmoothedpdf
            &=
            \expectedvalueof{\kernelh(\symmetrycenter + \x - \samplerandomvar)}
            =
            \expectedvalueof{\kernelh(\x + \samplerandomvar - \symmetrycenter)}
            \\
            &=
            \expectedvalueof{\kernelh(\symmetrycenter - \x - \samplerandomvar)}
            =
            \smoothedpdf(\symmetrycenter - \x)
            \,.
        \end{split}
    \end{equation*}

    Now, to see the unimodality of $\smoothedpdf$, we will check that $\lambdafunction{\x}{\untranslatedsmoothedpdf}$ is a strictly decreasing function of $\absvalof{\x}$ over its support by calculating its derivative.
    Without loss of generality, we will assume that $\x \geq 0$ since the symmetry proved above will automatically cover the case of a negative argument.
    Therefore, straightforward calculations yield
    \begin{equation}
        \label{eq:smoothed-pdf}
        \begin{aligned}[b]
            \untranslatedsmoothedpdf
             & =
            \intoverreals
            \kernelh(\symmetrycenter + \x - \y)
            \pdf(\y)
            \ \dy
            =
            \intoverreals
            \kernelh(\x - \y)
            \pdf(\symmetrycenter + \y)
            \ \dy
            \\
             & =
            \intoverreals
            \kernelh(\x - \y)
            \untranslatedpdf(\y)
            \ \dy
            \,,
        \end{aligned}
    \end{equation}
    and, subsequently, since $\derivativeofkernelh$ being bounded allows differentiation under the integral,
    \begin{equation}
        \label{eq:smoothed-pdf-derivative}
        \begin{aligned}[b]
            \derivativesmoothedpdf
             & =
            \intoverreals
            \derivativeofkernelh(\x - \y)
            \untranslatedpdf(\y)
            \ \dy
            =
            \intoverreals
            \untranslatedpdf(\x - \y)
            \derivativeofkernelh(\y)
            \ \dy
            \\
             & =
            \intoverpositivereals
            [\untranslatedpdf(\x - \y) - \untranslatedpdf(\x + \y)]
            \derivativeofkernelh(\y)
            \ \dy
            \,.
        \end{aligned}
    \end{equation}
    If $\untranslatedpdf$ and $\kernelh$ have support $[-\pdfsupportedge, \pdfsupportedge]$ and $[-\kernelsupportedge, \kernelsupportedge]$, respectively, for $\pdfsupportedge, \kernelsupportedge \in (0, \infty]$, one can easily see from~\eqref{eq:smoothed-pdf} that the support of $\lambdafunction{\x}{\untranslatedsmoothedpdf}$ is $[-\pdfsupportedge - \kernelsupportedge, \pdfsupportedge + \kernelsupportedge]$.
    Hence, we shall check that~\eqref{eq:smoothed-pdf-derivative} is zero for $\x = 0$ and negative for $\x \in (0, \pdfsupportedge + \kernelsupportedge)$.

    When $\x = 0$, since $\untranslatedpdf$ is symmetric about zero, the derivative is zero, and thus the maximum is reached.
    For $\x > 0$, we will split the proof into two separate cases: $\pdfsupportedge = \infty$, on the one hand, and $\pdfsupportedge < \infty$, on the other.

    First, let us assume that $\pdfsupportedge = \infty$.
    If $\x > 0$, then $x$ and $y$ in~\eqref{eq:smoothed-pdf-derivative} are positive, and we have $\absvalof{\x - \y} < x + y$.
    Hence, $\untranslatedpdf(\x - \y) > \untranslatedpdf(\x + \y)$ in~\eqref{eq:smoothed-pdf-derivative} for all $\y > 0$ since both arguments lie in the support of $\untranslatedpdf$.
    Finally, the net sign of~\eqref{eq:smoothed-pdf-derivative} is negative because $\derivativeofkernelh$ is negative over $(0, \kernelsupportedge)$.

    Secondly, let us assume that $\pdfsupportedge < \infty$.
    The conclusion easily follows if $\x < \pdfsupportedge$, similarly to the case $\pdfsupportedge = \infty$ above, so we shall suppose $\x \geq \pdfsupportedge$.
    In such case, we have $\untranslatedpdf(\x + \y) = 0$ in~\eqref{eq:smoothed-pdf-derivative} for all $\y > 0$, and
    \begin{equation*}
        \begin{split}
            \derivativesmoothedpdf
            & =
            \intoverpositivereals
            \untranslatedpdf(\x - \y)
            \derivativeofkernelh(\y)
            \ \dy
            =
            \int_{-\infty}^{\x}
            \derivativeofkernelh(\x - \y)
            \untranslatedpdf(\y)
            \ \dy
            \\
            &=
            \int_{\max\{\x - \kernelsupportedge, -\pdfsupportedge\}}^{\pdfsupportedge}
            \derivativeofkernelh(\x - \y)
            \untranslatedpdf(\y)
            \ \dy
            \,,
        \end{split}
    \end{equation*}
    where the last equality is obtained by intersecting the supports of $\lambdafunction{\y}{\derivativeofkernelh(\x - \y)}$ and $\untranslatedpdf$ with $(-\infty, \x)$.
    Finally, the result follows after observing that $
        \lambdafunction{\y}
        {\derivativeofkernelh(\x - \y) \untranslatedpdf(\y)}
    $
    is negative over $
        (\max\{\x - \kernelsupportedge, -\pdfsupportedge\}, \pdfsupportedge)
    $,
    while the latter interval is not empty because $\x < \pdfsupportedge + \kernelsupportedge$.
\end{proof}

\begin{proof}[Proof of Theorem~\ref{th:symmetric-kme}]
    Let $\randomvarx \samedistribution \randomvary$ denote equality in distribution between two random variables $\randomvarx$ and $\randomvary$.
    We have to verify that $\centeredkme \samedistribution \symmetrycenter - \samplemode$, which is equivalent to $\samplemode \samedistribution 2\symmetrycenter - \samplemode$.
    Note that the right-hand side of the last equality in distribution is the \gls*{kme} reflected about $\symmetrycenter$, which is equal to the \gls*{kme} corresponding to the \gls*{kde} based on the reflected sample $(2\symmetrycenter - \samplerandomvarindex{1}, \dots, 2\symmetrycenter - \samplerandomvarindex{\samplesize})$.
    Since (i) the latter \gls*{iid} sample follows the same distribution as the original one, (ii) the kernel $\kernel$ is symmetric, and (iii) ties among candidate modes are broken uniformly at random in a symmetric way, we necessarily have that both \glspl*{kme} have the same distribution.
    Finally, the unbiasedness of $\samplemode$ is derived by taking expectations at both sides of the equality in distribution condition of symmetry.
\end{proof}

\begin{proof}[Proof of Theorem~\ref{th:sample-mode-asymptotic-normality}]
    Most of the proof is reduced to applying~\citet[Theorem 10.7]{Maronna2019}.
    Indeed,~\eqref{eq:sample-mode-asymptotic-variance} corresponds to Equation (2.24) in that same reference, with $\mestimatorlossderivative(\x) = -\derivativeofkernelh(\x)$.
    There only remains to check that the numerator and denominator in~\eqref{eq:sample-mode-asymptotic-variance} are finite, and that the latter is non-null.
    Since $\derivativeofkernelh$ and $\secondderivativeof{\kernelh}$ are bounded, the numerator and denominator are finite.
    Also due to the boundedness of $\secondderivativeofkernelh$, differentiation can go under the integral sign in the denominator as in~\citet[Theorem 10.7]{Maronna2019}.

    Finally, let us see that the denominator is not zero.
    Let $\inflectionpoint$ be the positive inflection point of $\kernelh$.
    Note that
    \begin{equation*}
        \begin{split}
            \int_0^{\inflectionpoint}
            \absvalof{\secondderivativeof{\kernelh}(\x)}
            \untranslatedpdf(\x)
            \ \dx
            &>
            \untranslatedpdf(\inflectionpoint)
            \int_0^{\inflectionpoint}
            \absvalof{\secondderivativeof{\kernelh}(\x)}
            \ \dx
            =
            \untranslatedpdf(\inflectionpoint)
            \int_{\inflectionpoint}^{\infty}
            \absvalof{\secondderivativeof{\kernelh}(\x)}
            \ \dx
            \\
            &\geq
            \int_{\inflectionpoint}^{\infty}
            \absvalof{\secondderivativeof{\kernelh}(\x)}
            \untranslatedpdf(\x)
            \ \dx
            \,,
        \end{split}
    \end{equation*}
    where we have used that $\untranslatedpdf$ is decreasing for positive $\x$ and that, since $\kernelh$ is bell-shaped,
    $
        \int_0^{\inflectionpoint}
        \absvalof{\secondderivativeof{\kernelh}(\x)}
        \ \dx
        =
        \int_{\inflectionpoint}^{\infty}
        \absvalof{\secondderivativeof{\kernelh}(\x)}
        \ \dx
        >
        0
    $.
    Therefore, we finally have
    \begin{equation*}
        \begin{split}
            \intoverreals
            \secondderivativeof{\kernelh}(\symmetrycenter - \x)
            \pdfatx
            \ \dx
            &=
            \intoverreals
            \secondderivativeof{\kernelh}(\x)
            \untranslatedpdf(\x)
            \ \dx
            =
            2
            \intoverpositivereals
            \secondderivativeof{\kernelh}(\x)
            \untranslatedpdf(\x)
            \ \dx
            \\
            &=
            2
            \left(
            \int_{\inflectionpoint}^{\infty}
            \absvalof{\secondderivativeof{\kernelh}(\x)}
            \untranslatedpdf(\x)
            \ \dx
            -
            \int_0^{\inflectionpoint}
            \absvalof{\secondderivativeof{\kernelh}(\x)}
            \untranslatedpdf(\x)
            \ \dx
            \right)
            <
            0
            \,.
        \end{split}
    \end{equation*}
\end{proof}

\begin{proof}[Proof of Theorem~\ref{th:limit-variance-h}]
    To prove part~\ref{it:limit-variance-h-tends-to-zero}, we note that, given the smoothness of $\pdf$, and the fact that $\derivativeofkernel$ and $\derivativeofpdf$ are bounded, we have $\convolutionvariancedenominator = \convolvedpdfsecondderivative$~\citep[Exercise 2.25]{Wand1995}.
    Then, $\limhzero \convolvedpdfsecondderivative = \curvatureatsymmetrycenter$ since $\kernelh$ is an approximation to the identity for the convolution operation~\citep[Theorem 1A]{Parzen1962}, given the remaining assumptions on $\pdf$ and $\kernel$.
    Therefore, using that $\pdfatsymmetrycenter$ and $\curvatureatsymmetrycenter$ are finite, and $\pdf$ is bounded and continuous, calculations yield
    \begin{equation*}
        \begin{split}
            \limhzero
            \squaredsmoothedstddevmode
            &=
            \abscurvatureatsymmetrycenter^{-2}
            \limhzero
            \convolutionvariancenumerator
            \\
            &=
            \abscurvatureatsymmetrycenter^{-2}
            \limhzero
            \intoverreals \pdfatmuminusx \derivativeofkernelh(\x)^2 \ \dx
            \\
            &=
            \abscurvatureatsymmetrycenter^{-2}
            \limhzero
            \bandwidth^{-3}
            \intoverreals \pdfatmuminushx \derivativeofkernel(\x)^2 \ \dx
            \\
            &=
            \squarestddevmodeexpanded
            \limhzero
            \bandwidth^{-3}
            =
            \infty
            \,,
        \end{split}
    \end{equation*}
    regardless of $\curvatureatsymmetrycenter$ being zero or not.
    Finally, if $\curvatureatsymmetrycenter < 0$, nearly identical steps show that $\limhzero \bandwidth^3 \squaredsmoothedstddevmode / \squarestddevmode = 1$.

    To prove part~\ref{it:limit-variance-h-tends-to-infty}, first note that
    \begin{equation*}
        \bandwidth \derivativeofkernel(\inversebandwidth \x)
        =
        \bandwidth
        [\derivativeofkernel(\inversebandwidth \x) - \derivativeofkernel(0)]
        =
        \bandwidth \int_0^{\x/\bandwidth} \secondderivativeofkernel(\y) \ \dy
        =
        \int_0^{\x} \secondderivativeofkernel(\inversebandwidth\y) \ \dy
        \,.
    \end{equation*}
    From last equation, since $\secondderivativeofkernel$ is bounded, we get $
        \limhinf
        \bandwidth \derivativeofkernel(\inversebandwidth \x)
        =
        \x
        \secondderivativeofkernel(0)
    $.
    Similarly, there exists some $\constant > 0$ such that $\smallabsvalof{\bandwidth \derivativeofkernel(\inversebandwidth \x)} \leq \constant \smallabsvalof{\x}$ for all $\x \in \reals$ and all $\bandwidth > 0$.
    Then, calculations using the dominated convergence theorem yield
    \begin{equation*}
        \begin{split}
            \limhinf
            \squaredsmoothedstddevmode
            &=
            \limhinf
            \frac
            {
                \bandwidth^{-6}
                \intoverreals
                \pdfatmuminusx
                [\bandwidth \derivativeofkernel(\inversebandwidth \x)]^2
                \ \dx
            }
            {
                \bandwidth^{-6}
                \left(
                \intoverreals
                \pdfatmuminusx
                \secondderivativeofkernel(\inversebandwidth \x)
                \ \dx
                \right)^2
            }
            \\
            &=
            \frac
            {
                \secondderivativeofkernel(0)^2
                \intoverreals
                \x^2 \pdfatmuminusx
                \ \dx
            }
            {
                \left(
                \secondderivativeofkernel(0)
                \intoverreals
                \pdfatmuminusx
                \ \dx
                \right)^2
            }
            \\
            &=
            \frac
            {
                \intoverreals
                (\centeredx)^2 \pdfatx
                \ \dx
            }
            {
                \left(
                \intoverreals
                \pdfatx
                \ \dx
                \right)^2
            }
            =
            \squarestddev
            \,.
        \end{split}
    \end{equation*}
\end{proof}

\begin{proof}[Proof of Theorem~\ref{th:regular-variation}]
    Without loss of generality, let us assume that $\absvalof{\secondderivativeofkernel(0)} = 1$.
    Otherwise, simply consider the normalized kernel $\x \mapsto \kernelatx / \absvalof{\secondderivativeofkernel(0)}$.
    Defining
    \begin{equation*}
        \asymptoticsmoothedvariance
        =
        \bandwidth^2 \intoverreals \pdfatmuminusx \derivativeofkernel(\inversebandwidth \x)^2 \ \dx
        \,,
    \end{equation*}
    we shall prove that, on the one hand, $\limhinf \bandwidth^{-(\regularvariationindexplusthree)} (\asymptoticsmoothedvariance - \squarestddev) = \negativelimit$, and, on the other hand, $\limhinf \bandwidth^{-(\regularvariationindexplusthree)} (\squaredsmoothedstddevmode - \asymptoticsmoothedvariance) = 0$.

    To check the first limit, we note that
    \begin{equation}
        \label{eq:asymptotic-variance-numerator}
        \begin{split}
            \asymptoticsmoothedvariance - \squarestddev
            &=
            \squarebandwidth
            \intoverreals
            \pdfatmuminusx
            [\derivativeofkernel(\inversebandwidth \x)^2 - (\inversebandwidth \x)^2]
            \
            \dx
            \\
            &=
            2\bandwidth^3
            \intoverpositivereals
            \pdfatmuplushx
            [\derivativeofkernel(\x)^2 - \x^2]
            \
            \dx
            \\
            &=
            2\bandwidth^3
            \pdfatmuplush
            \intoverpositivereals
            \frac{\pdfatmuplushx}{\pdfatmuplush}
            [\derivativeofkernel(\x)^2 - \x^2]
            \
            \dx
            \\
            &=
            2\bandwidth^{\regularvariationindexplusthree}
            \slowlyvaryingcomponent(\bandwidth)
            \intoverpositivereals
            \xtoregularvariationindex
            \frac{\slowlyvaryingcomponent(\bandwidth \x)}{\slowlyvaryingcomponent(\bandwidth)}
            [\derivativeofkernel(\x)^2 - \x^2]
            \
            \dx
            \\
            &\asymptoticallyequivalent
            2\bandwidth^{\regularvariationindexplusthree}
            \slowlyvaryingcomponent(\bandwidth)
            \intoverpositivereals
            \xtoregularvariationindex
            [\derivativeofkernel(\x)^2 - \x^2]
            \
            \dx
            \,,
        \end{split}
    \end{equation}
    where the boundedness of $\slowlyvaryingcomponent$ is used to bring the limit inside the integral.

    To check the second limit, let us first define
    \begin{equation*}
        \smoothedvariancedenominator
        =
        \intoverreals \pdfatmuminusx \secondderivativeofkernel(\inversebandwidth \x) \ \dx
        \,.
    \end{equation*}
    Then, we note that
    \begin{equation*}
        \squaredsmoothedstddevmode - \asymptoticsmoothedvariance
        =
        \squaredsmoothedstddevmode
        (1 + \smoothedvariancedenominator)
        (1 - \smoothedvariancedenominator)
        \asymptoticallyequivalent
        2\squarestddev
        (1 + \smoothedvariancedenominator)
        \,.
    \end{equation*}
    Finally, similarly to~\eqref{eq:asymptotic-variance-numerator},
    \begin{equation*}
        \begin{split}
            1 + \smoothedvariancedenominator
            &=
            \intoverreals
            \pdfatmuminusx
            [1 + \secondderivativeofkernel(\inversebandwidth \x)]
            \ \dx
            \\
            &=
            2\bandwidth
            \intoverpositivereals
            \pdfatmuplushx
            [1 + \secondderivativeofkernel(\x)]
            \
            \dx
            \\
            &=
            2\bandwidth
            \pdfatmuplush
            \intoverpositivereals
            \frac{\pdfatmuplushx}{\pdfatmuplush}
            [1 + \secondderivativeofkernel(\x)]
            \
            \dx
            \\
            &\asymptoticallyequivalent
            2\bandwidth^{\regularvariationindex + 1}
            \slowlyvaryingcomponent(\bandwidth)
            \intoverpositivereals
            \xtoregularvariationindex
            [1 + \secondderivativeofkernel(\x)]
            \
            \dx
            \,,
        \end{split}
    \end{equation*}
    and the result follows after $\limhinf \bandwidth^{-(\regularvariationindexplusthree)} \bandwidth^{\regularvariationindex + 1} = \limhinf \bandwidth^{-2} = 0$.
\end{proof}

\begin{proof}[Proof of Corollary~\ref{th:corollary-optimal-h}]
    It is a direct consequence of Theorem~\ref{th:limit-variance-h} and Theorem~\ref{th:regular-variation}, with analogous reasoning to the proof of~\citet[Theorem 1]{Chacon2007}.
    The result follows from the following facts: (i) $\variancefunction$ is continuous, (ii) $\limhzero \variancefunction(\bandwidth) = \infty$, (iii) $\limhinf \variancefunction(\bandwidth) = \squarestddev$, and (iv) $\variancefunction(\bandwidth) < \squarestddev$ for all sufficiently large $\bandwidth$.
\end{proof}

The proof of Theorem~\ref{th:chacon-serrano-kernel} relies on the following lemma.

\begin{lemma}
    \label{th:lemma-integrability}
    Let $\functiondef{\genericfunction}{[0, 1]}{\reals}$ be strictly increasing, continuous, and such that $\genericfunction(0) = 0$ and $\genericfunction(1) = \infty$.
    Given $\genericexponent < 0$, we have $\finiteintegral < \infty$ whenever $\lambdafunction{\x}{\integrablefunction}$ is integrable near zero.
\end{lemma}

\begin{proof}
    Using a Taylor expansion, for every $\x \in (0, 1)$, we have
    \begin{equation*}
        \exponentialminusone
        =
        -\genericfunction(\x)
        +
        \frac{e^{\taylorintermediatevalue}}{2}
        \genericfunction(\x)^2
        \,,
    \end{equation*}
    for some $\taylorintermediatevalue \in (-\genericfunction(\x), 0)$.
    Let $\cutpoint \in (0, 1)$ be the unique point such that $\genericfunction(\cutpoint) = 1$.
    Then, for every $\x \in (0, \cutpoint)$,
    \begin{equation*}
        \absexponentialminusone
        \leq
        \genericfunction(\x)
        +
        \frac{1}{2}
        \genericfunction(\x)^2
        \leq
        \frac{3}{2}
        \genericfunction(\x)
        \,.
    \end{equation*}
    Consequently, given $\genericexponent < 0$, it follows that
    \begin{equation*}
        \finiteintegral
        \leq
        \frac{3}{2}
        \int_0^{\cutpoint}
        \integrablefunction
        \ \dx
        +
        \int_{\cutpoint}^1
        \x^{\genericexponent}
        \absexponentialminusone
        \ \dx
        \,,
    \end{equation*}
    and the integrability of the left-hand side holds if $\integrablefunction$ is integrable near zero.
\end{proof}

\begin{proof}[Proof of Theorem~\ref{th:chacon-serrano-kernel}]
    Let us first check the hypotheses of Theorem~\ref{th:limit-variance-h}.
    The second derivative of the kernel is
    \begin{equation}
        \label{eq:chacon-serrano-kernel-second-derivative}
        \secondderivativeofchaconserranokernel(\x)
        \propto
        \left[
            \frac
            {\chaconserranokernelparam
                \absvalof{\x}^{\chaconserranokernelparam}}
            {(1 - \absvalof{\x}^{\chaconserranokernelparam})^2}
            -
            1
            \right]
        \bumplikefunction(\x)
        \,.
    \end{equation}
    From~\eqref{eq:chacon-serrano-kernel-second-derivative}, it is not difficult to see that $\chaconserranokernel$ is bell-shaped, having as positive inflection point $
        \inflectionpoint
        =
        [
        (1 + \chaconserranokernelparam / 2)
        -
        \sqrt{%
            (1 + \chaconserranokernelparam / 2)^2
            -
            1
        }
        ]^{1/\chaconserranokernelparam}
        \in
        (0, 1)
    $.
    Also, $\secondderivativeofchaconserranokernel(0) \propto -1 / e < 0$.
    The remaining assumptions in Theorem~\ref{th:limit-variance-h}, part~\ref{it:limit-variance-h-tends-to-infty}, follow from $\chaconserranokernel$ having compact support.

    Now, let us check the two integral conditions in Theorem~\ref{th:regular-variation}.
    Remark~\ref{th:remark-integral-conditions} applies since $\chaconserranokernel$ has compact support $[-1, 1]$.
    Therefore, for the first integral condition, we have
    \begin{equation}
        \label{eq:first-integral-condition-chacon-serrano-kernel}
        \xtoregularvariationindex
        \left[
            \frac
            {\derivativeofchaconserranokernel(\x)^2}
            {\secondderivativeofchaconserranokernel(0)^2}
            -
            \x^2
            \right]
        =
        \x^{\regularvariationindex + 2}
        \left[
            \exp
            \left(
            \frac
            {-2 \x^{\chaconserranokernelparam}}
            {1 - \x^{\chaconserranokernelparam}}
            \right)
            -
            1
            \right]
        \,,
    \end{equation}
    which is non-positive, and for the second one,
    \begin{equation}
        \label{eq:second-integral-condition-chacon-serrano-kernel}
        \xtoregularvariationindex
        \left[
            1
            +
            \frac
            {\secondderivativeofchaconserranokernel(\x)}
            {\smallabsvalof{\secondderivativeofchaconserranokernel(0)}}
            \right]
        =
        \frac
        {%
            \chaconserranokernelparam
            \x^{\regularvariationindex + \chaconserranokernelparam}
        }
        {(1 - \x^{\chaconserranokernelparam})^2}
        \exp
        \left(
        \frac
        {-\x^{\chaconserranokernelparam}}
        {1 - \x^{\chaconserranokernelparam}}
        \right)
        -
        \xtoregularvariationindex
        \left[
            \exp
            \left(
            \frac
            {-\x^{\chaconserranokernelparam}}
            {1 - \x^{\chaconserranokernelparam}}
            \right)
            -
            1
            \right]
        \,.
    \end{equation}
    We shall see that~\eqref{eq:first-integral-condition-chacon-serrano-kernel} and~\eqref{eq:second-integral-condition-chacon-serrano-kernel} are integrable over $[0, 1]$.
    In both cases, we can use Lemma~\ref{th:lemma-integrability}.
    For~\eqref{eq:first-integral-condition-chacon-serrano-kernel}, integrability holds if $\absvalof{\regularvariationindexplusthree} < \chaconserranokernelparam$; for~\eqref{eq:second-integral-condition-chacon-serrano-kernel}, if $\absvalof{\regularvariationindex + 1} < \chaconserranokernelparam$.
    Finally, since $\absvalof{\regularvariationindexplusthree} < \absvalof{\regularvariationindex + 1}$, the result follows.
\end{proof}

\ifnum\journal=0
    \renewcommand*{\mkbibcompletename}[1]{\textsc{#1}}
\fi

\printbibliography[filter=references, title={References}, sorting=nyt]

\ifnum\journal=0
    \renewcommand*{\mkbibcompletename}[1]{#1}
\fi

\end{document}